\documentclass[reprint,floatfix,superscriptaddress,longbibliography,twocolumn,notitlepage,nofootinbib]{revtex4-2}
\usepackage[T1]{fontenc}
\usepackage{tabularx}
\usepackage{booktabs}
\usepackage{pst-node}
\usepackage{graphicx}
\usepackage{amsmath}
\usepackage{amssymb}
\usepackage{xcolor}
\usepackage{xfrac}
\usepackage{caption}
\usepackage{subcaption}
\usepackage[colorlinks=true,linkcolor=blue, citecolor=blue]{hyperref}
\usepackage{cleveref}
\newcommand{\bra}[1]{\left\langle #1 \right|}
\newcommand{\ket}[1]{\left|#1\right\rangle}

\newcommand{\ketbra}[2]{\ket{#1}\!\!\bra{#2}}
\def\tr{\mathrm{tr}}
\def\deg{\mathrm{deg}}

\def\bolddot#1#2{\boldsymbol{#1}\cdot \boldsymbol{#2}}
\def\bbF{\mathbb{F}_2}
\newcommand\XQAOA[1]{\mathrm{XQAOA}_1^{\mathrm{#1}}}
\def\etal{\textit{et al.}}
\def\bbackslash{\backslash\!\!\backslash}

\newcommand*\circled[1]{\tikz[baseline=(char.base)]{
            \node[shape=circle,draw,inner sep=2pt] (char) {#1};}}

\usepackage{tikz}
\usepackage{pgfplots}
\usepackage{tikz-3dplot}

\tdplotsetmaincoords{60}{115}
\pgfplotsset{compat=newest}

\usepackage{mathtools}
\usepackage{float}
\usepackage{orcidlink}
\usepackage{amsthm}
\usepackage{algpseudocode}
\newtheorem{theorem}{Theorem}
\newtheorem{corollary}[theorem]{Corollary}
\newtheorem{lemma}[theorem]{Lemma}

\begin{document}
\title{An Expressive Ansatz for Low-Depth Quantum Approximate Optimisation}
\author{V Vijendran\,\orcidlink{0000-0003-3398-1821}}
\email{v.vijendran@anu.edu.au}
\affiliation{Centre for Quantum Computation and Communication Technologies (CQC2T), Department of Quantum Science, Research School of Physics and Engineering, Australian National University, Acton 2601, Australia}
\affiliation{A*STAR Quantum Innovation Centre (Q.InC), Institute of Materials Research and Engineering (IMRE), Agency for Science, Technology and Research (A*STAR), 2 Fusionopolis Way, Innovis \#08-03, Singapore 138634, Republic of Singapore}

\author{Aritra Das\,\orcidlink{0000-0001-7840-5292}}
\email{aritra.das@anu.edu.au}
\affiliation{Centre for Quantum Computation and Communication Technologies (CQC2T), Department of Quantum Science, Research School of Physics and Engineering, Australian National University, Acton 2601, Australia}

\author{Dax Enshan Koh\,\orcidlink{0000-0002-8968-591X}}
\email{dax\textunderscore koh@ihpc.a-star.edu.sg}
\affiliation{A*STAR Quantum Innovation Centre (Q.InC), Institute of High Performance Computing (IHPC), Agency for Science, Technology and Research (A*STAR), 1 Fusionopolis Way, \#16-16 Connexis, Singapore 138632, Republic of Singapore}

\author{Syed M Assad\,\orcidlink{0000-0002-5416-7098}}
\email{cqtsma@gmail.com}
\affiliation{Centre for Quantum Computation and Communication Technologies (CQC2T), Department of Quantum Science, Research School of Physics and Engineering, Australian National University, Acton 2601, Australia}
\affiliation{A*STAR Quantum Innovation Centre (Q.InC), Institute of Materials Research and Engineering (IMRE), Agency for Science, Technology and Research (A*STAR),
2 Fusionopolis Way, Innovis \#08-03, Singapore 138634, Republic of Singapore}

\author{Ping Koy Lam\,\orcidlink{0000-0002-4421-601X}}
\email{pingkoy@imre.a-star.edu.sg}
\affiliation{Centre for Quantum Computation and Communication Technologies (CQC2T), Department of Quantum Science, Research School of Physics and Engineering, Australian National University, Acton 2601, Australia}
\affiliation{A*STAR Quantum Innovation Centre (Q.InC), Institute of Materials Research and Engineering (IMRE), Agency for Science, Technology and Research (A*STAR),
2 Fusionopolis Way, Innovis \#08-03, Singapore 138634, Republic of Singapore}

\begin{abstract}
The Quantum Approximate Optimisation Algorithm (QAOA) is a hybrid quantum-classical algorithm used to approximately solve combinatorial optimisation problems. It involves multiple iterations of a parameterised ansatz that consists of a problem and mixer Hamiltonian, with the parameters being classically optimised. While QAOA can be implemented on near-term quantum hardware, physical limitations such as gate noise, restricted qubit connectivity, and state-preparation-and-measurement (SPAM) errors can limit circuit depth and decrease performance. To address these limitations, this work introduces the eXpressive QAOA (XQAOA), an overparameterised variant of QAOA that assigns more classical parameters to the ansatz to improve its performance at low depths. XQAOA also introduces an additional Pauli-Y component in the mixer Hamiltonian, allowing the mixer to implement arbitrary unitary transformations on each qubit. To benchmark the performance of XQAOA at unit depth, we derive its closed-form expression for the MaxCut problem and compare it to QAOA, Multi-Angle QAOA (MA-QAOA) [Sci Rep 12, 6781 (2022)], a Classical-Relaxed algorithm, and the state-of-the-art Goemans-Williamson algorithm on a set of unweighted regular graphs with 128 and 256 nodes for degrees ranging from 3 to 10. Our results indicate that at unit depth, XQAOA has benign loss landscapes with local minima concentrated near the global optimum, allowing it to consistently outperform QAOA, MA-QAOA, and the Classical-Relaxed algorithm on all graph instances and the Goemans-Williamson algorithm on graph instances with degrees greater than 4. Small-scale simulations also reveal that unit-depth XQAOA invariably surpasses both QAOA and MA-QAOA on all tested depths up to five. Additionally, we find an infinite family of graphs for which XQAOA solves MaxCut exactly and analytically show that for some graphs in this family, special cases of XQAOA are capable of achieving a much larger approximation ratio than QAOA. Overall, XQAOA is a more viable choice for variational quantum optimisation on near-term quantum devices, offering competitive performance at low depths.
\end{abstract}

\maketitle

\section{Introduction} \label{introduction}

Full-fledged fault-tolerant quantum computers capable of executing quantum algorithms that can solve problems of interest are expected to involve at least millions of physical qubits, high-fidelity gate operations, and quantum error correction techniques~\cite{beverland2022assessing}. While the physical realisation of such devices is still a long way off, noisy intermediate-scale quantum (NISQ) devices capable of running quantum algorithms with limited circuit depth are becoming more widely available \cite{preskill2018quantum,lau2022nisq}. Particularly promising are the variational quantum algorithms (VQAs)~\cite{cerezo2021variational, RevModPhys.94.015004,tilly2022variational,peruzzo2014variational,farhi2014quantum} capable of potentially realising a quantum advantage on NISQ devices. Unlike traditional quantum algorithms like Shor's algorithm \cite{shor1999polynomial} that use specially designed quantum circuits to solve specific problems, VQAs use parameterised quantum circuits whose objective is to drive a quantum state close to the desired state that minimises a cost function by varying the gate parameters.

The Quantum Approximate Optimisation Algorithm (QAOA)~\cite{farhi2014quantum} is one such algorithm that can solve optimisation problems by encoding their solutions into the ground state of a quantum Hamiltonian and preparing a quantum state that approximates this ground state. QAOA involves a $p$-level quantum circuit described by a collection of $2p$ classical parameters to generate a quantum state. The classical parameters are fine-tuned to optimise the expectation of the cost for the generated quantum state. This quantum state can then be measured to obtain an approximate solution to the optimisation problem. Besides its ability to solve combinatorial optimisation problems, QAOA can be used to perform universal quantum computation~\cite{lloyd2018quantum, morales2020universality}. Moreover, even at its lowest level $p=1$, QAOA can efficiently generate probability distributions that likely cannot be generated efficiently by classical computers \cite{farhi2016quantum,dalzell2020how}.

Several variants of the original QAOA algorithm have been developed, each with different operators and initial states \cite{bartschi2020grover, hadfield2019quantum, wurtz2021classically, wang2020x, egger2021warm, tate2020bridging, golden2021threshold, sack2021quantum, PhysRevResearch.4.013141, chalupnik2022augmenting,golden2023quantum,lee2022depth,leontica2022quantum} or different objective functions for tuning the variational parameters \cite{barkoutsos2020improving,li2020quantum}. Depth-reduction techniques \cite{majumdar2021optimizing,majumdar2021depth} or methods like circuit cutting \cite{bechtold202investigating,peng2020simulating} that optimise QAOA circuits while taking into account quantum hardware limitations; as well as classical aspects such as hyper-parameter optimisation and exploitation of problem structure, have been studied as well \cite{herrman2021impact, shaydulin2021classical, jain2022graph, streif2020training, akshay2021parameter, wurtz2022counterdiabaticity, egger2021warm, tate2020bridging, bravyi2020obstacles, bravyi2022hybrid}. However, one key drawback of realistic QAOA implementations is the need for deep quantum circuits with many qubits \cite{guerreschi2019qaoa, herrman2021lower, wurtz2021fixed, akshay2020reachability, wurtz2021maxcut, farhi2020quantum}. This poses a hurdle since NISQ devices are significantly limited due to their restricted qubit connectivity, inadequate qubit control, limited coherence times, and absence of quantum error correction, causing noise to grow with circuit depth and eventually affecting the fidelity of the resulting quantum state \cite{xue2021effects, wang2021noise, marshall2020characterizing, alam2019analysis, alam2020design, streif2021quantum, anschuetz2022beyond,stilck2021limitations,weidinger2023error,shaydulin2021error}.

There are several approaches that have been proposed to improve the performance of low-depth QAOA by adding new parameters to the ansatz \cite{chalupnik2022augmenting,herrman2022multi, PhysRevA.104.062428, yu2022quantum, PhysRevResearch.4.033029, tate2022warm}. These approaches include Multi-Angle QAOA (MA-QAOA)~\cite{herrman2022multi}, which increases the number of classical parameters added in each layer for more precise control of the optimisation process; Free-Axis Mixer QAOA (FAM-QAOA)~\cite{PhysRevA.104.062428}, which includes additional variational parameters in the mixer Hamiltonian that allow for rotation about an axis in the XY plane; QAOA with Adaptive Bias Fields (AB-QAOA)~\cite{yu2022quantum}, which adds a Pauli-Z component to the mixer Hamiltonian; Adaptive Derivative Assembled Problem Tailored QAOA (ADAPT-QAOA)~\cite{PhysRevResearch.4.033029}, which grows the ansatz iteratively using a gradient criterion; and QAOA+~\cite{chalupnik2022augmenting}, which augments the traditional QAOA ansatz with an additional multi-parameter layer that is independent of the specific problem being solved. Despite these improvements, there remains an imperative for problem-inspired quantum ansatzes with minimal computational overhead, which are not only expressive but also readily trainable allowing for greater flexibility in the optimisation process.

This paper presents a modified version of the QAOA called eXpressive QAOA (XQAOA). It shares the same inspiration behind the recently proposed Multi-Angle QAOA (MA-QAOA) approach~\cite{herrman2022multi} but goes beyond it by including an additional Pauli-Y component in the mixing Hamiltonian. This modification strategically overparameterises the quantum ansatz, facilitating the exploration of all relevant directions of the Hilbert space by allowing the mixer to effectively implement arbitrary unitary operations on each qubit with just a single iteration. As a result, XQAOA does not suffer from \textit{reachability deficits}~\cite{akshay2020reachability,akshay2021reachability}; with appropriately chosen angles, XQAOA can output any computational-basis state. To quantify the performance of the quantum algorithm, we apply it to the problem of maximum cut (MaxCut) on arbitrary graphs. We derive closed-form expressions for XQAOA, MA-QAOA, and QAOA at $p=1$ for the MaxCut problem and benchmark their performance against a naive Classical-Relaxed (CR) algorithm and the state-of-the-art Goemans-Williamson (GW)~\cite{goemans1995improved} algorithm on unweighted $D$-regular graphs---graphs where every node is connected to $D$ other nodes---with 128 and 256 nodes for $3 \leq D \leq 10$. The benchmark reveals that at $p = 1$, XQAOA outperforms MA-QAOA, QAOA, and the CR algorithm on all graph instances and the GW algorithm on graphs with $D > 4$; interestingly, the CR algorithm also outperforms QAOA and MA-QAOA on all graphs with QAOA matching MA-QAOA's performance for graphs with $D > 5$. We find that the exceptional performance of the XQAOA ansatz is attributed to the favourable characteristics of its benign loss landscape, which is notably free of barren plateaus and spurious local minima, with any remaining local minima being concentrated around the global optimum. Lastly, we show that for unweighted triangle-free graphs with edges of odd degrees, XQAOA can solve MaxCut exactly. Here, the edge degree $d(e)$ of an edge $e = \{u,v\}$ is defined as the number of neighbours of $e$, i.e., $ d(e) = |\mathcal{N}(u)\cup \mathcal{N}(v)|-2$, where $\mathcal{N}(w)$ is the set of all nodes connected to the node $w$.

The structure of the remainder of this paper is as follows: in \cref{prelim}, we review the necessary background material, where we explain the MaxCut problem and the challenges in finding its optimal solution (\cref{maxcut}), the traditional QAOA ansatz and its application to the MaxCut problem (\cref{qaoa}), and the MA-QAOA ansatz and its extension to MaxCut on arbitrary graphs in (\cref{ma_qaoa}). In \cref{xqaoa}, we introduce XQAOA and discuss its variants and other notable properties. In \cref{computational_results}, we present the results of our numerical simulations. In \cref{discussion}, we interpret and discuss our results, and in \cref{conclusion}, we provide some concluding remarks.

\begin{figure*}
  \includegraphics[width=\textwidth]{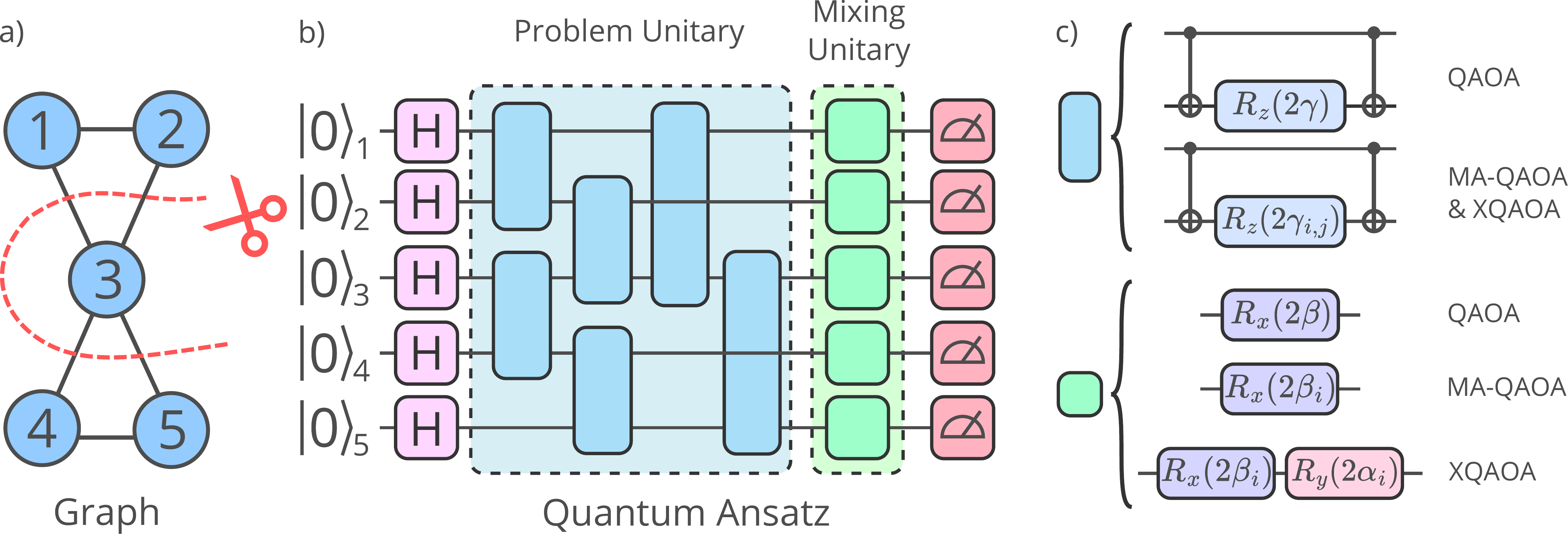}
  \captionsetup{justification=raggedright, singlelinecheck=false}
  \caption{a) A specific instance of a graph for which we want to identify a set of vertices that maximises the number of edges that are cut. b) A quantum circuit with a single iteration of a quantum ansatz applied to it. The quantum ansatz consists of a unitary operation specific to the problem being solved and a problem-independent mixing unitary. c) Decomposing the problem and mixing unitaries for QAOA, MA-QAOA, and XQAOA into CNOT and single-qubit rotation gates.}
  \label{fig:xqaoa}
\end{figure*}

\section{Preliminaries} \label{prelim}
\subsection{Maximum Cut (MaxCut)} \label{maxcut} 
Many real-world problems can be phrased as combinatorial optimisation problems~\cite{vazirani2001approximation}. Here, we lay emphasis on XQAOA's application to an archetypal problem known as MaxCut, which has numerous applications in computer science and operations research, including statistical physics and circuit layout design~\cite{barahona1988application}, analysis of social networks~\cite{mining}, data clustering~\cite{data_clustering}, semi-supervised learning~\cite{wang2013semi}, and more~\cite{DEZA1994191, DEZA1994217}. The (weighted) MaxCut problem is an optimisation problem in which we are given an undirected weighted graph and asked to partition its vertices into two disjoint sets $S$ and $\overline S$ such that the sum of the weights of the edges between the two sets is as large as possible. 

Formally, given an undirected graph $G = (V, E)$ and non-negative weights $w_{uv} = w_{vu}$ on the edges $\{u, v\} \in E$, the MaxCut problem is that of finding a set $S$ of vertices that maximises the weight of the edges in the cut $(S, \overline{S})$; that is, the weight of the edges with one endpoint in $S$ and the other in $\overline{S}$. The MaxCut problem can be formulated as a binary quadratic program of the form
\begin{equation}
\begin{array}{ll@{}ll}
\text{Maximise}  & \displaystyle\sum\limits_{\{u,v\} \in E} \frac{1}{2} w_{uv} \left(1 - y_uy_v \right) &\\[0.5cm]
\text{s.t}& y_u \in \{-1, 1 \} \quad \forall u \in V.
\end{array}
\label{maxcut_quadratic_program}
\end{equation}
The optimisation problem given by \cref{maxcut_quadratic_program} is $\mathsf{NP}$-hard\footnote{Historically, the $\mathsf{NP}$-hardness of MaxCut was one of the earliest results known in computational complexity theory: the decision version of the MaxCut problem was one of Karp's first $\mathsf{NP}$-complete problems~\cite{karp1972reducibility}. Here, a decision problem is a problem in which a yes-or-no answer is sought. A decision version of the MaxCut problem may be phrased as follows: given a graph $G$ and an integer $j$, determine if $G$ has a cut whose size is at least $j$.
}, which suggests that it is highly plausible that no efficient algorithm exists that can solve it.

However, there are approximation algorithms that can find good solutions in polynomial time for many instances of the problem. The GW algorithm holds the current record for an approximation ratio guarantee on generic graphs, achieving an approximation ratio of $r^* \approx 0.87856$ using semidefinite programming~\cite{goemans1995improved}. When confined to unweighted 3-regular graphs, this lower bound can be increased to $r^* \approx 0.9326$~\cite{halperin2004max}\footnote{This bound by Halperin \etal~is an improvement over an earlier result by Feige \etal, who found a smaller lower bound of $r^* \approx 0.924$ for unweighted 3-regular graphs~\cite{feige2002improved}.}.
Assuming the unique games conjecture \cite{khot2002power}\footnote{The unique games conjecture asserts that the problem of estimating the approximate value of a certain type of game, known as a unique game, has an $\mathsf{NP}$-Hard computational complexity.} and that $\mathsf P \neq \mathsf {NP}$,
this is the best possible approximation ratio for MaxCut~\cite{khot2007optimal, 5497893, khot2015unique} that polynomial-time classical algorithms can achieve.
Additionally, it has been proven that it is $\mathsf{NP}$-hard to approximate the MaxCut value with an approximation ratio that is better than $r^* \geq 16/17 \approx 0.94117$~\cite{haastad2001some, trevisan2000gadgets}.

\subsection{Quantum Approximate Optimisation Algorithm (QAOA)} \label{qaoa}
Combinatorial optimisation problems can be formulated using $n$ bits and $m$ clauses, where each clause represents a constraint on a subset of the bits that is satisfied for certain combinations of values for those bits but not for others. We consider the case when each clause $\mu$ is associated with a cost $c_\mu \in \mathbb R$. The objective function defined on $n$-bit strings is then given by the sum of the costs of the satisfied clauses: \begin{equation}
C(z) = \sum\limits^m_{\mu=1}\, C_\mu (z),
\label{FarhiEq1}
\end{equation}
where $z= z_1 z_2\cdots z_n \in \{0,1\}^n$ is an $n$-bit string and 
$C_{\mu} (z) = c_\mu$ if $z$ satisfies the clause $\mu$ and 0 otherwise. An approximate optimisation algorithm aims to find a string $z$ that achieves a desired approximation ratio $r^\star$, i.e., it seeks a string $z$ that satisfies
\begin{equation}
\frac{C(z)}{C_{\max }} \geq r^{*},
\end{equation}
where $C_{\max }=\max _{z} C(z)$. The QAOA algorithm consists of two operators (see \cref{fig:xqaoa}): the problem unitary and the mixing unitary, which are generated by the problem Hamiltonian and mixing Hamiltonian, respectively. The problem unitary is defined as the following unitary operator $U (C,\gamma)$ which depends on a real-valued angle $\gamma \in \mathbb R$:
\begin{equation}
U (C,\gamma) = e^{-i\gamma C} = \prod\limits^{m}_{\mu =1}\, e^{-i\gamma C_\mu }.
\end{equation}
The operators $C = \sum_z C(z) \ket z \!\!\bra z$ and $C_\mu = \sum_z C_\mu(z) \ket z \!\!\bra z$ are the diagonal operators whose entries are the objective function values. Next, the mixing unitary is defined as the $\beta$-dependent product of commuting one-qubit unitaries
\begin{equation}
U (B, \beta) = e^{-i\beta B} = \prod\limits^n_{\nu=1}\, e^{-i \beta  X_\nu},
\end{equation}
where $\beta \in [0, \pi)$ and $B$ is the sum of all single-qubit Pauli-X operators
\begin{equation}
B = \sum\limits^n_{\nu=1} X_\nu.
\end{equation}
For any positive integer $p \geq 1$, the QAOA algorithm generates an angle-dependent quantum state using $2p$ angles, $\boldsymbol{\gamma} = [\gamma_1, \gamma_2, \ldots, \gamma_p]$ and $\boldsymbol{\beta} = [\beta_1, \beta_2, \ldots, \beta_p]$, where the subscripts of $\gamma$ and $\beta$ indicate the iterate number of the quantum ansatz. The quantum state has the form
\begin{equation}
\ket{\boldsymbol{\gamma}, \boldsymbol{\beta}} = U (B, \beta_p) \, U (C, \gamma_p) \cdots U (B, \beta_1) \, U (C, \gamma_1)\, \ket{s},
\label{QAOA_Var_State}
\end{equation}
where $ \ket{s} $ denotes the uniform superposition over all $n$-bit strings
\begin{equation}
\ket{s} = \frac{1}{\sqrt{2^n}} \! \! \! \sum_{z\in \{0,1\}^n} \! \! \! \! \ket{z} .
\end{equation}
We then compute the expectation value of $C$ for the variational state described in \cref{QAOA_Var_State}
\begin{equation}
\langle C \rangle  = \bra{\boldsymbol{\gamma}, \boldsymbol{\beta}} C  \ket{\boldsymbol{\gamma}, \boldsymbol{\beta}},
\label{FarhiEq8}
\end{equation}
which is accomplished by repeated measurements of fresh copies of the quantum system in the computational basis. The optimal parameters $(\boldsymbol{\gamma}^*, \boldsymbol{\beta}^*)$ that maximise the expectation value $\langle C \rangle$ are found using a classical computer:
\begin{equation}
(\boldsymbol{\gamma}^*, \boldsymbol{\beta}^*)=\underset{\boldsymbol{\gamma}, \boldsymbol{\beta}}{\arg \max }\, \langle C \rangle .
\end{equation}
Typically, this is performed by estimating the parameters and then optimising them using simplex or gradient techniques. The approximation ratio $r^*$ is a relevant metric for assessing the performance of QAOA, where
\begin{equation}
r^*=\frac{\langle C \rangle}{C_{\max }}.
\end{equation}

We will focus on applying QAOA to the MaxCut problem for the rest of this paper. To this end, note that the optimisation problem in \cref{maxcut_quadratic_program} is equivalent to finding the maximum eigenvalue of the problem Hamiltonian $C$ for MaxCut:
\begin{align}
C = \! \! \! \! \sum_{\{u,v\} \in E} \! \! \! C_{uv},  \quad \mbox{with } C_{uv} = \frac{1}{2}w_{uv}\left (1-Z_u Z_v  \right) ,
\end{align}
where $Z_i$ denotes the Pauli-Z matrix acting on the $i$-th qubit.

Before proceeding with the rest of the section, let us make a few definitions that will be used throughout the paper. For $w \in V$, let $\mathcal{N}(w) = \{ x \in V: \{x,w\} \in E \}$ be the set of neighbours of $w$, i.e.~vertices which are adjacent to $w$. Then, for an edge $\{u, v\} \in E$,
we have that
\begin{itemize} \setlength\itemsep{0.75em}
    \item $e = \mathcal{N}(v) \backslash \{u\}$ is the set of vertices other than $u$ that are connected to $v$.
    \item $d = \mathcal{N}(u) \backslash \{v\}$ is the set of vertices other than $v$ that are connected to $u$.
    \item $F =\mathcal{N}(u) \cap \mathcal{N}(v)$ is the set of vertices that form a triangle with the edge $\{u,v\}$. In other words, $F$ is the set of vertices that are neighbours of both $u$ and $v$.
\end{itemize}
The following theorem can be used to compute the expectation value of the cost function for QAOA at $p = 1$ (QAOA$_1$) for MaxCut on arbitrary weighted graphs, thereby allowing us to assess the performance of QAOA$_1$.

\begin{theorem}
\label{qaoa_hadfield1}
Consider the $\mathrm{QAOA}_1$ state $\ket{\gamma,\beta}$ for MaxCut on an arbitrary weighted graph $G$. Then, the expectation value of $C$ in $\ket{\gamma,\beta}$ is $\bra{\gamma,\beta} C\ket{\gamma,\beta}= \sum_{\{u,v\}\in E} \langle C_{uv}\rangle$,
where
\begin{widetext}
\begin{align}
    \left\langle C_{u v}\right\rangle  &=\frac{w_{uv}}{2} +\frac{w_{uv}}{4}\Bigg[ \sin 4 \beta \sin\gamma_{uv}' \left(\prod_{w \in e}\cos\gamma_{wv}' + \prod_{w \in d}\cos \gamma_{uw}' \right)  \nonumber\\
    & \quad+ \sin^2 2 \beta \prod_{\substack{w \in e \\ w \notin F}}\cos\gamma_{wv}' \prod_{\substack{w \in d \\ w \notin F}}\cos \gamma_{uw}' \Bigg( \prod_{f \in F} \cos(\gamma_{uf}' + \gamma_{vf}') - \prod_{f \in F} \cos(\gamma_{uf}' - \gamma_{vf}') \Bigg)  \Bigg]
\label{qaoa_formula}
\end{align}
\end{widetext}
and $\gamma'_{ij} = \gamma w_{ij}$.
\end{theorem}
In \cref{MA_QAOA_Proof}, we give a proof of \cref{qaoa_hadfield1}, which we show follows as a straightforward corollary of our main theorem (\cref{XQAOA_Full_Thm}). By taking $w_{ij}=1$ if $\{i,j\}\in E$ and 0 otherwise, \cref{qaoa_formula} simplifies for unweighted graphs to:
\begin{align}
    \langle C_{uv} \rangle
    &= \frac 12 + \frac 14 \Big\{ \sin 4\beta \sin\gamma (\cos^{|e|} \gamma +\cos^{|d|} \gamma) \nonumber\\[0.5em]
    &\quad+\sin^2 2\beta \cos^{|e|+|d|-2|F|}\gamma (\cos^{|F|} 2 \gamma-1)  \Big\},
    \label{eq:qaoa_unweighted}
\end{align}
which has previously appeared as eq.~(14) of \cite{wang2018quantum}\footnote{One could also find similar analytical expressions for unweighted MaxCut in, for example, \cite[eq.~(5.10)]{hadfield2018quantum}. See also \cite{bravyi2020obstacles,hadfield2022analytical,ozaeta2022expectation}, which provide analytical expressions for more general cost functions.}.

From \cref{qaoa_hadfield1}, we see that at $p=1$, the expectation value  $\left\langle C_{u v}\right\rangle$ of any edge in a graph depends on only the nodes and edges adjacent to it. The overall expectation value for QAOA$_1$ can then be calculated by summing the expectation values over all edges in the graph. For an $n$-node graph, the right-hand side of \cref{qaoa_formula} can be computed in linear time $O(n)$. Since the total number of edges in any graph is at most $\binom{n}{2} = O(n^2)$, computing the expectation value of QAOA would take at most $O(n^3)$ time. However, to find an actual bit string that represents an approximate solution for an arbitrary graph, here we use the QAOA quantum circuit to generate a quantum state on which measurement is performed.

\subsection{Multi-Angle QAOA (MA-QAOA)} \label{ma_qaoa}
The Multi-Angle QAOA (MA-QAOA)~\cite{herrman2022multi} varies from the original QAOA in that it allows each summand of the problem and mixing Hamiltonians to have its own angle, as opposed to these Hamiltonians having a single angle each\footnote{Predating Herrman \etal~\cite{herrman2022multi} was earlier work by Farhi \etal~\cite{farhi2017quantum}, who first considered allowing for multiple angles in QAOA.}. In this modification for $p=1$ (called MA-QAOA$_1$), the problem and mixing unitaries are defined as
\begin{align}
U (\boldsymbol C, \boldsymbol{\gamma}) &= e^{-i \sum_\mu \gamma_\mu C_\mu} = \prod\limits^{m}_{\mu =1}\, e^{-i\gamma_\mu C_\mu }, \mbox{ and} \\[0.75cm]
U (\boldsymbol B, \boldsymbol{\beta}) &= e^{-i \sum_\nu \beta_\nu B_\nu} = \prod\limits^n_{\nu=1}\, e^{-i \beta_\nu  X_\nu},
\end{align}
respectively, where $\boldsymbol{C} = (C_\mu)_{\mu=1,\ldots,m}$ and $\boldsymbol{B} = (B_\nu)_{\nu=1,\ldots,n}$ denote collections of operators. Thus, MA-QAOA$_1$ generates an angle-dependent quantum state of the form
\begin{align}
    \ket{\boldsymbol{\gamma}, \boldsymbol{\beta}} &=  U (\boldsymbol B, \boldsymbol{\beta}) \, U (\boldsymbol C, \boldsymbol{\gamma})\, \ket{s} \nonumber\\[0.5em]
    &=  \prod\limits^n_{\nu=1}\, e^{-i \beta_\nu  X_\nu} \prod_{\mu =1}^m e^{-i \gamma_{\mu} C_{\mu}} \ket{s},
\label{maqaoa_angle_dependent_state}
\end{align}
where $\boldsymbol{\gamma} = [\gamma_1, \gamma_2, \dots, \gamma_m]$ and $\boldsymbol{\beta} = [\beta_1, \beta_2, \dots, \beta_n]$. The subscript in $\gamma_\mu$ refers to the $\mu$-th clause, and the subscript in $\beta_\nu$ refers to the $\nu$-th qubit. In the context of MaxCut, $\mu$ and $\nu$ index the edges and vertices, respectively, of the graph involved. The approximation ratio obtained using QAOA lower bounds that of MA-QAOA, and MA-QAOA's guarantee of convergence to the exact solution as $p \to \infty$ follows immediately from \cite[eq.~(10)]{farhi2014quantum} and from noting that MA-QAOA is a generalisation of QAOA.

Herrman \etal~\cite{herrman2022multi} provide an analytical formula for computing the performance of MA-QAOA$_1$ on MaxCut for unweighted triangle-free graphs. We generalise their result with the following theorem, where we present an analytical formula for the expectation value of the cost function for MA-QAOA$_1$ for MaxCut on arbitrary weighted graphs, allowing for the assessment of MA-QAOA$_1$'s performance on general graphs.
\begin{theorem}
\label{ma_qaoa_theorem}
Consider the $\mathrm{MA\text{-}QAOA}_1$ state for MaxCut on an arbitrary graph G. Then, the expectation value of $C$ in $\ket{\boldsymbol{\gamma},\boldsymbol{\beta}}$ is $\bra{\boldsymbol{\gamma},\boldsymbol{\beta}} C\ket{\boldsymbol{\gamma},\boldsymbol{\beta}}= \sum_{\{u,v\}\in E} \langle C_{uv}\rangle_{\mathrm{MA}}$, where
\begin{widetext}
\begin{align}
    \left\langle C_{u v}\right\rangle_{\mathrm{MA}}  &=\frac{w_{uv}}{2} +\frac{w_{uv}}{2}\Bigg[ \cos 2 \beta_u  \sin 2 \beta_v \sin\gamma_{uv}' \prod_{w \in e}\cos\gamma_{wv}'  +  \sin 2 \beta_u  \cos 2 \beta_v \sin \gamma_{uv}' \prod_{w \in d}\cos \gamma_{uw}' \nonumber\\[0.5em]
    & + \frac{1}{2} \sin 2 \beta_u \sin 2 \beta_v \prod_{\substack{w \in e \\ w \notin F}}\cos\gamma_{wv}' \prod_{\substack{w \in d \\ w \notin F}}\cos \gamma_{uw}' \Bigg( \prod_{f \in F} \cos(\gamma_{uf}' + \gamma_{vf}') - \prod_{f \in F} \cos(\gamma_{uf}' - \gamma_{vf}') \Bigg) \Bigg]
\label{ma_qaoa_full_exp}
\end{align}
\end{widetext}
and $\gamma_{jk}' = \gamma_{jk}w_{jk}$.
\end{theorem}
We present, in \cref{MA_QAOA_Proof}, a proof of \cref{ma_qaoa_theorem}, which again, is a straightforward corollary of our main theorem (\cref{XQAOA_Full_Thm}). Like \cref{qaoa_formula}, the expectation value $\left\langle C_{u v}\right\rangle_{\mathrm{MA}}$ for any edge in a graph depends on only its neighbouring nodes and edges, and the overall expectation value is the sum of the expectation values over all edges in the graph; hence, computing \cref{ma_qaoa_full_exp} for an arbitrary graph has a time complexity of $O(n^3)$. However, this time complexity has a larger constant prefactor compared to that of computing \cref{qaoa_formula}. While QAOA$_1$ involves only two hyperparameters regardless of the size of the problem, MA-QAOA$_1$ involves $|V| + |E| = O(n) + O(n^2)$ classical hyperparameters.

\section{Expressive QAOA (XQAOA)} \label{xqaoa}
The eXpressive QAOA builds on MA-QAOA by introducing an additional $\boldsymbol{\alpha}$-dependent unitary operator to the mixing Hamiltonian. Let us define the $\boldsymbol{\alpha}$-dependent operator to be the following product of commuting one-qubit operators:
\begin{equation}
U (\boldsymbol A, \boldsymbol{\alpha}) = e^{-i \sum_j \alpha_j A_j} = \prod\limits^n_{j=1}\, e^{-i \alpha_j  Y_j},
\end{equation}
where $\boldsymbol{A} = (A_i)_{i=1,\ldots,n}$, and $\alpha \in [0, \pi)$\footnote{Due to the $\alpha\rightarrow \alpha+\pi$ and $\beta\rightarrow \beta+\pi$ translational symmetries of the QAOA output state, one could without loss of generality assume that $\alpha$ and $\beta$ lie in the interval $[0, \pi)$. For the purposes of our simulations though, we do not place such an explicit restriction, since $\alpha$ and $\beta$ repeat in intervals of $\pi$ (in addition, for unweighted graphs, $\gamma$ repeats in intervals of $2\pi$) anyways. The data in \cref{xqaoa_angles} were adjusted to fit the ranges mentioned in this paper.}.

The mixing unitary is then given by the product of the $U (\boldsymbol B, \boldsymbol{\beta})$ and $U (\boldsymbol A, \boldsymbol{\alpha})$ unitary operators:
\begin{align}
\label{eq:ma-qaoa-mixer}
    U (\boldsymbol A, \boldsymbol{\alpha}) U (\boldsymbol B, \boldsymbol{\beta}) &= e^{-i \sum_j \alpha_j A_j} e^{-i \sum_j \beta_j B_j}\nonumber\\[0.25cm]
    &= \prod\limits^n_{j=1}\, e^{-i \alpha_j  Y_j} e^{-i \beta_j  X_j}.
\end{align}

Thus at $p=1$, XQAOA generates an angle-dependent quantum state of the form
\begin{equation}
\begin{split}
    \ket{\boldsymbol{\gamma}, \boldsymbol{\beta}, \boldsymbol{\alpha}} &=  U (\boldsymbol A, \boldsymbol{\alpha}) U (\boldsymbol B, \boldsymbol{\beta}) \, U (\boldsymbol C, \boldsymbol{\gamma})\, \ket{s} \\[0.25cm]
    &=  \prod\limits^n_{j=1}\, e^{-i \alpha_j  Y_j} e^{-i \beta_j  X_j} \prod_{\mu =1}^m e^{-i \gamma_{\mu} C_{\mu}} \ket{s},
\end{split}
\label{xqaoa_angle_dependent_state}
\end{equation}
where $\boldsymbol{\alpha} = [\alpha_1, \alpha_2, \dots, \alpha_n]$, $\boldsymbol{\beta} = [\beta_1, \beta_2, \dots, \beta_n]$, and $\boldsymbol{\gamma} = [\gamma_1, \gamma_2, \dots, \gamma_m]$. Similarly to the \cref{maqaoa_angle_dependent_state}, the subscript in $\gamma_i$ denotes the $i$-th clause, and the subscripts in $\alpha_i$ and $\beta_i$ refer to the $i$-th qubit, which in the context of MaxCut correspond to the edges and vertices, respectively, of the graph.

One motivation for introducing the XQAOA is that, unlike QAOA and MA-QAOA, the XY mixer\footnote{The XY mixer used in XQAOA differs from Wang et al.~\cite{wang2020x}'s approach. We utilise a single-qubit mixer to increase the range of Hilbert Space explored for binary combinatorial optimisation, while Wang et al.~employs a multi-qubit mixer in the Quantum Alternating Operator Ansatz to confine the search space to feasible solutions in integer-valued optimisation problems.} in \cref{eq:ma-qaoa-mixer} is the most general product (with respect to the $n$ registers in the circuit) unitary operator one could write for $p=1$ XQAOA, up to an unphysical global phase incurred when the system is measured immediately after the mixer unitary is applied. This makes XQAOA a natural generalisation of QAOA to consider, as one aims to maximise the expressiveness of the ansatz given the limitations on its depth, and also gives XQAOA the ability to output any computational-basis state given appropriate angles $\boldsymbol{\gamma}$, $\boldsymbol{\beta}$, and $\boldsymbol{\alpha}$. To see this, note that if we set the angles $\boldsymbol{\gamma} = \boldsymbol{\beta} = \mathbf{0}$ in \eqref{xqaoa_angle_dependent_state}, we are left with single-qubit Y-rotations on the $\ket +$ states. Choosing appropriate angles $\alpha_j$ on each qubit will bring $\ket +$ to $\ket 0$ or $\ket 1$. The same is true for when $\boldsymbol{\gamma} = \mathbf{0}$ and $\boldsymbol{\alpha} = \boldsymbol{\beta}$. Consequently, as we mentioned in \cref{introduction}, XQAOA is able to eschew any reachability deficits~\cite{akshay2020reachability,akshay2021reachability}.

\begin{table}[htpb]
\centering
\begin{tabularx}{\columnwidth}{>{\centering\arraybackslash}X >{\centering\arraybackslash}X}
\toprule
\textbf{Ansatz} & \textbf{No. of Parameters} \\
\midrule
MA-QAOA$_p$    & $(n+m)p$                      \\[0.3em]
XQAOA$^{\text{XY}}_p$  & $(2n+m)p$                     \\[0.3em]
XQAOA$^{\text{Y}}_p$   & $(n+m)p$                      \\[0.3em]
XQAOA$^{\text{X=Y}}_p$ & $(n+m)p$                      \\
\bottomrule
\end{tabularx}
\captionsetup{justification=raggedright, singlelinecheck=false}
\caption{Summary of the XQAOA Ansatz family and the
associated number of free parameters for $p$ iterations of the ansatz for MaxCut on graphs with $n$ vertices and $m$ edges.}
\label{XQAOA_ansatz_family}
\end{table}

From \cref{xqaoa_angle_dependent_state}, it is clear that several variations of the XQAOA ansatz can be generated by placing restrictions on the allowed angles of the mixing unitaries. The MA-QAOA is a special case of the XQAOA ansatz obtained by setting all $\alpha_i = 0$. Other configurations of the XQAOA ansatz worth noting are those with the XY Mixer, Y Mixer, and the X=Y Mixer, respectively. The XY Mixer is the most general mixer and uses individual angles $\alpha_i$, $\beta_i$ for each unitary in the mixing Hamiltonian. The Y Mixer consists of only Pauli-Y gates and is obtained by setting all $\beta_i$ to zero. The X=Y Mixer includes both Pauli-X and Pauli-Y gates but uses a single angle for both, with $\alpha_i$ equal to $\beta_i$.

As we summarise in \cref{XQAOA_ansatz_family}, the XQAOA ansatzes with the XY mixer, Y mixer, and X=Y mixer for $n$ qubits and $m$ clauses require the classical optimisation of $2n + m$, $n + m$, and $n + m$ angles, respectively. While the performances of these mixers are not known \emph{a priori}, the XY mixer is expected to have a higher computational overhead than the other two mixers due to the presence of an additional $n$ classical parameters. The X=Y mixer is expected to perform better than the Y mixer because it is able to trace a larger portion of the Bloch sphere due to its non-trivial trajectory, whereas the Y mixer is limited to the XZ plane.

In the remainder of the paper, we will use the superscript notation to indicate the specific variant of the XQAOA ansatz being used, i.e.~$\XQAOA{XY}$, $\XQAOA{X=Y}$, and $\XQAOA{Y}$ refer to $p=1$ XQAOA with the XY, X=Y, and Y mixers, respectively. The next theorem---the main theorem of this paper---allows us to calculate the expectation value of the cost function for $\XQAOA{XY}$ for MaxCut on arbitrary weighted graphs, which in turn allows us to evaluate the performance of XQAOA.
\begin{theorem}
Consider the $\mathrm{XQAOA}_1^{\mathrm{XY}}$ state $\ket{\boldsymbol{\gamma}, \boldsymbol{\beta}, \boldsymbol{\alpha}}$ for MaxCut on an arbitrary weighted graph $G$.  Then, the expectation value of $C$ in $\ket{\boldsymbol{\gamma},\boldsymbol{\beta}, \boldsymbol{\alpha}}$ is $\bra{\boldsymbol{\gamma},\boldsymbol{\beta}, \boldsymbol{\alpha}} C\ket{\boldsymbol{\gamma},\boldsymbol{\beta}, \boldsymbol{\alpha}}= \sum_{\{u,v\}\in E} \langle C_{uv}\rangle_{\mathrm{XY}}$, where
\begin{widetext}
\begin{align}
    \left\langle C_{u v}\right\rangle_{\mathrm{XY}}  &= \frac{w_{uv}}{2}+\frac{w_{uv}}{2}\Bigg[\cos 2 \alpha_u \cos 2 \alpha_v \sin \gamma_{uv}' \left(\cos 2 \beta_u \sin 2 \beta_v  \prod_{w \in e}\cos \gamma_{wv}' +  \sin 2 \beta_u  \cos 2 \beta_v  \prod_{w \in d}\cos \gamma_{uw}' \right)   \nonumber\\[0.5em]
    & - \frac{1}{2} \sin 2 \alpha_u \sin 2 \alpha_v \prod_{\substack{w \in e \\ w \notin F}}\cos\gamma_{wv}' \prod_{\substack{w \in d \\ w \notin F}}\cos \gamma_{uw}' \left( \prod_{f \in F} \cos(\gamma_{uf}' + \gamma_{vf}') + \prod_{f \in F} \cos(\gamma_{uf}' - \gamma_{vf}') \right) \nonumber\\[0.5em]
    &  + \frac{1}{2} \cos 2\alpha_u \sin 2 \beta_u \cos 2\alpha_v \sin 2 \beta_v \prod_{\substack{w \in e \\ w \notin F}}\cos\gamma_{wv}' \prod_{\substack{w \in d \\ w \notin F}}\cos \gamma_{uw}' \left( \prod_{f \in F} \cos(\gamma_{uf}' + \gamma_{vf}') - \prod_{f \in F} \cos(\gamma_{uf}' - \gamma_{vf}') \right) \Bigg] 
\label{xqaoa_full_exp}
\end{align}
\end{widetext}
and $\gamma_{jk}' = \gamma_{jk}w_{jk}$.
\label{XQAOA_Full_Thm}
\end{theorem}
We present a proof of \cref{XQAOA_Full_Thm} in \cref{xqaoa_proof1}. Like \cref{qaoa_formula} and \cref{ma_qaoa_full_exp}, the expectation value of any edge in a graph $\left\langle C_{u v}\right\rangle_{\mathrm{XY}}$ is determined by its neighbouring nodes and edges, and the overall expectation value is the sum of the expectation values of all edges in the graph. Calculating \cref{xqaoa_full_exp} for an arbitrary graph also has a time complexity of $O(n^3)$, but has a larger prefactor compared to both \cref{qaoa_formula} and \cref{ma_qaoa_full_exp}. While QAOA$_1$ requires only two parameters regardless of the problem size and MA-QAOA$_1$ requires $n + n^2$ parameters, $\XQAOA{XY}$ requires $2n + n^2$ parameters. In contrast, both $\XQAOA{Y}$ and $\XQAOA{X=Y}$ require $n + n^2$ parameters.

In the following corollary, we show that for unweighted graphs with edges of odd edge degrees, the $\XQAOA{Y}$ ansatz can solve MaxCut exactly. Here, the edge degree $d(e)$ of an edge $e = \{u, v\}\in E $ is defined as the number of neighbours of $e$, i.e., $ d(e) = |\mathcal N(u)\cup \mathcal N(v)|-2$.

\begin{corollary}
Consider an unweighted graph $G$ where the edge degree of every edge is odd. Then, when $\gamma = \pi$ and $\alpha = \frac{\pi}{4}$, the $\XQAOA{Y}$ state $\ket{\gamma, \alpha}$ provides the exact MaxCut solution for $G$, where $\ket{\gamma, \alpha}$ denotes the state in \cref{xqaoa_angle_dependent_state} where all $\gamma_i = \gamma$, $\beta_i = 0$, and $\alpha_i=\alpha$.
\label{XQAOA_Exact_Thm}
\end{corollary}
We give a proof of corollary~\ref{XQAOA_Exact_Thm} in \cref{xqaoa_proof2}. One consequence of corollary~\ref{XQAOA_Exact_Thm} is that it allows us to identify a graph instance for which we can analytically prove a separation between $\XQAOA{Y}$ and $\mathrm{QAOA}_1$. Our next corollary elucidates this result.
\begin{corollary}
    For the unweighted 5-vertex star graph $G$, $\XQAOA{Y}$ with optimal angles (say, from corollary~\ref{XQAOA_Exact_Thm}) computes the MaxCut of $G$ with an expected (and worst-case) approximation ratio of 1, whereas the expected approximation ratio of $\mathrm{QAOA}_1$ with optimal angles is only 0.75. 
\label{cor:proof_separation}
\end{corollary}
We give a proof of corollary~\ref{cor:proof_separation} in \cref{app:proof_separation}. While our result above pertains to the 5-vertex star graph (see~\cref{fig:4-star}), one can readily generalise this proof to any $t$-vertex star graph, where $t\geq 5$ is odd (here, the oddness criterion arises because it is only for odd-vertex star graphs that the edge degrees of the graph are all odd). The statement that QAOA$_1$ achieves an optimal expected approximation ratio of 3/4 for these graphs could be considered a finite-dimensional analogue of \cite[Section IV]{herrman2022multi}'s result that in the limit as the number of vertices tends to infinity, the performance of QAOA$_{1}$ approaches 0.75 for star graphs. In terms of the expected approximation ratio that can be achieved, this infinite class of graphs instantiates a clear advantage that XQAOA has over QAOA.

\begin{figure}
    \begin{tikzpicture}
    \node[circle,fill=black,scale=0.6] at (10:0mm) (center) {};
    \foreach \n in {1,...,4}{
        \node[circle,fill=black,scale=0.6] at ({\n*360/4}:1cm) (n\n) {};
        \draw (center)--(n\n);
    }
\end{tikzpicture}
\captionsetup{justification=raggedright, singlelinecheck=false}
\caption{Diagrammatic representation of the 5-vertex star graph $S_4$, which we use to show an advantage that XQAOA has over QAOA. More specifically, we show that while $\XQAOA{Y}$ can find the MaxCut of $S_4$ with an approximation ratio of 1, QAOA$_1$ can achieve an approximation ratio of at most 3/4.}
    \label{fig:4-star}
\end{figure}
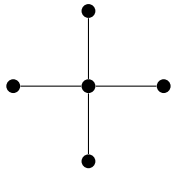

\section{Computational Results} \label{computational_results}
\begin{figure}[htbp]
\captionsetup{justification=raggedright, singlelinecheck=false}
\centering
\includegraphics[width=\columnwidth]{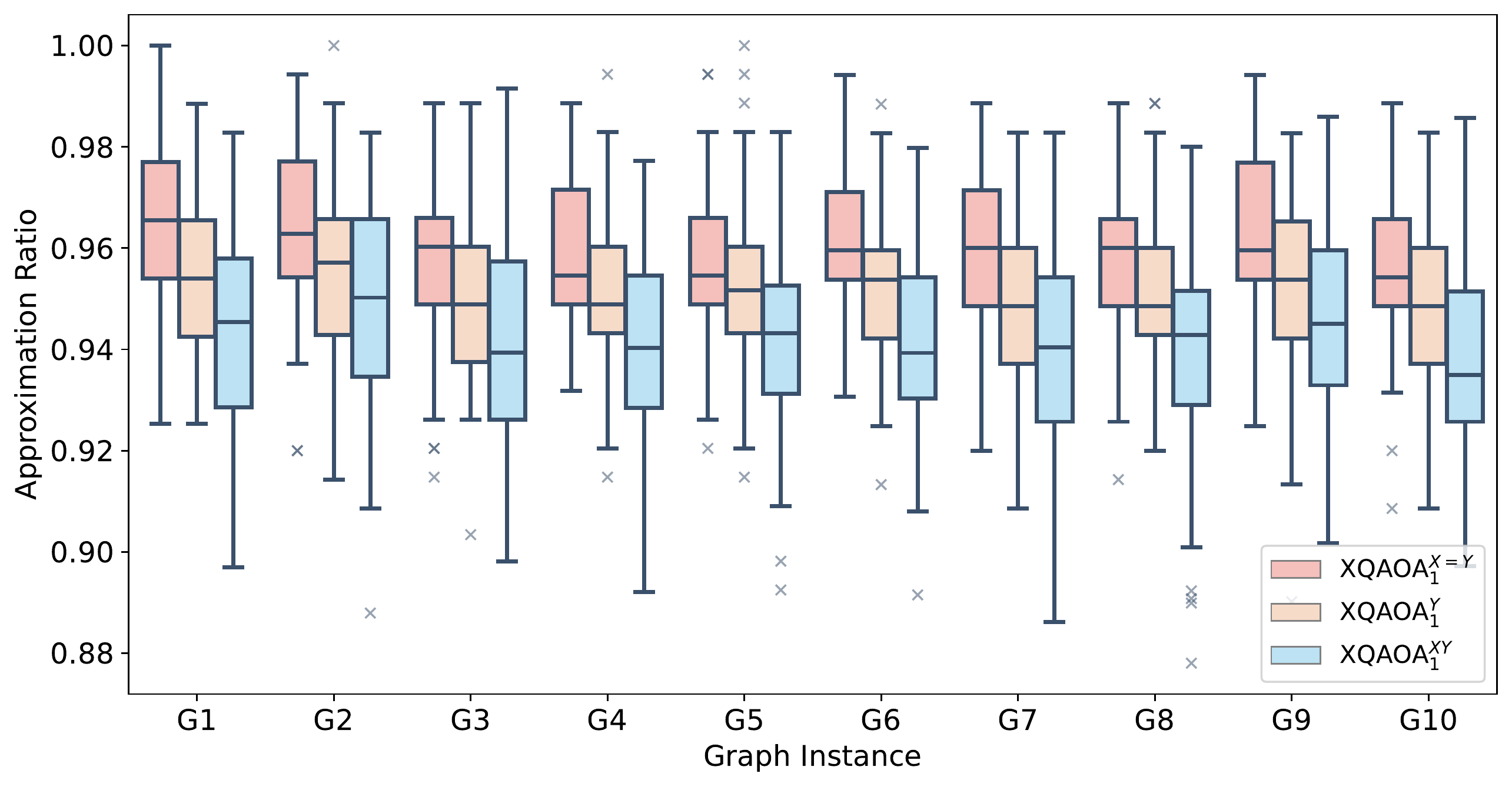}
    \caption{Improvement of $\XQAOA{X=Y}$ ansatz over the $\XQAOA{XY}$ and $\XQAOA{Y}$ ansatz for ten different instances of 3-regular graphs with 128 vertices. For each of the three ansatz variants, the Parallel-LBFGS optimiser was run 100 times with random initial values for each of the ten graph instances.}
    \label{xqaoa_variants}
\end{figure}

\begin{figure*}[htbp]
  \includegraphics[width=\textwidth]{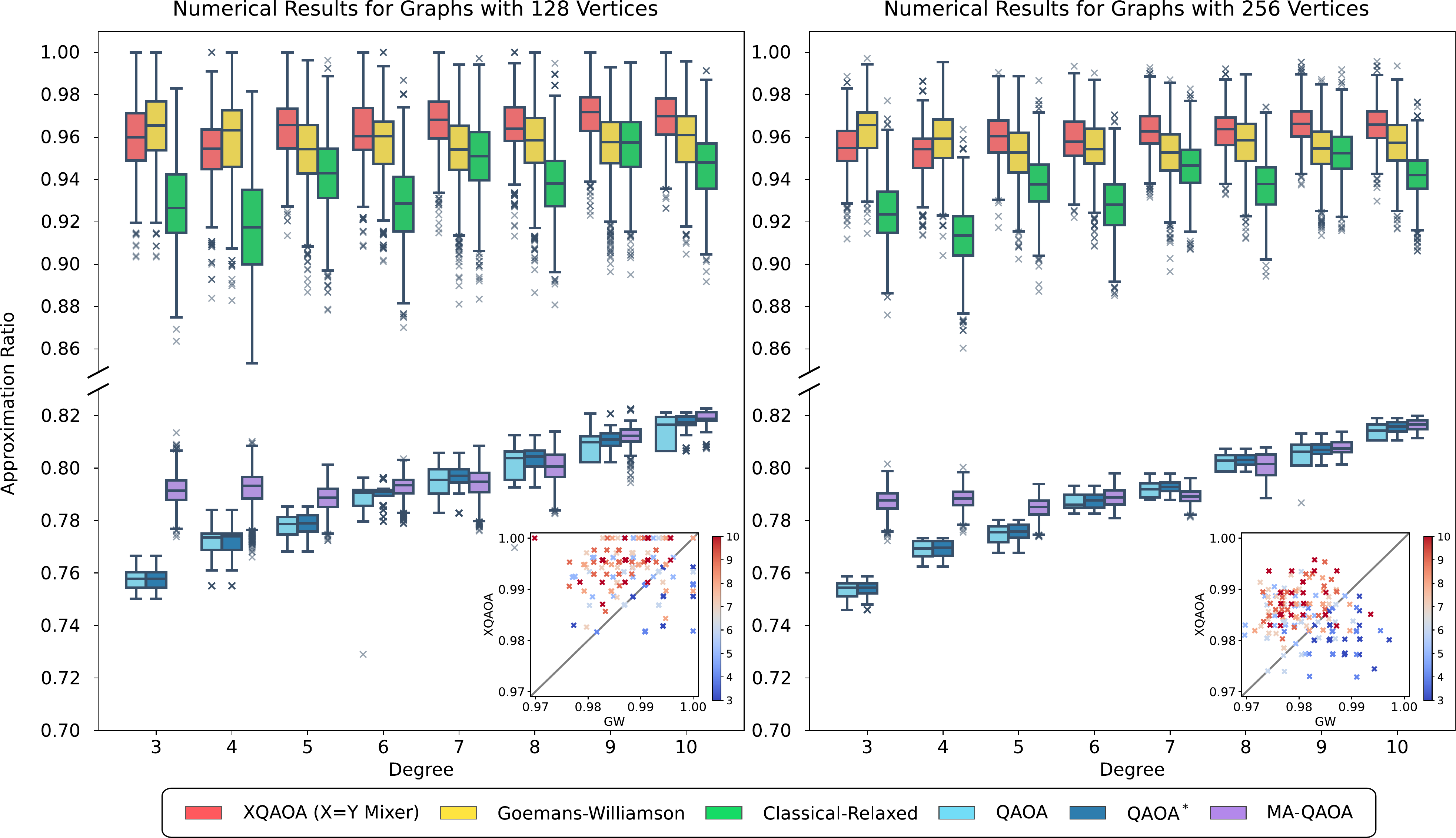}
  \captionsetup{justification=raggedright, singlelinecheck=false}
  \caption{This boxplot compares the sizes of cuts obtained by quantum and classical algorithms against the maximum cut size found by GUROBI. The data used to create the boxplots comes from 100 random cuts generated using the GW algorithm, 100 random initialisations for several other algorithms ($\XQAOA{X=Y}$, MA-QAOA$_1$, QAOA$_1$, and CR), and 100 informed initialisations for QAOA$_1$ (labelled as QAOA$^*$) applied to 25 different instances of regular graphs with 128 or 256 nodes and degree values ranging from 3 to 10. The whiskers extend up to data points within 1.5 times the interquartile range from the upper and lower quartiles, and crosses represent outliers. For clarity, outliers of QAOA$_1$, with approximation ratios less than $0.7$, that resulted from the barren plateau have been omitted. The inset scatterplot compares the best-found solutions for 100 runs for the $\XQAOA{X=Y}$ and the GW algorithm on all the graph instances. The points in the upper-left corner show that $\XQAOA{X=Y}$ performs better than GW, while the opposite is true for points in the lower-right corner. The colour of the points indicates the degree of the corresponding graph.}
  \label{benchmark_plots}
\end{figure*}

We evaluate the performance of the XQAOA algorithm by benchmarking it on the MaxCut problem on unweighted $D$-regular graphs that were generated using an algorithm developed by Steger and Wormald~\cite{steger1999generating}. Since the three different configurations $\XQAOA{XY}$, $\XQAOA{X=Y}$ and $\XQAOA{Y}$ of XQAOA that we consider have performances that are not known \emph{a priori}, we first benchmark them on 10 randomly generated instances of $3$-regular graphs with $128$ vertices. The best-performing configuration of XQAOA$_1$ is then benchmarked against MA-QAOA$_1$, QAOA$_1$, and the CR and GW algorithms on 25 instances of $D$-regular graphs with $128$ and $256$ vertices for $3 \leq D \leq 10$. The XQAOA$_1$, MA-QAOA$_1$, QAOA$_1$, and CR algorithms are benchmarked by performing 100 runs of the classical optimiser with random initial points, whereas the GW algorithm is benchmarked by first solving the relaxed problem and then generating 100 random vectors for hyperplane rounding. For an explanation of the GW algorithm, we refer the reader to \cref{gw_algorithm}. The CR algorithm computes the MaxCut by simply running the optimiser on the relaxed version of \cref{maxcut_quadratic_program}, i.e.
\begin{equation}
\begin{array}{ll@{}ll}
\text{Maximise}  & \displaystyle\sum\limits_{\{u,v\} \in E} \frac{1}{2} w_{uv} \left(1 - \sin\theta_u\sin\theta_v \right) ,
\end{array}
\label{maxcut_relaxed}
\end{equation}
where the maximisation is performed over angles $\theta_u$ and $\theta_v$. 

To compute the approximation ratio, we need to obtain the exact MaxCut values, which we did using the GUROBI solver~\cite{gurobi}, a widely used industry tool. Although proving optimality with GUROBI takes exponential time, it can find solutions quickly. GUROBI was able to identify optimal solutions for 128-node graphs and near-optimal solutions for 256-node graphs with at most $6\%$ MIPGap\footnote{The MIPGap is the gap between the lower and upper objective bound divided by the absolute value of the incumbent objective value.}. To compute the expectation values of QAOA$_1$, MA-QAOA$_1$, and XQAOA$_1$ for large $n$, we used the analytical results from theorems~\ref{qaoa_hadfield1},~\ref{ma_qaoa_theorem}, and~\ref{XQAOA_Full_Thm}. The Parallel-LBFGS algorithm was used to optimise the variational parameters of QAOA$_1$, MA-QAOA$_1$, XQAOA$_1$, and the CR algorithm.

Finally, to assess the performance of quantum algorithms at increased depths, we expanded our benchmarking to include depths ranging from 1 to 5 for the most effective XQAOA variant, as well as QAOA and MA-QAOA. Due to the lack of analytical formulas for larger $p$ values and the significant computational complexity associated with simulating deep quantum circuits, we conducted our benchmarks using the Qiskit~\cite{Qiskit} simulator on small graphs.

\subsection{Benchmark Results for \texorpdfstring{$p = 1$}{p = 1}} \label{benchmark_results_unit_depth}

Our comparative analysis of the three XQAOA variants on $3$-regular graphs revealed that the $\XQAOA{X=Y}$ variant performed the best (see \cref{xqaoa_variants}). When benchmarked against MA-QAOA$_1$, QAOA$_1$, CR, and the GW algorithm on $D$-regular graphs with 128 and 256 vertices for $3 \leq D \leq 10$, $\XQAOA{X=Y}$ consistently outperformed QAOA$_1$, MA-QAOA$_1$, and the CR algorithm on all graph instances. Notably, $\XQAOA{X=Y}$ demonstrated competitive performance against the GW algorithm for $3$ and $4$-regular graphs and exceeded it for $D$-regular graphs with $D > 4$ (see \cref{benchmark_plots}\footnote{Here, we note that in \cref{benchmark_plots}, the boxplots for the $\XQAOA{X=Y}$, CR, and GW algorithms show the distributions of the approximation ratios obtained for individual solutions. Notably, for $\XQAOA{X=Y}$, individual solutions were classically extracted by leveraging the quantum-classical transition, a phenomenon detailed in \cref{qc_transition}. In contrast, the boxplots for QAOA$_1$, QAOA$^*_1$, and MA-QAOA$_1$ show the distributions of expected approximation ratios; we do this in lieu of computing the approximation ratios of individual solutions, as the latter would require implementation of the quantum algorithm on a quantum computer in order to obtain samples from measuring the output states of the circuits involved. Hence, the actual solutions obtained from QAOA$_1$, QAOA$^*_1$, and MA-QAOA$_1$ may differ from the expected approximation ratios shown in the boxplots, and could be either higher or lower. For an extended discussion of the distinction between using expected approximation ratios and approximation ratios of individual samples, we refer the reader to Larkin \textit{et al.}~\cite{larkin2022evaluation}.}). The boxplots reveal that the lower, middle and upper quartile values of $\XQAOA{X=Y}$ are significantly higher than those of the GW algorithm for graphs with $D > 4$, indicating that $\XQAOA{X=Y}$ is more likely to produce a better solution irrespective of the parameter initialisation strategy. This robustness of the $\XQAOA{X=Y}$ to initial parameter choices can be attributed to its overparameterised ansatz, which is further explained in \cref{xqaoa_overparameterisation}.

Our analysis also highlighted a linear increase in the approximation ratio of QAOA$_1$ with the degree of the graph. Its performance nears that of MA-QAOA$_1$ for $D>5$, hinting at a possible reachability deficit~\cite{akshay2020reachability,akshay2021reachability}, limiting MA-QAOA$_1$'s ability to find an approximate solution close to the optimal. It is also important to note that the QAOA$_1$ ansatz experiences the barren plateau phenomenon~\cite{mcclean2018barren}, with the size of the plateau increasing with the degree of the graph. To mitigate this, careful selection of initial points for the classical optimiser is crucial, as outlined in \cref{qaoa_optimization}. The results of this strategy are also presented in \cref{benchmark_plots} as QAOA$^*$.

\subsection{Quantum--Classical Transition} \label{qc_transition}
\begin{figure}[htbp]
\includegraphics[width=\columnwidth]{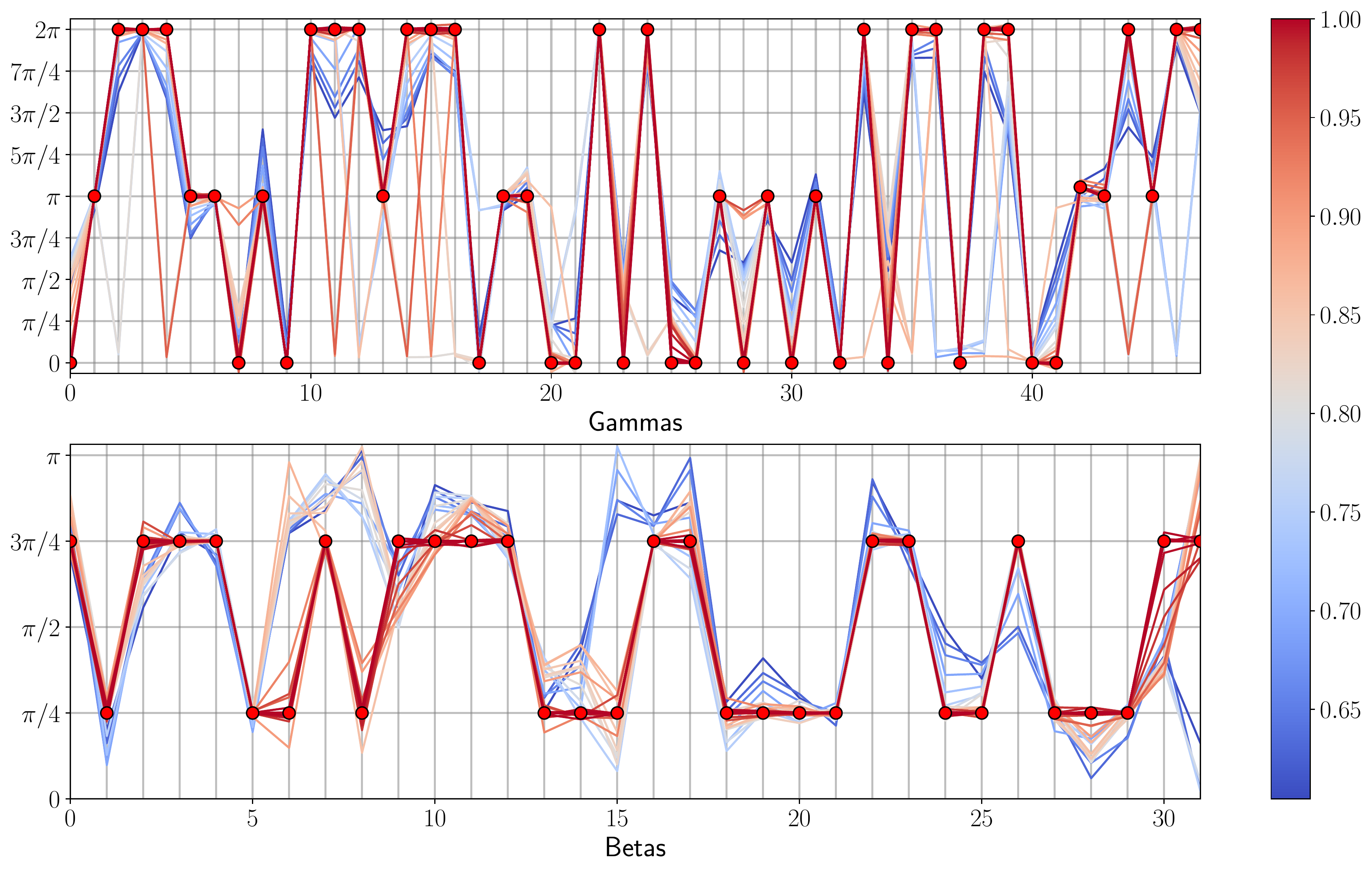}
    \captionsetup{justification=raggedright, singlelinecheck=false}
    \caption{This graph illustrates the changes in the $\boldsymbol{\gamma}$ and $\boldsymbol{\beta}$ angles of the $\XQAOA{X=Y}$ ansatz as the classical optimiser reaches the optimal solution for a 3-regular graph with 32 vertices and 48 edges. Each line in the graph represents a set of evaluated parameters, with the colour indicating the corresponding approximation ratio. The solid red circles in the graph represent the final angles determined by the classical optimiser for the $\XQAOA{X=Y}$ ansatz.}
    \label{xqaoa_angles}
\end{figure}
In our numerical simulations with the $\XQAOA{X=Y}$ ansatz, we observed that as the classical optimiser converged to an optimum, the optimal angles $(\boldsymbol{\gamma}^*, \boldsymbol{\beta}^*)$ stabilised at specific values. Specifically, $\gamma_{uv}^*$ converged to values in $\{0, \pi, 2\pi\}$\footnote{In weighted graphs, the optimal angles $\gamma_{uv}^*$ are scaled by the edge weight $w_{uv}$. The effective optimal angle, discounting the weight prefactor, can be calculated as $\gamma_{uv}^* = w_{uv} \gamma_{uv}^* \mod 2\pi$.}, while $\beta_u^*$ converged to values in $\{\pi/4, 3\pi/4\}$ (see \cref{xqaoa_angles}). When $\gamma_{uv}^* \in \{0, 2\pi\}$, the two-qubit R$_{\mathrm{ZZ}}(0) $ and  R$_{\mathrm{ZZ}}(2\pi)$ gates act as the identity gate, leaving qubits $u$ and $v$ unchanged. When $\gamma_{uv}^*=\pi$, the two-qubit gate R$_{\mathrm{ZZ}}(\pi) = \ket{+,+}\bra{-,-} + \ket{+,-}\bra{-,+} +  \ket{-,+}\bra{+,-} + \ket{-,-}\bra{+,+}$. In other words, it swaps the states $\ket{+}_u \leftrightarrow \ket{-}_u$ and $\ket{+}_v \leftrightarrow \ket{-}_v$. Since the initial state $\ket{s}$ is a product state $\ket{+\ldots+}$, the state after all two-qubit gates will be a product state of $\ket{+}$s and $\ket{-}$s. Specifically, qubit  $u$ will be in the $\ket{+}_u$ state if it is swapped an even number of times and in the $\ket{-}_u$ state if it is swapped an odd number of times. Finally, the action of the mixing gates is to convert the $\ket{+}$ and  $\ket{-}$ states to  $\ket{0}$ or  $\ket{1}$ depending on the value of $\beta^*$:
\begin{align}
\text{When } \beta^*_i = \pi/4\,, & \begin{cases} \ket{+}_i \rightarrow
  \ket{1}_i\\
 \ket{-}_i \rightarrow
  \ket{0}_i\end{cases}  \\
\text{and when } \beta_i^* = 3\pi/4\,, & \begin{cases} \ket{+}_i \rightarrow
  \ket{0}_i\\
 \ket{-}_i \rightarrow
  \ket{1}_i\end{cases}  \,,
\end{align}
allowing us to read off the classical bit-string. It is important to acknowledge that due to degenerate local optimums, a relatively small subset of angles might not stabilise at the specified values even after convergence. Given the benign characteristics of the $\XQAOA{X=Y}$ ansatz's loss landscape, such deviations are highly unlikely. However, in cases where deviations occur, we force $\gamma^*_{uv}$ to take the closest value in $\{0,\pi,2\pi\}$\footnote{For example, in a two-vertex graph, if $\gamma^*_{uv} \notin \{0, \pi, 2\pi\}$, the final state may be entangled, such as $c_0\ket{01} + c_1\ket{10}$. In this scenario, both measurement outcomes `01' and `10' yield the same MaxCut value, rendering the specific choice of $\gamma^*_{uv}$ inconsequential.}.

Reviewing \cref{xqaoa_angles} from a different perspective, we observe that the $\XQAOA{X=Y}$ ansatz, initially set with random angles, creates a highly entangled quantum state with a multitude of superposed states. As optimisation progresses, this entanglement gradually decreases, leading to a marked reduction in the number of superposed states. When the optimum is reached, the entangling layer disappears, leaving the system in a singular definitive state. In essence, what begins as a distinctly quantum state, through the course of optimisation, evolves into a classical state. This transition, marked by the disappearance of the entangling layer, facilitates the extraction of the solution through classical means, negating the need for quantum computers. This quantum-to-classical transition raises a natural question of whether the entangling layer is fundamentally necessary.
\begin{figure}[htbp]
\centering
\includegraphics[width=\columnwidth]{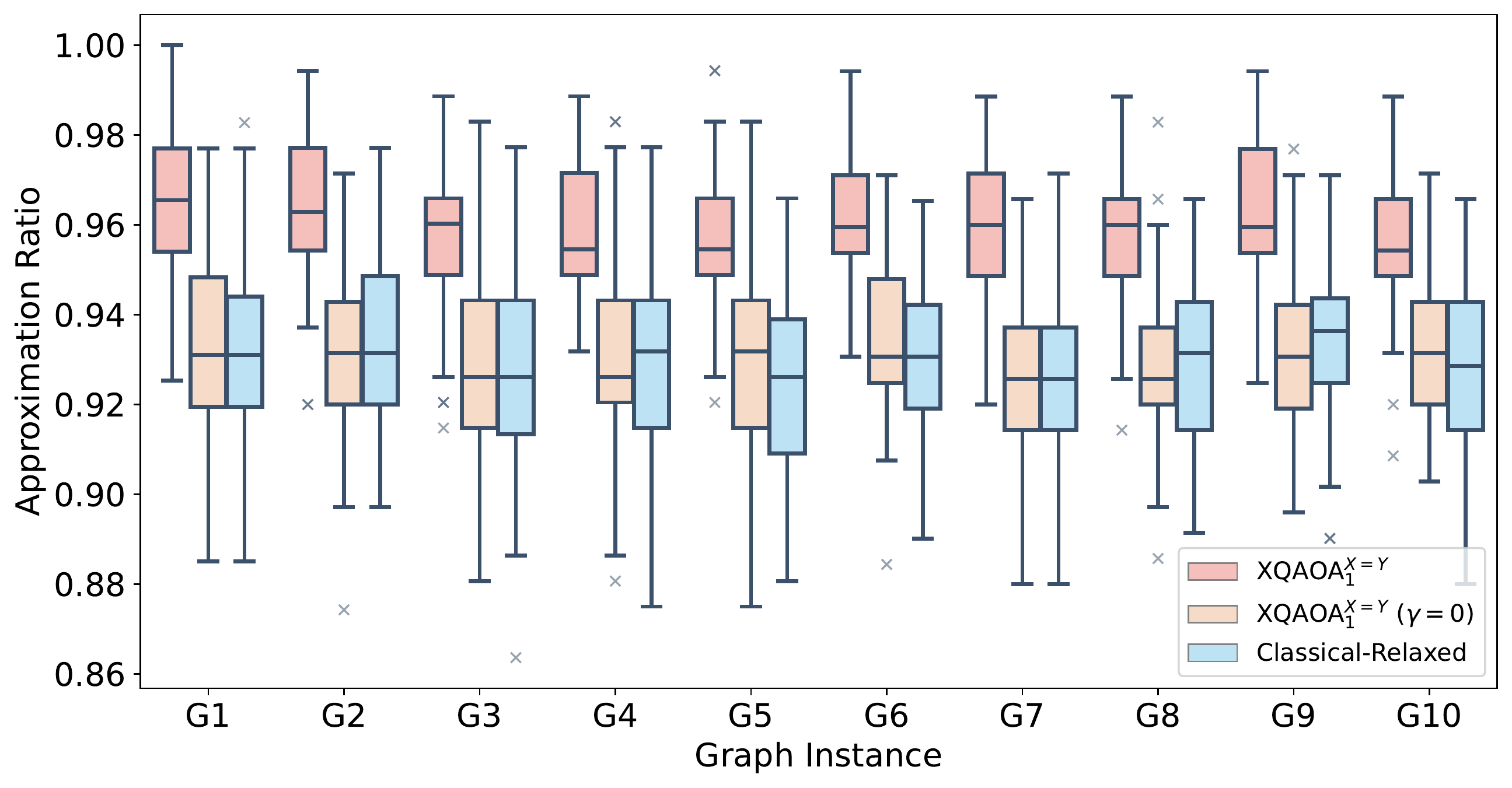}
    \captionsetup{justification=raggedright, singlelinecheck=false}
    \caption{Improvement of $\XQAOA{X=Y}$ ansatz over $\XQAOA{X=Y}$ with $\boldsymbol{\gamma} = 0$, and the Classical-Relaxed algorithms for ten different instances of 3-regular graphs with 128 vertices. For each of the three ansatz variants, the Parallel-LBFGS optimiser was run 100 times with random initial values for each of the ten graph instances.}
    \label{xqaoa_no_gamma}
\end{figure}
To answer this, we conducted further numerical simulations by setting $\boldsymbol{\gamma} = \mathbf{0}$ and optimising only the $\boldsymbol{\beta}$ angles. Our numerical results showed that without any entangling layer, the performance of $\XQAOA{X=Y}$ was similar to that of CR (see \cref{xqaoa_no_gamma}). In fact, when we set $\boldsymbol{\gamma} = \mathbf{0}$, \cref{xqaoa_full_exp} reduces to 
\begin{equation}
	\left\langle C_{u v}\right\rangle_{\mathrm{X=Y}} = \frac{1}{2} w_{uv} (1 - \sin 2\beta_u \sin 2\beta_v),
\end{equation}
which is the same as \cref{maxcut_relaxed} with the angles having an additional factor of $2$. This demonstrates that the overparameterised entangling layer augments the landscape, making the gradient-based classical optimiser less susceptible to getting trapped in local optima.

\subsection{Benchmark Results for \texorpdfstring{$1 \leq p \leq 5$}{1 <= p <= 5}} \label{benchmark_results_higher_depth}
\begin{figure}[htbp]
\captionsetup{justification=raggedright, singlelinecheck=false}
\centering
\includegraphics[width=\columnwidth]{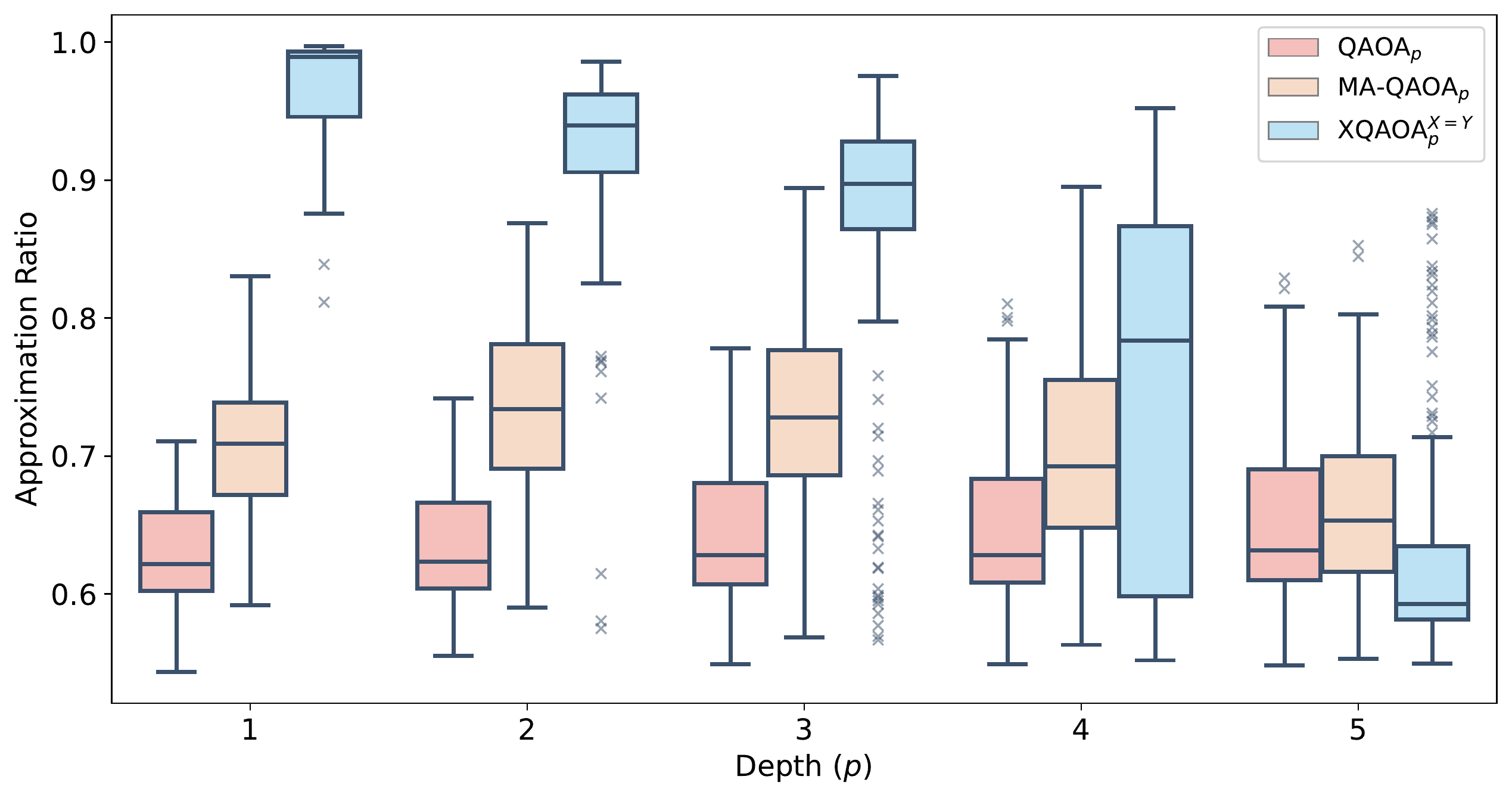}
    \caption{Comparison of approximation ratios achieved by QAOA$_{p}$, MA-QAOA$_{p}$, and XQAOA$^{\text{X=Y}}_{p}$ for $1 \leq p \leq 5$, evaluated across 20 instances of 3-regular graphs, each with 16 vertices. The Powell optimiser was run 10 times with random initialisation for each graph instance across all three algorithms.}
    \label{large_p_sim_results}
\end{figure}
In previous sections, we showed that $\XQAOA{X=Y}$ achieves near-optimal solutions and surpasses QAOA$_1$ and MA-QAOA$_1$. This led us to compare XQAOA$^{\text{X=Y}}_p$'s performance with QAOA$_p$ and MA-QAOA$_p$ for $p > 1$. Lacking analytical formulas for expectation values for $p > 1$ and constrained by the computational overhead of large-scale quantum simulation, we performed small-scale simulations using Qiskit~\cite{Qiskit}, benchmarking 20 random 3-regular graph instances with 16 vertices each. Given the variability in expectation values due to the limited number of shots (1024), we utilised the Powell optimiser~\cite{powell1964efficient} with ten random restarts to optimise the ansatzes' parameters. The simulation results, presented in \cref{large_p_sim_results}, reveal that the median approximation ratio of XQAOA$^{\text{X=Y}}_p$ consistently outperforms QAOA$_p$ and MA-QAOA$_p$ up to $p = 4$, with $\XQAOA{X=Y}$'s median nearly reaching 1.0. Interestingly, while QAOA$_p$ shows modest improvement with increasing depth, MA-QAOA$_p$ peaks at $p = 3$ before declining at $p \geq 4$\footnote{It should be noted that MA-QAOA$_p$'s performance is lower-bounded by QAOA$_p$ when its parameter optimisation is warm-started with the optimal angles of QAOA$_p$. The suboptimal results of MA-QAOA$p$ in our simulations stem from the use of random parameter initialisation. For a detailed analysis of MA-QAOA$_{p>1}$'s performance with various parameter initialisation strategies, see Gaidai \textit{et al.}\cite{gaidai2023performance}.}. In contrast, XQAOA$^{\text{X=Y}}_p$ exhibits a gradual decline in performance beyond $p = 1$, particularly noticeable at $p \geq 4$. This decline is attributed to barren plateaus due to the circuit's overexpressiveness, rather than a lack of effectiveness at higher depths~\cite{holmes2022connecting, larocca2022diagnosing}.

\section{Discussion} \label{discussion}

\subsection{Randomness in the Approximation Algorithms}

We optimise the classical parameters of XQAOA$_1$, MA-QAOA$_1$, QAOA$_1$, and CR algorithms using a gradient-based classical optimiser. We start the optimisation process with a randomly chosen initial point which affects the quality of the solution found, especially if the optimisation landscape is non-convex and has non-trivial features. If the initial point is near a local optimum or barren plateau, the gradient-based optimiser will converge to the local optimum and return a suboptimal solution. As a result, XQAOA$_1$, MA-QAOA$_1$, QAOA$_1$, and CR algorithms have a wide range of approximate solutions for the same problem. The maximum approximation ratio returned by MA-QAOA$_1$ is 0.82, whereas XQAOA$_1$ can often achieve an approximation ratio of 1.0 with the right choice of initial parameters. The range of approximation ratios for the CR algorithm is similar to XQAOA$_1$, but numerical simulations suggest that the cost landscape of the CR algorithm may be difficult to navigate and plagued with local optima, which may require an exponential number of initial points for the CR to match XQAOA$_1$'s performance, thus negating the benefit of having a polynomial-time approximation algorithm.

The GW algorithm also has randomness in its process. After it solves the relaxed version of the MaxCut problem, it generates an $n$-dimensional random vector $\boldsymbol{r}$ to perform its hyperplane rounding to find the optimal cut. While GW has an expected approximation ratio of $0.87856$ in the worst case at the asymptotic limit, the distribution of the approximation ratio returned by the GW algorithm for a finite number of randomly generated $\boldsymbol{r}$ vectors and problem instances can vary significantly, which is why we see a distribution for the output of the GW algorithm. While XQAOA$_1$, MA-QAOA$_1$, and the CR algorithm each require solving the problem anew for each random initial point, the GW algorithm solves the relaxed MaxCut problem just once and then efficiently generates individual solutions through hyperplane rounding with randomly generated vectors, avoiding repetitive computations.

\subsection{Classical Simulability of the XQAOA Ansatz}

Computing expectation values of QAOA$_1$, MA-QAOA$_1$, and XQAOA$_1$ for MaxCut on arbitrary graphs all have a time complexity of $O(n^3)$, albeit with varying prefactors. XQAOA$_1$ is unique in that its entangling layer vanishes at the optimal solution, in contrast to QAOA$_1$ and MA-QAOA$_1$, which maintain their entangling layer post-convergence. This entangling layer in QAOA$_1$ and MA-QAOA$_1$ requires generating an entangled quantum state through quantum computation followed by measurements to assign values to variables, precluding efficient classical simulation. In contrast, XQAOA$_1$'s optimal state, characterised by zero gammas, is non-entangled and permits classical bit assignment without quantum computation. However, in cases where XQAOA$_1$ falls short, a $p > 1$ XQAOA$_p$ might be required, necessitating quantum computation. 

It remains an open question whether a simple analytical formula exists for efficiently computing the mean values of Pauli operators of the XQAOA$_1$, MA-QAOA$_1$, and QAOA$_1$ states when dealing with $k$-local Ising Hamiltonians for $k > 2$. Additionally, if the problem assignments take integer-valued arguments, the quantum circuit would require more qubits per variable, further increasing the entangling and non-locality of the circuit, potentially allowing for quantum advantage. The $p = 1$ Recursive QAOA (RQAOA$_1$)~\cite{bravyi2020obstacles} is another quantum algorithm that can solve the MaxCut problem classically in time $O(n^4)$ without requiring any quantum computation or multi-qubit measurements.

\subsection{Role of Overparameterisation in XQAOA} \label{xqaoa_overparameterisation}

The efficacy of XQAOA largely stems from its overparameterised ansatz. Overparameterisation, which involves introducing additional parameters to increase model dimensionality, reshapes the loss landscape to facilitate easier optimisation. To understand the impact of overparameterisation, consider the role of local minima in this context: a local minimum is a point in the loss landscape with a loss value lower or equal to those of its neighbours within an $\epsilon$ radius. While such minima are prevalent in lower dimensions, they become less likely with increasing dimensionality as it becomes harder for their loss values to remain the lowest across all new dimensions, converting them from minima into saddle points. Saddle points differ from local minima as they are not the lowest points in all directions; thus, optimisers can navigate past them more readily. As overparameterisation turns more potential local minima into saddle points, the path to the global minimum becomes less obstructed. Although the overparameterised models may not be completely devoid of local minima, the remaining local minima tend to be close to the global minimum. This proximity reduces the likelihood of settling for suboptimal solutions and facilitates more efficient convergence to the global minimum~\cite{allen2019convergence, du2019gradient, buhai2020empirical, du2018gradient, brutzkus2017sgd}. This advantageous effect of overparameterisation is evident in \cref{benchmark_plots}, where the lower quartiles for $\XQAOA{X=Y}$ consistently surpass $0.92$ approximation ratio on all the benchmarked graph instances, strongly hinting at the benign loss landscape devoid of barren plateaus and suboptimal local minima.

While the previous discussion intuitively explains how overparameterisation enhances the efficacy of XQAOA, it's crucial to place this within the wider context of ongoing research on the trainability\footnote{In this context, `trainability' refers to the ability of a PQC to efficiently adjust its variational parameters for optimising a cost function. The terms `train,' `trainable,' and `trainability' are often used interchangeably with `optimise' and `optimisable' in quantum algorithms literature due to the parallel drawn between PQCs and Quantum Neural Networks, where the process of optimising network variables is commonly referred to as `training.'} of parameterised quantum circuits (PQCs). These studies, which form the basis of quantum landscape theory (QLT)~\cite{larocca2023theory}, offer a deeper understanding of quantum loss landscapes~\cite{mcclean2018barren, holmes2022connecting, larocca2022diagnosing, arrasmith2022equivalence, you2022convergence, liu2023analytic, stkechly2023connecting, anschuetz2022beyond, wang2021noise, kiani2020learning, garcia2023effects, akshay2021parameter, brandao2018fixed}. QLT defines a PQC to be overparameterised when it has sufficiently many parameters to explore all relevant directions of its state space~\cite{larocca2023theory}. An essential aspect of this definition is the inherent expressiveness of the PQC, which is characterised by its ability to generate a wide range of unitaries under varied parameter settings~\cite{sim2019expressibility}. This distinction is particularly evident when comparing $\XQAOA{X=Y}$ and MA-QAOA$_1$ for the MaxCut problem; despite having an equal number of parameters, $\XQAOA{X=Y}$ consistently outperforms MA-QAOA$_1$ on all problem instances. However, high expressivity also has drawbacks, such as the barren plateau phenomenon, where circuits show vanishingly small gradients due to their expressiveness~\cite{holmes2022connecting}. XQAOA addresses this by using a problem-specific ansatz, tailoring its circuit design to the task at hand. This approach allows XQAOA to be overparameterised with a quadratic number of parameters, in contrast to generic ansatzes that may require exponentially more parameters and deeper circuits~\cite{haug_quantum_geometry, larocca2023theory}. XQAOA thus achieves an optimal balance, avoiding both the limitations of underparameterisation, such as spurious local minima and reachability deficits~\cite{anschuetz2022beyond}, and the challenges of high expressivity like barren plateaus. This positions XQAOA in an optimal `Goldilocks zone' of trainability and expressivity.

\subsection{Shallow XQAOA vs Deep QAOA} \label{xqaoa_vs_qaoa}
The decision to use shallow XQAOA versus deep QAOA circuits hinges on their trainability for a given problem. Trainability depends on the loss landscapes of their ansatzes, influenced by factors like parameter count, initialisation strategy, circuit depth, and hardware noise. These factors can adversely affect trainability, hindering optimisation. Thus, the choice between shallow XQAOA and deep QAOA is determined by their relative trainability under these conditions.

It is evident that to surpass the performance of XQAOA$_1$, QAOA circuits need to be sufficiently deep. For example, it is conjectured that a minimum depth of $p = 12$ is necessary for QAOA to outperform the GW algorithm on the MaxCut problem on unweighted 3-regular graphs~\cite{wurtz2021fixed}. This conjecture not only posits a depth requirement but also suggests a parameter initialisation strategy of using a set of predetermined angles for warm-starting the optimisation. However, this strategy becomes ineffective in the presence of hardware noise~\cite{wang2021noise}. The noise distorts the loss landscape, leading to barren plateaus, and necessitates deeper, more noise-prone circuits. In contrast, our findings suggest that $\XQAOA{X=Y}$ is adequate for the MaxCut problem, questioning the need for deeper XQAOA circuits on NISQ devices. While XQAOA may require $p > 1$ for certain problems, necessitating quantum computation, its relatively shallower depth compared to QAOA makes it less susceptible to noise. With sufficiently low noise levels, XQAOA can still achieve near-optimal solutions, albeit with additional random restarts or parameter initialisation strategies~\cite{garcia2023effects}.

For problems unlike MaxCut on unweighted 3-regular graphs, where patterns are ambiguous~\cite{brandao2018fixed, akshay2021parameter} and \emph{a priori} information is limited, training deep QAOA circuits is more challenging~\cite{PhysRevLett.127.120502}. This is due to the difficulty in determining the minimal effective depth~\cite{qcma_hard} and suitable parameter initialisation strategies in scenarios where random initialisation is suboptimal~\cite{mcclean2018barren}. For such problems, where extracting useful \emph{a priori} information is challenging, XQAOA may be a preferable choice, regardless of whether the quantum computers are NISQ or fault-tolerant. An example of such a problem could be the MaxCut on randomly weighted regular graphs or randomly generated graphs.

\section{Conclusion} \label{conclusion}

In this work, we presented the XQAOA ansatz and its variants and explained how they generalise the MA-QAOA and QAOA ansatzes. Our numerical simulations reveal that a single iteration of the XQAOA ansatz, especially with the X=Y mixer, outperforms a single iteration of both MA-QAOA and QAOA. This enhanced performance of XQAOA is attributed to its overparameterised ansatz, which enables exploration in all relevant directions of its state space. The incorporation of the Pauli-Y rotation gate also significantly contributes to this improved efficacy. Our benchmarks also reveal that XQAOA performs just as well as the state-of-the-art Goemans-Williamson algorithm and even outperforms it for unweighted regular graphs with degrees greater than 4. Additionally, we find that the naive Classical-Relaxed algorithm with fewer classical parameters than MA-QAOA performs better by a large margin and that the performance of QAOA grows arbitrarily close to MA-QAOA for regular graphs with increasing degrees. Finally, we find an infinite family of graphs for which XQAOA solves MaxCut exactly and show analytically that for some graphs in this family, special cases of XQAOA are capable of achieving a much larger approximation ratio than QAOA.

Interestingly, we found that as the $\XQAOA{X=Y}$ ansatz converges to an optimum, its entangling layer disappears, leaving behind only single-qubit unitaries, making it possible to efficiently solve and extract the solution classically. Although the entangling layer disappears as the ansatz reaches an optimal solution, it is necessary for the optimisation process, without which the performance of the $\XQAOA{X=Y}$ ansatz deteriorates to that of the Classical-Relaxed algorithm. Although for the problem of MaxCut, the efficient classical simulation of the $p=1$ XQAOA ansatz eliminates quantum advantage, it remains open whether this is still the case for larger $p$ and problems whose Ising formulations are 2-local with external fields or $k$-local with $k > 2$.

We have also shown that despite the XQAOA ansatz being overparameterised—with a quadratic increase in free parameters in the worst-case scenario—it is significantly easier to train compared to the underparameterised QAOA ansatz and the adequately parameterised Classical-Relaxed algorithm. The QAOA ansatz struggles with issues like spurious local minima, barren plateaus, and reachability deficits, while the Classical-Relaxed algorithm often encounters sub-optimal local minima far from the global optimum. In contrast, the XQAOA ansatz, like other overparameterised models, features a more benign loss landscape, free from barren plateaus, spurious local minima, and reachability deficits. This characteristic enables the classical optimiser to consistently converge to optimal or near-optimal solutions, independent of the parameter initialisation strategy. While the increased number of free parameters in XQAOA might suggest higher computational costs, its faster convergence rate, eliminating the need for specific initialisation strategies or random restarts, compensate for the extra parameters' computational overhead.

Our work opens up new avenues for further research into improving quantum optimisation algorithms as well as their impact on various applications. For example, QAOA and its variants have already found numerous potential uses in solving various optimisation problems beyond MaxCut, including problems in graph theory~\cite{wang2022quantum, bengtsson2020improved, basso2022performance}, finance \cite{brandhofer2023benchmarking}, chemistry \cite{kremenetski2021quantum,mustafa2022variational}, and others \cite{hadfield2019quantum,mesman2021qpack}. Future work could extend these results by adopting and exploiting the advantages of XQAOA in various applications. In addition, due to its advantages at low depth, XQAOA could be tested and implemented on near-term quantum hardware and compared against existing experimental benchmarks \cite{leo2020quantum,mesman2021qpack,bengtsson2020improved,lubinski2023optimization}.

\section{Acknowledgements}
VV is thankful to Ye Jun from the A*STAR Institute of High Performance Computing and the A*STAR Computational Resource Centre for supporting this work through the use of their high-performance computing facilities. VV is thankful to Aaron Tranter for the stimulating discussions and insightful suggestions. We thank Truman Ng for helpful comments on an earlier version of this manuscript. This research is supported by A*STAR C230917003 and the Australian Research Council Centre of Excellence CE170100012. DEK acknowledges funding support from the A*STAR Central Research Fund (CRF) Award for Use-Inspired Basic Research; and the National Research Foundation, Singapore and A*STAR under the Quantum Engineering Programme (NRF2021-QEP2-02-P03).

\section{Data Availability}
We provide a total of 420 $D$-regular graphs that were used in this paper, along with their optimal cut and their solutions in a machine-readable \texttt{CSV} format. We also provide the simulation data and scripts used to generate the plots presented in this paper. The benchmark and simulation dataset and the scripts for generating the plots can be found at \href{https://github.com/vijeycreative/XQAOA-Dataset}{https://github.com/vijeycreative/XQAOA-Dataset} \cite{V_Vijendran_XQAOA-Dataset_2023}.

\bibliography{refs}

\begin{thebibliography}{123}%
\makeatletter
\providecommand \@ifxundefined [1]{%
 \@ifx{#1\undefined}
}%
\providecommand \@ifnum [1]{%
 \ifnum #1\expandafter \@firstoftwo
 \else \expandafter \@secondoftwo
 \fi
}%
\providecommand \@ifx [1]{%
 \ifx #1\expandafter \@firstoftwo
 \else \expandafter \@secondoftwo
 \fi
}%
\providecommand \natexlab [1]{#1}%
\providecommand \enquote  [1]{``#1''}%
\providecommand \bibnamefont  [1]{#1}%
\providecommand \bibfnamefont [1]{#1}%
\providecommand \citenamefont [1]{#1}%
\providecommand \href@noop [0]{\@secondoftwo}%
\providecommand \href [0]{\begingroup \@sanitize@url \@href}%
\providecommand \@href[1]{\@@startlink{#1}\@@href}%
\providecommand \@@href[1]{\endgroup#1\@@endlink}%
\providecommand \@sanitize@url [0]{\catcode `\\12\catcode `\$12\catcode
  `\&12\catcode `\#12\catcode `\^12\catcode `\_12\catcode `\%12\relax}%
\providecommand \@@startlink[1]{}%
\providecommand \@@endlink[0]{}%
\providecommand \url  [0]{\begingroup\@sanitize@url \@url }%
\providecommand \@url [1]{\endgroup\@href {#1}{\urlprefix }}%
\providecommand \urlprefix  [0]{URL }%
\providecommand \Eprint [0]{\href }%
\providecommand \doibase [0]{https://doi.org/}%
\providecommand \selectlanguage [0]{\@gobble}%
\providecommand \bibinfo  [0]{\@secondoftwo}%
\providecommand \bibfield  [0]{\@secondoftwo}%
\providecommand \translation [1]{[#1]}%
\providecommand \BibitemOpen [0]{}%
\providecommand \bibitemStop [0]{}%
\providecommand \bibitemNoStop [0]{.\EOS\space}%
\providecommand \EOS [0]{\spacefactor3000\relax}%
\providecommand \BibitemShut  [1]{\csname bibitem#1\endcsname}%
\let\auto@bib@innerbib\@empty
\bibitem [{\citenamefont {Beverland}\ \emph {et~al.}(2022)\citenamefont
  {Beverland}, \citenamefont {Murali}, \citenamefont {Troyer}, \citenamefont
  {Svore}, \citenamefont {Hoefler}, \citenamefont {Kliuchnikov}, \citenamefont
  {Low}, \citenamefont {Soeken}, \citenamefont {Sundaram},\ and\ \citenamefont
  {Vaschillo}}]{beverland2022assessing}%
  \BibitemOpen
  \bibfield  {author} {\bibinfo {author} {\bibfnamefont {M.~E.}\ \bibnamefont
  {Beverland}}, \bibinfo {author} {\bibfnamefont {P.}~\bibnamefont {Murali}},
  \bibinfo {author} {\bibfnamefont {M.}~\bibnamefont {Troyer}}, \bibinfo
  {author} {\bibfnamefont {K.~M.}\ \bibnamefont {Svore}}, \bibinfo {author}
  {\bibfnamefont {T.}~\bibnamefont {Hoefler}}, \bibinfo {author} {\bibfnamefont
  {V.}~\bibnamefont {Kliuchnikov}}, \bibinfo {author} {\bibfnamefont {G.~H.}\
  \bibnamefont {Low}}, \bibinfo {author} {\bibfnamefont {M.}~\bibnamefont
  {Soeken}}, \bibinfo {author} {\bibfnamefont {A.}~\bibnamefont {Sundaram}},\
  and\ \bibinfo {author} {\bibfnamefont {A.}~\bibnamefont {Vaschillo}},\
  }\bibfield  {title} {\bibinfo {title} {Assessing requirements to scale to
  practical quantum advantage},\ }\href
  {https://doi.org/10.48550/ARXIV.2211.07629} {\bibfield  {journal} {\bibinfo
  {journal} {arXiv preprint arXiv:2211.07629}\ } (\bibinfo {year}
  {2022})}\BibitemShut {NoStop}%
\bibitem [{\citenamefont {Preskill}(2018)}]{preskill2018quantum}%
  \BibitemOpen
  \bibfield  {author} {\bibinfo {author} {\bibfnamefont {J.}~\bibnamefont
  {Preskill}},\ }\bibfield  {title} {\bibinfo {title} {Quantum {C}omputing in
  the {NISQ} era and beyond},\ }\href
  {https://doi.org/10.22331/q-2018-08-06-79} {\bibfield  {journal} {\bibinfo
  {journal} {{Quantum}}\ }\textbf {\bibinfo {volume} {2}},\ \bibinfo {pages}
  {79} (\bibinfo {year} {2018})}\BibitemShut {NoStop}%
\bibitem [{\citenamefont {Lau}\ \emph {et~al.}(2022)\citenamefont {Lau},
  \citenamefont {Lim}, \citenamefont {Shrotriya},\ and\ \citenamefont
  {Kwek}}]{lau2022nisq}%
  \BibitemOpen
  \bibfield  {author} {\bibinfo {author} {\bibfnamefont {J.~W.~Z.}\
  \bibnamefont {Lau}}, \bibinfo {author} {\bibfnamefont {K.~H.}\ \bibnamefont
  {Lim}}, \bibinfo {author} {\bibfnamefont {H.}~\bibnamefont {Shrotriya}},\
  and\ \bibinfo {author} {\bibfnamefont {L.~C.}\ \bibnamefont {Kwek}},\
  }\bibfield  {title} {\bibinfo {title} {{NISQ} computing: where are we and
  where do we go?},\ }\href {https://doi.org/10.1007/s43673-022-00058-z}
  {\bibfield  {journal} {\bibinfo  {journal} {AAPPS Bulletin}\ }\textbf
  {\bibinfo {volume} {32}},\ \bibinfo {pages} {27} (\bibinfo {year}
  {2022})}\BibitemShut {NoStop}%
\bibitem [{\citenamefont {Cerezo}\ \emph {et~al.}(2021)\citenamefont {Cerezo},
  \citenamefont {Arrasmith}, \citenamefont {Babbush}, \citenamefont {Benjamin},
  \citenamefont {Endo}, \citenamefont {Fujii}, \citenamefont {McClean},
  \citenamefont {Mitarai}, \citenamefont {Yuan}, \citenamefont {Cincio} \emph
  {et~al.}}]{cerezo2021variational}%
  \BibitemOpen
  \bibfield  {author} {\bibinfo {author} {\bibfnamefont {M.}~\bibnamefont
  {Cerezo}}, \bibinfo {author} {\bibfnamefont {A.}~\bibnamefont {Arrasmith}},
  \bibinfo {author} {\bibfnamefont {R.}~\bibnamefont {Babbush}}, \bibinfo
  {author} {\bibfnamefont {S.~C.}\ \bibnamefont {Benjamin}}, \bibinfo {author}
  {\bibfnamefont {S.}~\bibnamefont {Endo}}, \bibinfo {author} {\bibfnamefont
  {K.}~\bibnamefont {Fujii}}, \bibinfo {author} {\bibfnamefont {J.~R.}\
  \bibnamefont {McClean}}, \bibinfo {author} {\bibfnamefont {K.}~\bibnamefont
  {Mitarai}}, \bibinfo {author} {\bibfnamefont {X.}~\bibnamefont {Yuan}},
  \bibinfo {author} {\bibfnamefont {L.}~\bibnamefont {Cincio}}, \emph
  {et~al.},\ }\bibfield  {title} {\bibinfo {title} {Variational quantum
  algorithms},\ }\href {http://dx.doi.org/10.1038/s42254-021-00348-9}
  {\bibfield  {journal} {\bibinfo  {journal} {Nature Reviews Physics}\ }\textbf
  {\bibinfo {volume} {3}},\ \bibinfo {pages} {625} (\bibinfo {year}
  {2021})}\BibitemShut {NoStop}%
\bibitem [{\citenamefont {Bharti}\ \emph {et~al.}(2022)\citenamefont {Bharti},
  \citenamefont {Cervera-Lierta}, \citenamefont {Kyaw}, \citenamefont {Haug},
  \citenamefont {Alperin-Lea}, \citenamefont {Anand}, \citenamefont {Degroote},
  \citenamefont {Heimonen}, \citenamefont {Kottmann}, \citenamefont {Menke},
  \citenamefont {Mok}, \citenamefont {Sim}, \citenamefont {Kwek},\ and\
  \citenamefont {Aspuru-Guzik}}]{RevModPhys.94.015004}%
  \BibitemOpen
  \bibfield  {author} {\bibinfo {author} {\bibfnamefont {K.}~\bibnamefont
  {Bharti}}, \bibinfo {author} {\bibfnamefont {A.}~\bibnamefont
  {Cervera-Lierta}}, \bibinfo {author} {\bibfnamefont {T.~H.}\ \bibnamefont
  {Kyaw}}, \bibinfo {author} {\bibfnamefont {T.}~\bibnamefont {Haug}}, \bibinfo
  {author} {\bibfnamefont {S.}~\bibnamefont {Alperin-Lea}}, \bibinfo {author}
  {\bibfnamefont {A.}~\bibnamefont {Anand}}, \bibinfo {author} {\bibfnamefont
  {M.}~\bibnamefont {Degroote}}, \bibinfo {author} {\bibfnamefont
  {H.}~\bibnamefont {Heimonen}}, \bibinfo {author} {\bibfnamefont {J.~S.}\
  \bibnamefont {Kottmann}}, \bibinfo {author} {\bibfnamefont {T.}~\bibnamefont
  {Menke}}, \bibinfo {author} {\bibfnamefont {W.-K.}\ \bibnamefont {Mok}},
  \bibinfo {author} {\bibfnamefont {S.}~\bibnamefont {Sim}}, \bibinfo {author}
  {\bibfnamefont {L.-C.}\ \bibnamefont {Kwek}},\ and\ \bibinfo {author}
  {\bibfnamefont {A.}~\bibnamefont {Aspuru-Guzik}},\ }\bibfield  {title}
  {\bibinfo {title} {Noisy intermediate-scale quantum algorithms},\ }\href
  {https://doi.org/10.1103/RevModPhys.94.015004} {\bibfield  {journal}
  {\bibinfo  {journal} {Rev. Mod. Phys.}\ }\textbf {\bibinfo {volume} {94}},\
  \bibinfo {pages} {015004} (\bibinfo {year} {2022})}\BibitemShut {NoStop}%
\bibitem [{\citenamefont {Tilly}\ \emph {et~al.}(2022)\citenamefont {Tilly},
  \citenamefont {Chen}, \citenamefont {Cao}, \citenamefont {Picozzi},
  \citenamefont {Setia}, \citenamefont {Li}, \citenamefont {Grant},
  \citenamefont {Wossnig}, \citenamefont {Rungger}, \citenamefont {Booth},\
  and\ \citenamefont {Tennyson}}]{tilly2022variational}%
  \BibitemOpen
  \bibfield  {author} {\bibinfo {author} {\bibfnamefont {J.}~\bibnamefont
  {Tilly}}, \bibinfo {author} {\bibfnamefont {H.}~\bibnamefont {Chen}},
  \bibinfo {author} {\bibfnamefont {S.}~\bibnamefont {Cao}}, \bibinfo {author}
  {\bibfnamefont {D.}~\bibnamefont {Picozzi}}, \bibinfo {author} {\bibfnamefont
  {K.}~\bibnamefont {Setia}}, \bibinfo {author} {\bibfnamefont
  {Y.}~\bibnamefont {Li}}, \bibinfo {author} {\bibfnamefont {E.}~\bibnamefont
  {Grant}}, \bibinfo {author} {\bibfnamefont {L.}~\bibnamefont {Wossnig}},
  \bibinfo {author} {\bibfnamefont {I.}~\bibnamefont {Rungger}}, \bibinfo
  {author} {\bibfnamefont {G.~H.}\ \bibnamefont {Booth}},\ and\ \bibinfo
  {author} {\bibfnamefont {J.}~\bibnamefont {Tennyson}},\ }\bibfield  {title}
  {\bibinfo {title} {The variational quantum eigensolver: A review of methods
  and best practices},\ }\href
  {https://doi.org/https://doi.org/10.1016/j.physrep.2022.08.003} {\bibfield
  {journal} {\bibinfo  {journal} {Physics Reports}\ }\textbf {\bibinfo {volume}
  {986}},\ \bibinfo {pages} {1} (\bibinfo {year} {2022})}\BibitemShut {NoStop}%
\bibitem [{\citenamefont {Peruzzo}\ \emph {et~al.}(2014)\citenamefont
  {Peruzzo}, \citenamefont {McClean}, \citenamefont {Shadbolt}, \citenamefont
  {Yung}, \citenamefont {Zhou}, \citenamefont {Love}, \citenamefont
  {Aspuru-Guzik},\ and\ \citenamefont {O’brien}}]{peruzzo2014variational}%
  \BibitemOpen
  \bibfield  {author} {\bibinfo {author} {\bibfnamefont {A.}~\bibnamefont
  {Peruzzo}}, \bibinfo {author} {\bibfnamefont {J.}~\bibnamefont {McClean}},
  \bibinfo {author} {\bibfnamefont {P.}~\bibnamefont {Shadbolt}}, \bibinfo
  {author} {\bibfnamefont {M.-H.}\ \bibnamefont {Yung}}, \bibinfo {author}
  {\bibfnamefont {X.-Q.}\ \bibnamefont {Zhou}}, \bibinfo {author}
  {\bibfnamefont {P.~J.}\ \bibnamefont {Love}}, \bibinfo {author}
  {\bibfnamefont {A.}~\bibnamefont {Aspuru-Guzik}},\ and\ \bibinfo {author}
  {\bibfnamefont {J.~L.}\ \bibnamefont {O’brien}},\ }\bibfield  {title}
  {\bibinfo {title} {A variational eigenvalue solver on a photonic quantum
  processor},\ }\href {https://doi.org/10.1038/ncomms5213} {\bibfield
  {journal} {\bibinfo  {journal} {Nature communications}\ }\textbf {\bibinfo
  {volume} {5}},\ \bibinfo {pages} {4213} (\bibinfo {year} {2014})}\BibitemShut
  {NoStop}%
\bibitem [{\citenamefont {Farhi}\ \emph {et~al.}(2014)\citenamefont {Farhi},
  \citenamefont {Goldstone},\ and\ \citenamefont {Gutmann}}]{farhi2014quantum}%
  \BibitemOpen
  \bibfield  {author} {\bibinfo {author} {\bibfnamefont {E.}~\bibnamefont
  {Farhi}}, \bibinfo {author} {\bibfnamefont {J.}~\bibnamefont {Goldstone}},\
  and\ \bibinfo {author} {\bibfnamefont {S.}~\bibnamefont {Gutmann}},\
  }\bibfield  {title} {\bibinfo {title} {A quantum approximate optimization
  algorithm},\ }\href {https://doi.org/10.48550/arXiv.1411.4028} {\bibfield
  {journal} {\bibinfo  {journal} {arXiv preprint arXiv:1411.4028}\ } (\bibinfo
  {year} {2014})}\BibitemShut {NoStop}%
\bibitem [{\citenamefont {Shor}(1999)}]{shor1999polynomial}%
  \BibitemOpen
  \bibfield  {author} {\bibinfo {author} {\bibfnamefont {P.~W.}\ \bibnamefont
  {Shor}},\ }\bibfield  {title} {\bibinfo {title} {Polynomial-time algorithms
  for prime factorization and discrete logarithms on a quantum computer},\
  }\href {https://doi.org/10.1137/S0036144598347011} {\bibfield  {journal}
  {\bibinfo  {journal} {SIAM review}\ }\textbf {\bibinfo {volume} {41}},\
  \bibinfo {pages} {303} (\bibinfo {year} {1999})}\BibitemShut {NoStop}%
\bibitem [{\citenamefont {Lloyd}(2018)}]{lloyd2018quantum}%
  \BibitemOpen
  \bibfield  {author} {\bibinfo {author} {\bibfnamefont {S.}~\bibnamefont
  {Lloyd}},\ }\bibfield  {title} {\bibinfo {title} {Quantum approximate
  optimization is computationally universal},\ }\href
  {https://doi.org/10.48550/arXiv.1812.11075} {\bibfield  {journal} {\bibinfo
  {journal} {arXiv preprint arXiv:1812.11075}\ } (\bibinfo {year}
  {2018})}\BibitemShut {NoStop}%
\bibitem [{\citenamefont {Morales}\ \emph {et~al.}(2020)\citenamefont
  {Morales}, \citenamefont {Biamonte},\ and\ \citenamefont
  {Zimbor{\'a}s}}]{morales2020universality}%
  \BibitemOpen
  \bibfield  {author} {\bibinfo {author} {\bibfnamefont {M.~E.}\ \bibnamefont
  {Morales}}, \bibinfo {author} {\bibfnamefont {J.~D.}\ \bibnamefont
  {Biamonte}},\ and\ \bibinfo {author} {\bibfnamefont {Z.}~\bibnamefont
  {Zimbor{\'a}s}},\ }\bibfield  {title} {\bibinfo {title} {On the universality
  of the quantum approximate optimization algorithm},\ }\href
  {https://doi.org/10.1007/s11128-020-02748-9} {\bibfield  {journal} {\bibinfo
  {journal} {Quantum Information Processing}\ }\textbf {\bibinfo {volume}
  {19}},\ \bibinfo {pages} {1} (\bibinfo {year} {2020})}\BibitemShut {NoStop}%
\bibitem [{\citenamefont {Farhi}\ and\ \citenamefont
  {Harrow}(2016)}]{farhi2016quantum}%
  \BibitemOpen
  \bibfield  {author} {\bibinfo {author} {\bibfnamefont {E.}~\bibnamefont
  {Farhi}}\ and\ \bibinfo {author} {\bibfnamefont {A.~W.}\ \bibnamefont
  {Harrow}},\ }\bibfield  {title} {\bibinfo {title} {Quantum supremacy through
  the quantum approximate optimization algorithm},\ }\href
  {https://doi.org/10.48550/arXiv.1602.07674} {\bibfield  {journal} {\bibinfo
  {journal} {arXiv preprint arXiv:1602.07674}\ } (\bibinfo {year}
  {2016})}\BibitemShut {NoStop}%
\bibitem [{\citenamefont {Dalzell}\ \emph {et~al.}(2020)\citenamefont
  {Dalzell}, \citenamefont {Harrow}, \citenamefont {Koh},\ and\ \citenamefont
  {La~Placa}}]{dalzell2020how}%
  \BibitemOpen
  \bibfield  {author} {\bibinfo {author} {\bibfnamefont {A.~M.}\ \bibnamefont
  {Dalzell}}, \bibinfo {author} {\bibfnamefont {A.~W.}\ \bibnamefont {Harrow}},
  \bibinfo {author} {\bibfnamefont {D.~E.}\ \bibnamefont {Koh}},\ and\ \bibinfo
  {author} {\bibfnamefont {R.~L.}\ \bibnamefont {La~Placa}},\ }\bibfield
  {title} {\bibinfo {title} {How many qubits are needed for quantum
  computational supremacy?},\ }\href
  {https://doi.org/10.22331/q-2020-05-11-264} {\bibfield  {journal} {\bibinfo
  {journal} {{Quantum}}\ }\textbf {\bibinfo {volume} {4}},\ \bibinfo {pages}
  {264} (\bibinfo {year} {2020})}\BibitemShut {NoStop}%
\bibitem [{\citenamefont {B{\"a}rtschi}\ and\ \citenamefont
  {Eidenbenz}(2020)}]{bartschi2020grover}%
  \BibitemOpen
  \bibfield  {author} {\bibinfo {author} {\bibfnamefont {A.}~\bibnamefont
  {B{\"a}rtschi}}\ and\ \bibinfo {author} {\bibfnamefont {S.}~\bibnamefont
  {Eidenbenz}},\ }\bibfield  {title} {\bibinfo {title} {Grover mixers for
  {QAOA}: Shifting complexity from mixer design to state preparation},\ }in\
  \href {https://doi.org/10.1109/QCE49297.2020.00020} {\emph {\bibinfo
  {booktitle} {2020 IEEE International Conference on Quantum Computing and
  Engineering (QCE)}}}\ (\bibinfo {organization} {IEEE},\ \bibinfo {year}
  {2020})\ pp.\ \bibinfo {pages} {72--82}\BibitemShut {NoStop}%
\bibitem [{\citenamefont {Hadfield}\ \emph {et~al.}(2019)\citenamefont
  {Hadfield}, \citenamefont {Wang}, \citenamefont {O’Gorman}, \citenamefont
  {Rieffel}, \citenamefont {Venturelli},\ and\ \citenamefont
  {Biswas}}]{hadfield2019quantum}%
  \BibitemOpen
  \bibfield  {author} {\bibinfo {author} {\bibfnamefont {S.}~\bibnamefont
  {Hadfield}}, \bibinfo {author} {\bibfnamefont {Z.}~\bibnamefont {Wang}},
  \bibinfo {author} {\bibfnamefont {B.}~\bibnamefont {O’Gorman}}, \bibinfo
  {author} {\bibfnamefont {E.~G.}\ \bibnamefont {Rieffel}}, \bibinfo {author}
  {\bibfnamefont {D.}~\bibnamefont {Venturelli}},\ and\ \bibinfo {author}
  {\bibfnamefont {R.}~\bibnamefont {Biswas}},\ }\bibfield  {title} {\bibinfo
  {title} {From the quantum approximate optimization algorithm to a quantum
  alternating operator ansatz},\ }\href
  {https://www.mdpi.com/1999-4893/12/2/34} {\bibfield  {journal} {\bibinfo
  {journal} {Algorithms}\ }\textbf {\bibinfo {volume} {12}} (\bibinfo {year}
  {2019})}\BibitemShut {NoStop}%
\bibitem [{\citenamefont {Wurtz}\ and\ \citenamefont
  {Love}(2021{\natexlab{a}})}]{wurtz2021classically}%
  \BibitemOpen
  \bibfield  {author} {\bibinfo {author} {\bibfnamefont {J.}~\bibnamefont
  {Wurtz}}\ and\ \bibinfo {author} {\bibfnamefont {P.~J.}\ \bibnamefont
  {Love}},\ }\bibfield  {title} {\bibinfo {title} {Classically optimal
  variational quantum algorithms},\ }\href
  {https://doi.org/10.1109/TQE.2021.3122568} {\bibfield  {journal} {\bibinfo
  {journal} {IEEE Transactions on Quantum Engineering}\ }\textbf {\bibinfo
  {volume} {2}},\ \bibinfo {pages} {1} (\bibinfo {year}
  {2021}{\natexlab{a}})}\BibitemShut {NoStop}%
\bibitem [{\citenamefont {Wang}\ \emph {et~al.}(2020)\citenamefont {Wang},
  \citenamefont {Rubin}, \citenamefont {Dominy},\ and\ \citenamefont
  {Rieffel}}]{wang2020x}%
  \BibitemOpen
  \bibfield  {author} {\bibinfo {author} {\bibfnamefont {Z.}~\bibnamefont
  {Wang}}, \bibinfo {author} {\bibfnamefont {N.~C.}\ \bibnamefont {Rubin}},
  \bibinfo {author} {\bibfnamefont {J.~M.}\ \bibnamefont {Dominy}},\ and\
  \bibinfo {author} {\bibfnamefont {E.~G.}\ \bibnamefont {Rieffel}},\
  }\bibfield  {title} {\bibinfo {title} {{$XY$} mixers: Analytical and
  numerical results for the quantum alternating operator ansatz},\ }\href
  {https://doi.org/10.1103/PhysRevA.101.012320} {\bibfield  {journal} {\bibinfo
   {journal} {Phys. Rev. A}\ }\textbf {\bibinfo {volume} {101}},\ \bibinfo
  {pages} {012320} (\bibinfo {year} {2020})}\BibitemShut {NoStop}%
\bibitem [{\citenamefont {Egger}\ \emph {et~al.}(2021)\citenamefont {Egger},
  \citenamefont {Mare{\v{c}}ek},\ and\ \citenamefont
  {Woerner}}]{egger2021warm}%
  \BibitemOpen
  \bibfield  {author} {\bibinfo {author} {\bibfnamefont {D.~J.}\ \bibnamefont
  {Egger}}, \bibinfo {author} {\bibfnamefont {J.}~\bibnamefont
  {Mare{\v{c}}ek}},\ and\ \bibinfo {author} {\bibfnamefont {S.}~\bibnamefont
  {Woerner}},\ }\bibfield  {title} {\bibinfo {title} {Warm-starting quantum
  optimization},\ }\href {https://doi.org/10.22331/q-2021-06-17-479} {\bibfield
   {journal} {\bibinfo  {journal} {{Quantum}}\ }\textbf {\bibinfo {volume}
  {5}},\ \bibinfo {pages} {479} (\bibinfo {year} {2021})}\BibitemShut {NoStop}%
\bibitem [{\citenamefont {Tate}\ \emph
  {et~al.}(2023{\natexlab{a}})\citenamefont {Tate}, \citenamefont {Farhadi},
  \citenamefont {Herold}, \citenamefont {Mohler},\ and\ \citenamefont
  {Gupta}}]{tate2020bridging}%
  \BibitemOpen
  \bibfield  {author} {\bibinfo {author} {\bibfnamefont {R.}~\bibnamefont
  {Tate}}, \bibinfo {author} {\bibfnamefont {M.}~\bibnamefont {Farhadi}},
  \bibinfo {author} {\bibfnamefont {C.}~\bibnamefont {Herold}}, \bibinfo
  {author} {\bibfnamefont {G.}~\bibnamefont {Mohler}},\ and\ \bibinfo {author}
  {\bibfnamefont {S.}~\bibnamefont {Gupta}},\ }\bibfield  {title} {\bibinfo
  {title} {Bridging classical and quantum with {SDP} initialized warm-starts
  for {QAOA}},\ }\href {https://doi.org/10.1145/3549554} {\bibfield  {journal}
  {\bibinfo  {journal} {ACM Transactions on Quantum Computing}\ }\textbf
  {\bibinfo {volume} {4}},\ \bibinfo {pages} {1} (\bibinfo {year}
  {2023}{\natexlab{a}})}\BibitemShut {NoStop}%
\bibitem [{\citenamefont {Golden}\ \emph {et~al.}(2021)\citenamefont {Golden},
  \citenamefont {B{\"a}rtschi}, \citenamefont {O’Malley},\ and\ \citenamefont
  {Eidenbenz}}]{golden2021threshold}%
  \BibitemOpen
  \bibfield  {author} {\bibinfo {author} {\bibfnamefont {J.}~\bibnamefont
  {Golden}}, \bibinfo {author} {\bibfnamefont {A.}~\bibnamefont
  {B{\"a}rtschi}}, \bibinfo {author} {\bibfnamefont {D.}~\bibnamefont
  {O’Malley}},\ and\ \bibinfo {author} {\bibfnamefont {S.}~\bibnamefont
  {Eidenbenz}},\ }\bibfield  {title} {\bibinfo {title} {Threshold-based quantum
  optimization},\ }in\ \href {https://doi.org/10.1109/QCE52317.2021.00030}
  {\emph {\bibinfo {booktitle} {2021 IEEE International Conference on Quantum
  Computing and Engineering (QCE)}}}\ (\bibinfo {organization} {IEEE},\
  \bibinfo {year} {2021})\ pp.\ \bibinfo {pages} {137--147}\BibitemShut
  {NoStop}%
\bibitem [{\citenamefont {Sack}\ and\ \citenamefont
  {Serbyn}(2021)}]{sack2021quantum}%
  \BibitemOpen
  \bibfield  {author} {\bibinfo {author} {\bibfnamefont {S.~H.}\ \bibnamefont
  {Sack}}\ and\ \bibinfo {author} {\bibfnamefont {M.}~\bibnamefont {Serbyn}},\
  }\bibfield  {title} {\bibinfo {title} {Quantum annealing initialization of
  the quantum approximate optimization algorithm},\ }\href
  {https://doi.org/10.22331/q-2021-07-01-491} {\bibfield  {journal} {\bibinfo
  {journal} {{Quantum}}\ }\textbf {\bibinfo {volume} {5}},\ \bibinfo {pages}
  {491} (\bibinfo {year} {2021})}\BibitemShut {NoStop}%
\bibitem [{\citenamefont {Chandarana}\ \emph {et~al.}(2022)\citenamefont
  {Chandarana}, \citenamefont {Hegade}, \citenamefont {Paul}, \citenamefont
  {Albarr\'an-Arriagada}, \citenamefont {Solano}, \citenamefont {del Campo},\
  and\ \citenamefont {Chen}}]{PhysRevResearch.4.013141}%
  \BibitemOpen
  \bibfield  {author} {\bibinfo {author} {\bibfnamefont {P.}~\bibnamefont
  {Chandarana}}, \bibinfo {author} {\bibfnamefont {N.~N.}\ \bibnamefont
  {Hegade}}, \bibinfo {author} {\bibfnamefont {K.}~\bibnamefont {Paul}},
  \bibinfo {author} {\bibfnamefont {F.}~\bibnamefont {Albarr\'an-Arriagada}},
  \bibinfo {author} {\bibfnamefont {E.}~\bibnamefont {Solano}}, \bibinfo
  {author} {\bibfnamefont {A.}~\bibnamefont {del Campo}},\ and\ \bibinfo
  {author} {\bibfnamefont {X.}~\bibnamefont {Chen}},\ }\bibfield  {title}
  {\bibinfo {title} {Digitized-counterdiabatic quantum approximate optimization
  algorithm},\ }\href {https://doi.org/10.1103/PhysRevResearch.4.013141}
  {\bibfield  {journal} {\bibinfo  {journal} {Phys. Rev. Research}\ }\textbf
  {\bibinfo {volume} {4}},\ \bibinfo {pages} {013141} (\bibinfo {year}
  {2022})}\BibitemShut {NoStop}%
\bibitem [{\citenamefont {Chalupnik}\ \emph {et~al.}(2022)\citenamefont
  {Chalupnik}, \citenamefont {Melo}, \citenamefont {Alexeev},\ and\
  \citenamefont {Galda}}]{chalupnik2022augmenting}%
  \BibitemOpen
  \bibfield  {author} {\bibinfo {author} {\bibfnamefont {M.}~\bibnamefont
  {Chalupnik}}, \bibinfo {author} {\bibfnamefont {H.}~\bibnamefont {Melo}},
  \bibinfo {author} {\bibfnamefont {Y.}~\bibnamefont {Alexeev}},\ and\ \bibinfo
  {author} {\bibfnamefont {A.}~\bibnamefont {Galda}},\ }\bibfield  {title}
  {\bibinfo {title} {Augmenting {QAOA} ansatz with multiparameter
  problem-independent layer},\ }in\ \href
  {https://doi.org/10.1109/QCE53715.2022.00028} {\emph {\bibinfo {booktitle}
  {2022 IEEE International Conference on Quantum Computing and Engineering
  (QCE)}}}\ (\bibinfo {year} {2022})\ pp.\ \bibinfo {pages}
  {97--103}\BibitemShut {NoStop}%
\bibitem [{\citenamefont {Golden}\ \emph {et~al.}(2023)\citenamefont {Golden},
  \citenamefont {Bärtschi}, \citenamefont {O'Malley},\ and\ \citenamefont
  {Eidenbenz}}]{golden2023quantum}%
  \BibitemOpen
  \bibfield  {author} {\bibinfo {author} {\bibfnamefont {J.}~\bibnamefont
  {Golden}}, \bibinfo {author} {\bibfnamefont {A.}~\bibnamefont {Bärtschi}},
  \bibinfo {author} {\bibfnamefont {D.}~\bibnamefont {O'Malley}},\ and\
  \bibinfo {author} {\bibfnamefont {S.}~\bibnamefont {Eidenbenz}},\ }\bibfield
  {title} {\bibinfo {title} {The quantum alternating operator ansatz for
  satisfiability problems},\ }in\ \href
  {https://doi.org/10.1109/QCE57702.2023.00042} {\emph {\bibinfo {booktitle}
  {2023 IEEE International Conference on Quantum Computing and Engineering
  (QCE)}}},\ Vol.~\bibinfo {volume} {01}\ (\bibinfo {year} {2023})\ pp.\
  \bibinfo {pages} {307--312}\BibitemShut {NoStop}%
\bibitem [{\citenamefont {Lee}\ \emph {et~al.}(2023)\citenamefont {Lee},
  \citenamefont {Xie}, \citenamefont {Cai}, \citenamefont {Saito},\ and\
  \citenamefont {Asai}}]{lee2022depth}%
  \BibitemOpen
  \bibfield  {author} {\bibinfo {author} {\bibfnamefont {X.}~\bibnamefont
  {Lee}}, \bibinfo {author} {\bibfnamefont {N.}~\bibnamefont {Xie}}, \bibinfo
  {author} {\bibfnamefont {D.}~\bibnamefont {Cai}}, \bibinfo {author}
  {\bibfnamefont {Y.}~\bibnamefont {Saito}},\ and\ \bibinfo {author}
  {\bibfnamefont {N.}~\bibnamefont {Asai}},\ }\bibfield  {title} {\bibinfo
  {title} {A depth-progressive initialization strategy for quantum approximate
  optimization algorithm},\ }\href {https://www.mdpi.com/2227-7390/11/9/2176}
  {\bibfield  {journal} {\bibinfo  {journal} {Mathematics}\ }\textbf {\bibinfo
  {volume} {11}},\ \bibinfo {pages} {2176} (\bibinfo {year}
  {2023})}\BibitemShut {NoStop}%
\bibitem [{\citenamefont {Leontica}\ and\ \citenamefont
  {Amaro}(2024)}]{leontica2022quantum}%
  \BibitemOpen
  \bibfield  {author} {\bibinfo {author} {\bibfnamefont {S.}~\bibnamefont
  {Leontica}}\ and\ \bibinfo {author} {\bibfnamefont {D.}~\bibnamefont
  {Amaro}},\ }\bibfield  {title} {\bibinfo {title} {Exploring the neighborhood
  of 1-layer {QAOA} with instantaneous quantum polynomial circuits},\ }\href
  {https://doi.org/10.1103/PhysRevResearch.6.013071} {\bibfield  {journal}
  {\bibinfo  {journal} {Phys. Rev. Res.}\ }\textbf {\bibinfo {volume} {6}},\
  \bibinfo {pages} {013071} (\bibinfo {year} {2024})}\BibitemShut {NoStop}%
\bibitem [{\citenamefont {Barkoutsos}\ \emph {et~al.}(2020)\citenamefont
  {Barkoutsos}, \citenamefont {Nannicini}, \citenamefont {Robert},
  \citenamefont {Tavernelli},\ and\ \citenamefont
  {Woerner}}]{barkoutsos2020improving}%
  \BibitemOpen
  \bibfield  {author} {\bibinfo {author} {\bibfnamefont {P.~K.}\ \bibnamefont
  {Barkoutsos}}, \bibinfo {author} {\bibfnamefont {G.}~\bibnamefont
  {Nannicini}}, \bibinfo {author} {\bibfnamefont {A.}~\bibnamefont {Robert}},
  \bibinfo {author} {\bibfnamefont {I.}~\bibnamefont {Tavernelli}},\ and\
  \bibinfo {author} {\bibfnamefont {S.}~\bibnamefont {Woerner}},\ }\bibfield
  {title} {\bibinfo {title} {Improving {V}ariational {Q}uantum {O}ptimization
  using {CV}a{R}},\ }\href {https://doi.org/10.22331/q-2020-04-20-256}
  {\bibfield  {journal} {\bibinfo  {journal} {{Quantum}}\ }\textbf {\bibinfo
  {volume} {4}},\ \bibinfo {pages} {256} (\bibinfo {year} {2020})}\BibitemShut
  {NoStop}%
\bibitem [{\citenamefont {Li}\ \emph {et~al.}(2020)\citenamefont {Li},
  \citenamefont {Fan}, \citenamefont {Coram}, \citenamefont {Riley},\ and\
  \citenamefont {Leichenauer}}]{li2020quantum}%
  \BibitemOpen
  \bibfield  {author} {\bibinfo {author} {\bibfnamefont {L.}~\bibnamefont
  {Li}}, \bibinfo {author} {\bibfnamefont {M.}~\bibnamefont {Fan}}, \bibinfo
  {author} {\bibfnamefont {M.}~\bibnamefont {Coram}}, \bibinfo {author}
  {\bibfnamefont {P.}~\bibnamefont {Riley}},\ and\ \bibinfo {author}
  {\bibfnamefont {S.}~\bibnamefont {Leichenauer}},\ }\bibfield  {title}
  {\bibinfo {title} {Quantum optimization with a novel {G}ibbs objective
  function and ansatz architecture search},\ }\href
  {https://doi.org/10.1103/PhysRevResearch.2.023074} {\bibfield  {journal}
  {\bibinfo  {journal} {Phys. Rev. Res.}\ }\textbf {\bibinfo {volume} {2}},\
  \bibinfo {pages} {023074} (\bibinfo {year} {2020})}\BibitemShut {NoStop}%
\bibitem [{\citenamefont {Majumdar}\ \emph
  {et~al.}(2021{\natexlab{a}})\citenamefont {Majumdar}, \citenamefont {Madan},
  \citenamefont {Bhoumik}, \citenamefont {Vinayagamurthy}, \citenamefont
  {Raghunathan},\ and\ \citenamefont {Sur-Kolay}}]{majumdar2021optimizing}%
  \BibitemOpen
  \bibfield  {author} {\bibinfo {author} {\bibfnamefont {R.}~\bibnamefont
  {Majumdar}}, \bibinfo {author} {\bibfnamefont {D.}~\bibnamefont {Madan}},
  \bibinfo {author} {\bibfnamefont {D.}~\bibnamefont {Bhoumik}}, \bibinfo
  {author} {\bibfnamefont {D.}~\bibnamefont {Vinayagamurthy}}, \bibinfo
  {author} {\bibfnamefont {S.}~\bibnamefont {Raghunathan}},\ and\ \bibinfo
  {author} {\bibfnamefont {S.}~\bibnamefont {Sur-Kolay}},\ }\bibfield  {title}
  {\bibinfo {title} {Optimizing ansatz design in {QAOA} for {Max-cut}},\ }\href
  {https://doi.org/10.48550/arXiv.2106.02812} {\bibfield  {journal} {\bibinfo
  {journal} {arXiv preprint arXiv:2106.02812}\ } (\bibinfo {year}
  {2021}{\natexlab{a}})}\BibitemShut {NoStop}%
\bibitem [{\citenamefont {Majumdar}\ \emph
  {et~al.}(2021{\natexlab{b}})\citenamefont {Majumdar}, \citenamefont
  {Bhoumik}, \citenamefont {Madan}, \citenamefont {Vinayagamurthy},
  \citenamefont {Raghunathan},\ and\ \citenamefont
  {Sur-Kolay}}]{majumdar2021depth}%
  \BibitemOpen
  \bibfield  {author} {\bibinfo {author} {\bibfnamefont {R.}~\bibnamefont
  {Majumdar}}, \bibinfo {author} {\bibfnamefont {D.}~\bibnamefont {Bhoumik}},
  \bibinfo {author} {\bibfnamefont {D.}~\bibnamefont {Madan}}, \bibinfo
  {author} {\bibfnamefont {D.}~\bibnamefont {Vinayagamurthy}}, \bibinfo
  {author} {\bibfnamefont {S.}~\bibnamefont {Raghunathan}},\ and\ \bibinfo
  {author} {\bibfnamefont {S.}~\bibnamefont {Sur-Kolay}},\ }\bibfield  {title}
  {\bibinfo {title} {Depth optimized ansatz circuit in {QAOA} for {Max-Cut}},\
  }\href {https://doi.org/10.48550/arXiv.2110.04637} {\bibfield  {journal}
  {\bibinfo  {journal} {arXiv preprint arXiv:2110.04637}\ } (\bibinfo {year}
  {2021}{\natexlab{b}})}\BibitemShut {NoStop}%
\bibitem [{\citenamefont {Bechtold}\ \emph {et~al.}(2023)\citenamefont
  {Bechtold}, \citenamefont {Barzen}, \citenamefont {Leymann}, \citenamefont
  {Mandl}, \citenamefont {Obst}, \citenamefont {Truger},\ and\ \citenamefont
  {Weder}}]{bechtold202investigating}%
  \BibitemOpen
  \bibfield  {author} {\bibinfo {author} {\bibfnamefont {M.}~\bibnamefont
  {Bechtold}}, \bibinfo {author} {\bibfnamefont {J.}~\bibnamefont {Barzen}},
  \bibinfo {author} {\bibfnamefont {F.}~\bibnamefont {Leymann}}, \bibinfo
  {author} {\bibfnamefont {A.}~\bibnamefont {Mandl}}, \bibinfo {author}
  {\bibfnamefont {J.}~\bibnamefont {Obst}}, \bibinfo {author} {\bibfnamefont
  {F.}~\bibnamefont {Truger}},\ and\ \bibinfo {author} {\bibfnamefont
  {B.}~\bibnamefont {Weder}},\ }\bibfield  {title} {\bibinfo {title}
  {Investigating the effect of circuit cutting in {QAOA} for the {MaxCut}
  problem on {NISQ} devices},\ }\href
  {https://doi.org/10.1088/2058-9565/acf59c} {\bibfield  {journal} {\bibinfo
  {journal} {Quantum Science and Technology}\ }\textbf {\bibinfo {volume}
  {8}},\ \bibinfo {pages} {045022} (\bibinfo {year} {2023})}\BibitemShut
  {NoStop}%
\bibitem [{\citenamefont {Peng}\ \emph {et~al.}(2020)\citenamefont {Peng},
  \citenamefont {Harrow}, \citenamefont {Ozols},\ and\ \citenamefont
  {Wu}}]{peng2020simulating}%
  \BibitemOpen
  \bibfield  {author} {\bibinfo {author} {\bibfnamefont {T.}~\bibnamefont
  {Peng}}, \bibinfo {author} {\bibfnamefont {A.~W.}\ \bibnamefont {Harrow}},
  \bibinfo {author} {\bibfnamefont {M.}~\bibnamefont {Ozols}},\ and\ \bibinfo
  {author} {\bibfnamefont {X.}~\bibnamefont {Wu}},\ }\bibfield  {title}
  {\bibinfo {title} {Simulating large quantum circuits on a small quantum
  computer},\ }\href {https://doi.org/10.1103/PhysRevLett.125.150504}
  {\bibfield  {journal} {\bibinfo  {journal} {Phys. Rev. Lett.}\ }\textbf
  {\bibinfo {volume} {125}},\ \bibinfo {pages} {150504} (\bibinfo {year}
  {2020})}\BibitemShut {NoStop}%
\bibitem [{\citenamefont {Herrman}\ \emph
  {et~al.}(2021{\natexlab{a}})\citenamefont {Herrman}, \citenamefont
  {Treffert}, \citenamefont {Ostrowski}, \citenamefont {Lotshaw}, \citenamefont
  {Humble},\ and\ \citenamefont {Siopsis}}]{herrman2021impact}%
  \BibitemOpen
  \bibfield  {author} {\bibinfo {author} {\bibfnamefont {R.}~\bibnamefont
  {Herrman}}, \bibinfo {author} {\bibfnamefont {L.}~\bibnamefont {Treffert}},
  \bibinfo {author} {\bibfnamefont {J.}~\bibnamefont {Ostrowski}}, \bibinfo
  {author} {\bibfnamefont {P.~C.}\ \bibnamefont {Lotshaw}}, \bibinfo {author}
  {\bibfnamefont {T.~S.}\ \bibnamefont {Humble}},\ and\ \bibinfo {author}
  {\bibfnamefont {G.}~\bibnamefont {Siopsis}},\ }\bibfield  {title} {\bibinfo
  {title} {Impact of graph structures for {QAOA} on {M}ax{C}ut},\ }\href
  {https://doi.org/10.1007/s11128-021-03232-8} {\bibfield  {journal} {\bibinfo
  {journal} {Quantum Information Processing}\ }\textbf {\bibinfo {volume}
  {20}},\ \bibinfo {pages} {1} (\bibinfo {year}
  {2021}{\natexlab{a}})}\BibitemShut {NoStop}%
\bibitem [{\citenamefont {Shaydulin}\ \emph {et~al.}(2021)\citenamefont
  {Shaydulin}, \citenamefont {Hadfield}, \citenamefont {Hogg},\ and\
  \citenamefont {Safro}}]{shaydulin2021classical}%
  \BibitemOpen
  \bibfield  {author} {\bibinfo {author} {\bibfnamefont {R.}~\bibnamefont
  {Shaydulin}}, \bibinfo {author} {\bibfnamefont {S.}~\bibnamefont {Hadfield}},
  \bibinfo {author} {\bibfnamefont {T.}~\bibnamefont {Hogg}},\ and\ \bibinfo
  {author} {\bibfnamefont {I.}~\bibnamefont {Safro}},\ }\bibfield  {title}
  {\bibinfo {title} {Classical symmetries and the quantum approximate
  optimization algorithm},\ }\href {https://doi.org/10.1007/s11128-021-03298-4}
  {\bibfield  {journal} {\bibinfo  {journal} {Quantum Information Processing}\
  }\textbf {\bibinfo {volume} {20}},\ \bibinfo {pages} {359} (\bibinfo {year}
  {2021})}\BibitemShut {NoStop}%
\bibitem [{\citenamefont {Jain}\ \emph {et~al.}(2022)\citenamefont {Jain},
  \citenamefont {Coyle}, \citenamefont {Kashefi},\ and\ \citenamefont
  {Kumar}}]{jain2022graph}%
  \BibitemOpen
  \bibfield  {author} {\bibinfo {author} {\bibfnamefont {N.}~\bibnamefont
  {Jain}}, \bibinfo {author} {\bibfnamefont {B.}~\bibnamefont {Coyle}},
  \bibinfo {author} {\bibfnamefont {E.}~\bibnamefont {Kashefi}},\ and\ \bibinfo
  {author} {\bibfnamefont {N.}~\bibnamefont {Kumar}},\ }\bibfield  {title}
  {\bibinfo {title} {Graph neural network initialisation of quantum approximate
  optimisation},\ }\href {https://doi.org/10.22331/q-2022-11-17-861} {\bibfield
   {journal} {\bibinfo  {journal} {{Quantum}}\ }\textbf {\bibinfo {volume}
  {6}},\ \bibinfo {pages} {861} (\bibinfo {year} {2022})}\BibitemShut {NoStop}%
\bibitem [{\citenamefont {Streif}\ and\ \citenamefont
  {Leib}(2020)}]{streif2020training}%
  \BibitemOpen
  \bibfield  {author} {\bibinfo {author} {\bibfnamefont {M.}~\bibnamefont
  {Streif}}\ and\ \bibinfo {author} {\bibfnamefont {M.}~\bibnamefont {Leib}},\
  }\bibfield  {title} {\bibinfo {title} {Training the quantum approximate
  optimization algorithm without access to a quantum processing unit},\ }\href
  {https://doi.org/10.1088/2058-9565/ab8c2b} {\bibfield  {journal} {\bibinfo
  {journal} {Quantum Science and Technology}\ }\textbf {\bibinfo {volume}
  {5}},\ \bibinfo {pages} {034008} (\bibinfo {year} {2020})}\BibitemShut
  {NoStop}%
\bibitem [{\citenamefont {Akshay}\ \emph
  {et~al.}(2021{\natexlab{a}})\citenamefont {Akshay}, \citenamefont
  {Rabinovich}, \citenamefont {Campos},\ and\ \citenamefont
  {Biamonte}}]{akshay2021parameter}%
  \BibitemOpen
  \bibfield  {author} {\bibinfo {author} {\bibfnamefont {V.}~\bibnamefont
  {Akshay}}, \bibinfo {author} {\bibfnamefont {D.}~\bibnamefont {Rabinovich}},
  \bibinfo {author} {\bibfnamefont {E.}~\bibnamefont {Campos}},\ and\ \bibinfo
  {author} {\bibfnamefont {J.}~\bibnamefont {Biamonte}},\ }\bibfield  {title}
  {\bibinfo {title} {Parameter concentrations in quantum approximate
  optimization},\ }\href {https://doi.org/10.1103/PhysRevA.104.L010401}
  {\bibfield  {journal} {\bibinfo  {journal} {Phys. Rev. A}\ }\textbf {\bibinfo
  {volume} {104}},\ \bibinfo {pages} {L010401} (\bibinfo {year}
  {2021}{\natexlab{a}})}\BibitemShut {NoStop}%
\bibitem [{\citenamefont {Wurtz}\ and\ \citenamefont
  {Love}(2022)}]{wurtz2022counterdiabaticity}%
  \BibitemOpen
  \bibfield  {author} {\bibinfo {author} {\bibfnamefont {J.}~\bibnamefont
  {Wurtz}}\ and\ \bibinfo {author} {\bibfnamefont {P.~J.}\ \bibnamefont
  {Love}},\ }\bibfield  {title} {\bibinfo {title} {Counterdiabaticity and the
  quantum approximate optimization algorithm},\ }\href
  {https://doi.org/10.22331/q-2022-01-27-635} {\bibfield  {journal} {\bibinfo
  {journal} {{Quantum}}\ }\textbf {\bibinfo {volume} {6}},\ \bibinfo {pages}
  {635} (\bibinfo {year} {2022})}\BibitemShut {NoStop}%
\bibitem [{\citenamefont {Bravyi}\ \emph {et~al.}(2020)\citenamefont {Bravyi},
  \citenamefont {Kliesch}, \citenamefont {Koenig},\ and\ \citenamefont
  {Tang}}]{bravyi2020obstacles}%
  \BibitemOpen
  \bibfield  {author} {\bibinfo {author} {\bibfnamefont {S.}~\bibnamefont
  {Bravyi}}, \bibinfo {author} {\bibfnamefont {A.}~\bibnamefont {Kliesch}},
  \bibinfo {author} {\bibfnamefont {R.}~\bibnamefont {Koenig}},\ and\ \bibinfo
  {author} {\bibfnamefont {E.}~\bibnamefont {Tang}},\ }\bibfield  {title}
  {\bibinfo {title} {Obstacles to variational quantum optimization from
  symmetry protection},\ }\href
  {https://doi.org/10.1103/PhysRevLett.125.260505} {\bibfield  {journal}
  {\bibinfo  {journal} {Phys. Rev. Lett.}\ }\textbf {\bibinfo {volume} {125}},\
  \bibinfo {pages} {260505} (\bibinfo {year} {2020})}\BibitemShut {NoStop}%
\bibitem [{\citenamefont {Bravyi}\ \emph {et~al.}(2022)\citenamefont {Bravyi},
  \citenamefont {Kliesch}, \citenamefont {Koenig},\ and\ \citenamefont
  {Tang}}]{bravyi2022hybrid}%
  \BibitemOpen
  \bibfield  {author} {\bibinfo {author} {\bibfnamefont {S.}~\bibnamefont
  {Bravyi}}, \bibinfo {author} {\bibfnamefont {A.}~\bibnamefont {Kliesch}},
  \bibinfo {author} {\bibfnamefont {R.}~\bibnamefont {Koenig}},\ and\ \bibinfo
  {author} {\bibfnamefont {E.}~\bibnamefont {Tang}},\ }\bibfield  {title}
  {\bibinfo {title} {Hybrid quantum-classical algorithms for approximate graph
  coloring},\ }\href {https://doi.org/10.22331/q-2022-03-30-678} {\bibfield
  {journal} {\bibinfo  {journal} {{Quantum}}\ }\textbf {\bibinfo {volume}
  {6}},\ \bibinfo {pages} {678} (\bibinfo {year} {2022})}\BibitemShut {NoStop}%
\bibitem [{\citenamefont {Guerreschi}\ and\ \citenamefont
  {Matsuura}(2019)}]{guerreschi2019qaoa}%
  \BibitemOpen
  \bibfield  {author} {\bibinfo {author} {\bibfnamefont {G.~G.}\ \bibnamefont
  {Guerreschi}}\ and\ \bibinfo {author} {\bibfnamefont {A.~Y.}\ \bibnamefont
  {Matsuura}},\ }\bibfield  {title} {\bibinfo {title} {{QAOA} for max-cut
  requires hundreds of qubits for quantum speed-up},\ }\href
  {https://doi.org/10.1038/s41598-019-43176-9} {\bibfield  {journal} {\bibinfo
  {journal} {Scientific reports}\ }\textbf {\bibinfo {volume} {9}},\ \bibinfo
  {pages} {6903} (\bibinfo {year} {2019})}\BibitemShut {NoStop}%
\bibitem [{\citenamefont {Herrman}\ \emph
  {et~al.}(2021{\natexlab{b}})\citenamefont {Herrman}, \citenamefont
  {Ostrowski}, \citenamefont {Humble},\ and\ \citenamefont
  {Siopsis}}]{herrman2021lower}%
  \BibitemOpen
  \bibfield  {author} {\bibinfo {author} {\bibfnamefont {R.}~\bibnamefont
  {Herrman}}, \bibinfo {author} {\bibfnamefont {J.}~\bibnamefont {Ostrowski}},
  \bibinfo {author} {\bibfnamefont {T.~S.}\ \bibnamefont {Humble}},\ and\
  \bibinfo {author} {\bibfnamefont {G.}~\bibnamefont {Siopsis}},\ }\bibfield
  {title} {\bibinfo {title} {Lower bounds on circuit depth of the quantum
  approximate optimization algorithm},\ }\href
  {https://doi.org/10.1007/s11128-021-03001-7} {\bibfield  {journal} {\bibinfo
  {journal} {Quantum Information Processing}\ }\textbf {\bibinfo {volume}
  {20}},\ \bibinfo {pages} {59} (\bibinfo {year}
  {2021}{\natexlab{b}})}\BibitemShut {NoStop}%
\bibitem [{\citenamefont {Wurtz}\ and\ \citenamefont
  {Lykov}(2021)}]{wurtz2021fixed}%
  \BibitemOpen
  \bibfield  {author} {\bibinfo {author} {\bibfnamefont {J.}~\bibnamefont
  {Wurtz}}\ and\ \bibinfo {author} {\bibfnamefont {D.}~\bibnamefont {Lykov}},\
  }\bibfield  {title} {\bibinfo {title} {Fixed-angle conjectures for the
  quantum approximate optimization algorithm on regular maxcut graphs},\ }\href
  {https://doi.org/10.1103/PhysRevA.104.052419} {\bibfield  {journal} {\bibinfo
   {journal} {Phys. Rev. A}\ }\textbf {\bibinfo {volume} {104}},\ \bibinfo
  {pages} {052419} (\bibinfo {year} {2021})}\BibitemShut {NoStop}%
\bibitem [{\citenamefont {Akshay}\ \emph {et~al.}(2020)\citenamefont {Akshay},
  \citenamefont {Philathong}, \citenamefont {Morales},\ and\ \citenamefont
  {Biamonte}}]{akshay2020reachability}%
  \BibitemOpen
  \bibfield  {author} {\bibinfo {author} {\bibfnamefont {V.}~\bibnamefont
  {Akshay}}, \bibinfo {author} {\bibfnamefont {H.}~\bibnamefont {Philathong}},
  \bibinfo {author} {\bibfnamefont {M.~E.~S.}\ \bibnamefont {Morales}},\ and\
  \bibinfo {author} {\bibfnamefont {J.~D.}\ \bibnamefont {Biamonte}},\
  }\bibfield  {title} {\bibinfo {title} {Reachability deficits in quantum
  approximate optimization},\ }\href
  {https://doi.org/10.1103/PhysRevLett.124.090504} {\bibfield  {journal}
  {\bibinfo  {journal} {Phys. Rev. Lett.}\ }\textbf {\bibinfo {volume} {124}},\
  \bibinfo {pages} {090504} (\bibinfo {year} {2020})}\BibitemShut {NoStop}%
\bibitem [{\citenamefont {Wurtz}\ and\ \citenamefont
  {Love}(2021{\natexlab{b}})}]{wurtz2021maxcut}%
  \BibitemOpen
  \bibfield  {author} {\bibinfo {author} {\bibfnamefont {J.}~\bibnamefont
  {Wurtz}}\ and\ \bibinfo {author} {\bibfnamefont {P.}~\bibnamefont {Love}},\
  }\bibfield  {title} {\bibinfo {title} {{M}ax{C}ut quantum approximate
  optimization algorithm performance guarantees for $p>1$},\ }\href
  {https://doi.org/10.1103/PhysRevA.103.042612} {\bibfield  {journal} {\bibinfo
   {journal} {Phys. Rev. A}\ }\textbf {\bibinfo {volume} {103}},\ \bibinfo
  {pages} {042612} (\bibinfo {year} {2021}{\natexlab{b}})}\BibitemShut
  {NoStop}%
\bibitem [{\citenamefont {Farhi}\ \emph {et~al.}(2020)\citenamefont {Farhi},
  \citenamefont {Gamarnik},\ and\ \citenamefont {Gutmann}}]{farhi2020quantum}%
  \BibitemOpen
  \bibfield  {author} {\bibinfo {author} {\bibfnamefont {E.}~\bibnamefont
  {Farhi}}, \bibinfo {author} {\bibfnamefont {D.}~\bibnamefont {Gamarnik}},\
  and\ \bibinfo {author} {\bibfnamefont {S.}~\bibnamefont {Gutmann}},\
  }\bibfield  {title} {\bibinfo {title} {The quantum approximate optimization
  algorithm needs to see the whole graph: Worst case examples},\ }\href
  {https://doi.org/10.48550/arXiv.2005.08747} {\bibfield  {journal} {\bibinfo
  {journal} {arXiv preprint arXiv:2005.08747}\ } (\bibinfo {year}
  {2020})}\BibitemShut {NoStop}%
\bibitem [{\citenamefont {Xue}\ \emph {et~al.}(2021)\citenamefont {Xue},
  \citenamefont {Chen}, \citenamefont {Wu},\ and\ \citenamefont
  {Guo}}]{xue2021effects}%
  \BibitemOpen
  \bibfield  {author} {\bibinfo {author} {\bibfnamefont {C.}~\bibnamefont
  {Xue}}, \bibinfo {author} {\bibfnamefont {Z.-Y.}\ \bibnamefont {Chen}},
  \bibinfo {author} {\bibfnamefont {Y.-C.}\ \bibnamefont {Wu}},\ and\ \bibinfo
  {author} {\bibfnamefont {G.-P.}\ \bibnamefont {Guo}},\ }\bibfield  {title}
  {\bibinfo {title} {Effects of quantum noise on quantum approximate
  optimization algorithm},\ }\href
  {https://doi.org/10.1088/0256-307X/38/3/030302} {\bibfield  {journal}
  {\bibinfo  {journal} {Chinese Physics Letters}\ }\textbf {\bibinfo {volume}
  {38}},\ \bibinfo {pages} {030302} (\bibinfo {year} {2021})}\BibitemShut
  {NoStop}%
\bibitem [{\citenamefont {Wang}\ \emph {et~al.}(2021)\citenamefont {Wang},
  \citenamefont {Fontana}, \citenamefont {Cerezo}, \citenamefont {Sharma},
  \citenamefont {Sone}, \citenamefont {Cincio},\ and\ \citenamefont
  {Coles}}]{wang2021noise}%
  \BibitemOpen
  \bibfield  {author} {\bibinfo {author} {\bibfnamefont {S.}~\bibnamefont
  {Wang}}, \bibinfo {author} {\bibfnamefont {E.}~\bibnamefont {Fontana}},
  \bibinfo {author} {\bibfnamefont {M.}~\bibnamefont {Cerezo}}, \bibinfo
  {author} {\bibfnamefont {K.}~\bibnamefont {Sharma}}, \bibinfo {author}
  {\bibfnamefont {A.}~\bibnamefont {Sone}}, \bibinfo {author} {\bibfnamefont
  {L.}~\bibnamefont {Cincio}},\ and\ \bibinfo {author} {\bibfnamefont {P.~J.}\
  \bibnamefont {Coles}},\ }\bibfield  {title} {\bibinfo {title} {Noise-induced
  barren plateaus in variational quantum algorithms},\ }\href
  {https://doi.org/10.1038/s41467-021-27045-6} {\bibfield  {journal} {\bibinfo
  {journal} {Nature communications}\ }\textbf {\bibinfo {volume} {12}},\
  \bibinfo {pages} {1} (\bibinfo {year} {2021})}\BibitemShut {NoStop}%
\bibitem [{\citenamefont {Marshall}\ \emph {et~al.}(2020)\citenamefont
  {Marshall}, \citenamefont {Wudarski}, \citenamefont {Hadfield},\ and\
  \citenamefont {Hogg}}]{marshall2020characterizing}%
  \BibitemOpen
  \bibfield  {author} {\bibinfo {author} {\bibfnamefont {J.}~\bibnamefont
  {Marshall}}, \bibinfo {author} {\bibfnamefont {F.}~\bibnamefont {Wudarski}},
  \bibinfo {author} {\bibfnamefont {S.}~\bibnamefont {Hadfield}},\ and\
  \bibinfo {author} {\bibfnamefont {T.}~\bibnamefont {Hogg}},\ }\bibfield
  {title} {\bibinfo {title} {Characterizing local noise in {QAOA} circuits},\
  }\href {https://doi.org/https://doi.org/10.1088/2633-1357/abb0d7} {\bibfield
  {journal} {\bibinfo  {journal} {IOP SciNotes}\ }\textbf {\bibinfo {volume}
  {1}},\ \bibinfo {pages} {025208} (\bibinfo {year} {2020})}\BibitemShut
  {NoStop}%
\bibitem [{\citenamefont {Alam}\ \emph {et~al.}(2019)\citenamefont {Alam},
  \citenamefont {Ash-Saki},\ and\ \citenamefont {Ghosh}}]{alam2019analysis}%
  \BibitemOpen
  \bibfield  {author} {\bibinfo {author} {\bibfnamefont {M.}~\bibnamefont
  {Alam}}, \bibinfo {author} {\bibfnamefont {A.}~\bibnamefont {Ash-Saki}},\
  and\ \bibinfo {author} {\bibfnamefont {S.}~\bibnamefont {Ghosh}},\ }\bibfield
   {title} {\bibinfo {title} {Analysis of quantum approximate optimization
  algorithm under realistic noise in superconducting qubits},\ }\href
  {https://doi.org/10.48550/arXiv.1907.09631} {\bibfield  {journal} {\bibinfo
  {journal} {arXiv preprint arXiv:1907.09631}\ } (\bibinfo {year}
  {2019})}\BibitemShut {NoStop}%
\bibitem [{\citenamefont {Alam}\ \emph {et~al.}(2020)\citenamefont {Alam},
  \citenamefont {Ash-Saki},\ and\ \citenamefont {Ghosh}}]{alam2020design}%
  \BibitemOpen
  \bibfield  {author} {\bibinfo {author} {\bibfnamefont {M.}~\bibnamefont
  {Alam}}, \bibinfo {author} {\bibfnamefont {A.}~\bibnamefont {Ash-Saki}},\
  and\ \bibinfo {author} {\bibfnamefont {S.}~\bibnamefont {Ghosh}},\ }\bibfield
   {title} {\bibinfo {title} {Design-space exploration of quantum approximate
  optimization algorithm under noise},\ }in\ \href
  {https://doi.org/10.1109/CICC48029.2020.9075903} {\emph {\bibinfo {booktitle}
  {2020 IEEE Custom Integrated Circuits Conference (CICC)}}}\ (\bibinfo
  {organization} {IEEE},\ \bibinfo {year} {2020})\ pp.\ \bibinfo {pages}
  {1--4}\BibitemShut {NoStop}%
\bibitem [{\citenamefont {Streif}\ \emph {et~al.}(2021)\citenamefont {Streif},
  \citenamefont {Leib}, \citenamefont {Wudarski}, \citenamefont {Rieffel},\
  and\ \citenamefont {Wang}}]{streif2021quantum}%
  \BibitemOpen
  \bibfield  {author} {\bibinfo {author} {\bibfnamefont {M.}~\bibnamefont
  {Streif}}, \bibinfo {author} {\bibfnamefont {M.}~\bibnamefont {Leib}},
  \bibinfo {author} {\bibfnamefont {F.}~\bibnamefont {Wudarski}}, \bibinfo
  {author} {\bibfnamefont {E.}~\bibnamefont {Rieffel}},\ and\ \bibinfo {author}
  {\bibfnamefont {Z.}~\bibnamefont {Wang}},\ }\bibfield  {title} {\bibinfo
  {title} {Quantum algorithms with local particle-number conservation: Noise
  effects and error correction},\ }\href
  {https://doi.org/10.1103/PhysRevA.103.042412} {\bibfield  {journal} {\bibinfo
   {journal} {Phys. Rev. A}\ }\textbf {\bibinfo {volume} {103}},\ \bibinfo
  {pages} {042412} (\bibinfo {year} {2021})}\BibitemShut {NoStop}%
\bibitem [{\citenamefont {Anschuetz}\ and\ \citenamefont
  {Kiani}(2022)}]{anschuetz2022beyond}%
  \BibitemOpen
  \bibfield  {author} {\bibinfo {author} {\bibfnamefont {E.~R.}\ \bibnamefont
  {Anschuetz}}\ and\ \bibinfo {author} {\bibfnamefont {B.~T.}\ \bibnamefont
  {Kiani}},\ }\bibfield  {title} {\bibinfo {title} {Quantum variational
  algorithms are swamped with traps},\ }\href
  {https://doi.org/10.1038/s41467-022-35364-5} {\bibfield  {journal} {\bibinfo
  {journal} {Nature Communications}\ }\textbf {\bibinfo {volume} {13}},\
  \bibinfo {pages} {7760} (\bibinfo {year} {2022})}\BibitemShut {NoStop}%
\bibitem [{\citenamefont {Stilck~Fran{\c{c}}a}\ and\ \citenamefont
  {Garcia-Patron}(2021)}]{stilck2021limitations}%
  \BibitemOpen
  \bibfield  {author} {\bibinfo {author} {\bibfnamefont {D.}~\bibnamefont
  {Stilck~Fran{\c{c}}a}}\ and\ \bibinfo {author} {\bibfnamefont
  {R.}~\bibnamefont {Garcia-Patron}},\ }\bibfield  {title} {\bibinfo {title}
  {Limitations of optimization algorithms on noisy quantum devices},\ }\href
  {https://doi.org/10.1038/s41567-021-01356-3} {\bibfield  {journal} {\bibinfo
  {journal} {Nature Physics}\ }\textbf {\bibinfo {volume} {17}},\ \bibinfo
  {pages} {1221} (\bibinfo {year} {2021})}\BibitemShut {NoStop}%
\bibitem [{\citenamefont {Weidinger}\ \emph {et~al.}(2023)\citenamefont
  {Weidinger}, \citenamefont {Mbeng},\ and\ \citenamefont
  {Lechner}}]{weidinger2023error}%
  \BibitemOpen
  \bibfield  {author} {\bibinfo {author} {\bibfnamefont {A.}~\bibnamefont
  {Weidinger}}, \bibinfo {author} {\bibfnamefont {G.~B.}\ \bibnamefont
  {Mbeng}},\ and\ \bibinfo {author} {\bibfnamefont {W.}~\bibnamefont
  {Lechner}},\ }\bibfield  {title} {\bibinfo {title} {Error mitigation for
  quantum approximate optimization},\ }\href
  {https://doi.org/10.1103/PhysRevA.108.032408} {\bibfield  {journal} {\bibinfo
   {journal} {Phys. Rev. A}\ }\textbf {\bibinfo {volume} {108}},\ \bibinfo
  {pages} {032408} (\bibinfo {year} {2023})}\BibitemShut {NoStop}%
\bibitem [{\citenamefont {Shaydulin}\ and\ \citenamefont
  {Galda}(2021)}]{shaydulin2021error}%
  \BibitemOpen
  \bibfield  {author} {\bibinfo {author} {\bibfnamefont {R.}~\bibnamefont
  {Shaydulin}}\ and\ \bibinfo {author} {\bibfnamefont {A.}~\bibnamefont
  {Galda}},\ }\bibfield  {title} {\bibinfo {title} {Error mitigation for deep
  quantum optimization circuits by leveraging problem symmetries},\ }in\ \href
  {https://doi.org/10.1109/QCE52317.2021.00046} {\emph {\bibinfo {booktitle}
  {2021 IEEE International Conference on Quantum Computing and Engineering
  (QCE)}}}\ (\bibinfo {year} {2021})\ pp.\ \bibinfo {pages}
  {291--300}\BibitemShut {NoStop}%
\bibitem [{\citenamefont {Herrman}\ \emph {et~al.}(2022)\citenamefont
  {Herrman}, \citenamefont {Lotshaw}, \citenamefont {Ostrowski}, \citenamefont
  {Humble},\ and\ \citenamefont {Siopsis}}]{herrman2022multi}%
  \BibitemOpen
  \bibfield  {author} {\bibinfo {author} {\bibfnamefont {R.}~\bibnamefont
  {Herrman}}, \bibinfo {author} {\bibfnamefont {P.~C.}\ \bibnamefont
  {Lotshaw}}, \bibinfo {author} {\bibfnamefont {J.}~\bibnamefont {Ostrowski}},
  \bibinfo {author} {\bibfnamefont {T.~S.}\ \bibnamefont {Humble}},\ and\
  \bibinfo {author} {\bibfnamefont {G.}~\bibnamefont {Siopsis}},\ }\bibfield
  {title} {\bibinfo {title} {Multi-angle quantum approximate optimization
  algorithm},\ }\href {https://doi.org/10.1038/s41598-022-10555-8} {\bibfield
  {journal} {\bibinfo  {journal} {Scientific Reports}\ }\textbf {\bibinfo
  {volume} {12}},\ \bibinfo {pages} {6781} (\bibinfo {year}
  {2022})}\BibitemShut {NoStop}%
\bibitem [{\citenamefont {Govia}\ \emph {et~al.}(2021)\citenamefont {Govia},
  \citenamefont {Poole}, \citenamefont {Saffman},\ and\ \citenamefont
  {Krovi}}]{PhysRevA.104.062428}%
  \BibitemOpen
  \bibfield  {author} {\bibinfo {author} {\bibfnamefont {L.~C.~G.}\
  \bibnamefont {Govia}}, \bibinfo {author} {\bibfnamefont {C.}~\bibnamefont
  {Poole}}, \bibinfo {author} {\bibfnamefont {M.}~\bibnamefont {Saffman}},\
  and\ \bibinfo {author} {\bibfnamefont {H.~K.}\ \bibnamefont {Krovi}},\
  }\bibfield  {title} {\bibinfo {title} {Freedom of the mixer rotation axis
  improves performance in the quantum approximate optimization algorithm},\
  }\href {https://doi.org/10.1103/PhysRevA.104.062428} {\bibfield  {journal}
  {\bibinfo  {journal} {Phys. Rev. A}\ }\textbf {\bibinfo {volume} {104}},\
  \bibinfo {pages} {062428} (\bibinfo {year} {2021})}\BibitemShut {NoStop}%
\bibitem [{\citenamefont {Yu}\ \emph {et~al.}(2022)\citenamefont {Yu},
  \citenamefont {Cao}, \citenamefont {Dewey}, \citenamefont {Wang},
  \citenamefont {Shannon},\ and\ \citenamefont {Joynt}}]{yu2022quantum}%
  \BibitemOpen
  \bibfield  {author} {\bibinfo {author} {\bibfnamefont {Y.}~\bibnamefont
  {Yu}}, \bibinfo {author} {\bibfnamefont {C.}~\bibnamefont {Cao}}, \bibinfo
  {author} {\bibfnamefont {C.}~\bibnamefont {Dewey}}, \bibinfo {author}
  {\bibfnamefont {X.-B.}\ \bibnamefont {Wang}}, \bibinfo {author}
  {\bibfnamefont {N.}~\bibnamefont {Shannon}},\ and\ \bibinfo {author}
  {\bibfnamefont {R.}~\bibnamefont {Joynt}},\ }\bibfield  {title} {\bibinfo
  {title} {Quantum approximate optimization algorithm with adaptive bias
  fields},\ }\href {https://doi.org/10.1103/PhysRevResearch.4.023249}
  {\bibfield  {journal} {\bibinfo  {journal} {Phys. Rev. Research}\ }\textbf
  {\bibinfo {volume} {4}},\ \bibinfo {pages} {023249} (\bibinfo {year}
  {2022})}\BibitemShut {NoStop}%
\bibitem [{\citenamefont {Zhu}\ \emph {et~al.}(2022)\citenamefont {Zhu},
  \citenamefont {Tang}, \citenamefont {Barron}, \citenamefont
  {Calderon-Vargas}, \citenamefont {Mayhall}, \citenamefont {Barnes},\ and\
  \citenamefont {Economou}}]{PhysRevResearch.4.033029}%
  \BibitemOpen
  \bibfield  {author} {\bibinfo {author} {\bibfnamefont {L.}~\bibnamefont
  {Zhu}}, \bibinfo {author} {\bibfnamefont {H.~L.}\ \bibnamefont {Tang}},
  \bibinfo {author} {\bibfnamefont {G.~S.}\ \bibnamefont {Barron}}, \bibinfo
  {author} {\bibfnamefont {F.~A.}\ \bibnamefont {Calderon-Vargas}}, \bibinfo
  {author} {\bibfnamefont {N.~J.}\ \bibnamefont {Mayhall}}, \bibinfo {author}
  {\bibfnamefont {E.}~\bibnamefont {Barnes}},\ and\ \bibinfo {author}
  {\bibfnamefont {S.~E.}\ \bibnamefont {Economou}},\ }\bibfield  {title}
  {\bibinfo {title} {Adaptive quantum approximate optimization algorithm for
  solving combinatorial problems on a quantum computer},\ }\href
  {https://doi.org/10.1103/PhysRevResearch.4.033029} {\bibfield  {journal}
  {\bibinfo  {journal} {Phys. Rev. Research}\ }\textbf {\bibinfo {volume}
  {4}},\ \bibinfo {pages} {033029} (\bibinfo {year} {2022})}\BibitemShut
  {NoStop}%
\bibitem [{\citenamefont {Tate}\ \emph
  {et~al.}(2023{\natexlab{b}})\citenamefont {Tate}, \citenamefont {Moondra},
  \citenamefont {Gard}, \citenamefont {Mohler},\ and\ \citenamefont
  {Gupta}}]{tate2022warm}%
  \BibitemOpen
  \bibfield  {author} {\bibinfo {author} {\bibfnamefont {R.}~\bibnamefont
  {Tate}}, \bibinfo {author} {\bibfnamefont {J.}~\bibnamefont {Moondra}},
  \bibinfo {author} {\bibfnamefont {B.}~\bibnamefont {Gard}}, \bibinfo {author}
  {\bibfnamefont {G.}~\bibnamefont {Mohler}},\ and\ \bibinfo {author}
  {\bibfnamefont {S.}~\bibnamefont {Gupta}},\ }\bibfield  {title} {\bibinfo
  {title} {Warm-{S}tarted {QAOA} with {C}ustom {M}ixers {P}rovably {C}onverges
  and {C}omputationally {B}eats {G}oemans-{W}illiamson's {M}ax-{C}ut at {L}ow
  {C}ircuit {D}epths},\ }\href {https://doi.org/10.22331/q-2023-09-26-1121}
  {\bibfield  {journal} {\bibinfo  {journal} {{Quantum}}\ }\textbf {\bibinfo
  {volume} {7}},\ \bibinfo {pages} {1121} (\bibinfo {year}
  {2023}{\natexlab{b}})}\BibitemShut {NoStop}%
\bibitem [{\citenamefont {Akshay}\ \emph
  {et~al.}(2021{\natexlab{b}})\citenamefont {Akshay}, \citenamefont
  {Philathong}, \citenamefont {Zacharov},\ and\ \citenamefont
  {Biamonte}}]{akshay2021reachability}%
  \BibitemOpen
  \bibfield  {author} {\bibinfo {author} {\bibfnamefont {V.}~\bibnamefont
  {Akshay}}, \bibinfo {author} {\bibfnamefont {H.}~\bibnamefont {Philathong}},
  \bibinfo {author} {\bibfnamefont {I.}~\bibnamefont {Zacharov}},\ and\
  \bibinfo {author} {\bibfnamefont {J.}~\bibnamefont {Biamonte}},\ }\bibfield
  {title} {\bibinfo {title} {Reachability {D}eficits in {Q}uantum {A}pproximate
  {O}ptimization of {G}raph {P}roblems},\ }\href
  {https://doi.org/10.22331/q-2021-08-30-532} {\bibfield  {journal} {\bibinfo
  {journal} {{Quantum}}\ }\textbf {\bibinfo {volume} {5}},\ \bibinfo {pages}
  {532} (\bibinfo {year} {2021}{\natexlab{b}})}\BibitemShut {NoStop}%
\bibitem [{\citenamefont {Goemans}\ and\ \citenamefont
  {Williamson}(1995)}]{goemans1995improved}%
  \BibitemOpen
  \bibfield  {author} {\bibinfo {author} {\bibfnamefont {M.~X.}\ \bibnamefont
  {Goemans}}\ and\ \bibinfo {author} {\bibfnamefont {D.~P.}\ \bibnamefont
  {Williamson}},\ }\bibfield  {title} {\bibinfo {title} {Improved approximation
  algorithms for maximum cut and satisfiability problems using semidefinite
  programming},\ }\href {https://doi.org/10.1145/227683.227684} {\bibfield
  {journal} {\bibinfo  {journal} {Journal of the ACM (JACM)}\ }\textbf
  {\bibinfo {volume} {42}},\ \bibinfo {pages} {1115} (\bibinfo {year}
  {1995})}\BibitemShut {NoStop}%
\bibitem [{\citenamefont {Vazirani}(2001)}]{vazirani2001approximation}%
  \BibitemOpen
  \bibfield  {author} {\bibinfo {author} {\bibfnamefont {V.~V.}\ \bibnamefont
  {Vazirani}},\ }\href
  {https://link.springer.com/book/10.1007/978-3-662-04565-7} {\emph {\bibinfo
  {title} {Approximation algorithms}}},\ Vol.~\bibinfo {volume} {1}\ (\bibinfo
  {publisher} {Springer},\ \bibinfo {year} {2001})\BibitemShut {NoStop}%
\bibitem [{\citenamefont {Barahona}\ \emph {et~al.}(1988)\citenamefont
  {Barahona}, \citenamefont {Gr{\"o}tschel}, \citenamefont {J{\"u}nger},\ and\
  \citenamefont {Reinelt}}]{barahona1988application}%
  \BibitemOpen
  \bibfield  {author} {\bibinfo {author} {\bibfnamefont {F.}~\bibnamefont
  {Barahona}}, \bibinfo {author} {\bibfnamefont {M.}~\bibnamefont
  {Gr{\"o}tschel}}, \bibinfo {author} {\bibfnamefont {M.}~\bibnamefont
  {J{\"u}nger}},\ and\ \bibinfo {author} {\bibfnamefont {G.}~\bibnamefont
  {Reinelt}},\ }\bibfield  {title} {\bibinfo {title} {An application of
  combinatorial optimization to statistical physics and circuit layout
  design},\ }\href {https://doi.org/10.1287/opre.36.3.493} {\bibfield
  {journal} {\bibinfo  {journal} {Operations Research}\ }\textbf {\bibinfo
  {volume} {36}},\ \bibinfo {pages} {493} (\bibinfo {year} {1988})}\BibitemShut
  {NoStop}%
\bibitem [{\citenamefont {Agrawal}\ \emph {et~al.}(2003)\citenamefont
  {Agrawal}, \citenamefont {Rajagopalan}, \citenamefont {Srikant},\ and\
  \citenamefont {Xu}}]{mining}%
  \BibitemOpen
  \bibfield  {author} {\bibinfo {author} {\bibfnamefont {R.}~\bibnamefont
  {Agrawal}}, \bibinfo {author} {\bibfnamefont {S.}~\bibnamefont
  {Rajagopalan}}, \bibinfo {author} {\bibfnamefont {R.}~\bibnamefont
  {Srikant}},\ and\ \bibinfo {author} {\bibfnamefont {Y.}~\bibnamefont {Xu}},\
  }\bibfield  {title} {\bibinfo {title} {Mining newsgroups using networks
  arising from social behavior},\ }in\ \href
  {https://doi.org/10.1145/775152.775227} {\emph {\bibinfo {booktitle}
  {Proceedings of the 12th International Conference on World Wide Web}}},\
  \bibinfo {series and number} {WWW '03}\ (\bibinfo  {publisher} {Association
  for Computing Machinery},\ \bibinfo {address} {New York, NY, USA},\ \bibinfo
  {year} {2003})\ p.\ \bibinfo {pages} {529–535}\BibitemShut {NoStop}%
\bibitem [{\citenamefont {Poland}\ and\ \citenamefont
  {Zeugmann}(2006)}]{data_clustering}%
  \BibitemOpen
  \bibfield  {author} {\bibinfo {author} {\bibfnamefont {J.}~\bibnamefont
  {Poland}}\ and\ \bibinfo {author} {\bibfnamefont {T.}~\bibnamefont
  {Zeugmann}},\ }\bibfield  {title} {\bibinfo {title} {Clustering pairwise
  distances with missing data: Maximum cuts versus normalized cuts},\ }in\
  \href {https://doi.org/10.1007/11893318_21} {\emph {\bibinfo {booktitle}
  {Discovery Science}}},\ \bibinfo {editor} {edited by\ \bibinfo {editor}
  {\bibfnamefont {L.}~\bibnamefont {Todorovski}}, \bibinfo {editor}
  {\bibfnamefont {N.}~\bibnamefont {Lavra{\v{c}}}},\ and\ \bibinfo {editor}
  {\bibfnamefont {K.~P.}\ \bibnamefont {Jantke}}}\ (\bibinfo  {publisher}
  {Springer Berlin Heidelberg},\ \bibinfo {address} {Berlin, Heidelberg},\
  \bibinfo {year} {2006})\ pp.\ \bibinfo {pages} {197--208}\BibitemShut
  {NoStop}%
\bibitem [{\citenamefont {Wang}\ \emph {et~al.}(2013)\citenamefont {Wang},
  \citenamefont {Jebara},\ and\ \citenamefont {Chang}}]{wang2013semi}%
  \BibitemOpen
  \bibfield  {author} {\bibinfo {author} {\bibfnamefont {J.}~\bibnamefont
  {Wang}}, \bibinfo {author} {\bibfnamefont {T.}~\bibnamefont {Jebara}},\ and\
  \bibinfo {author} {\bibfnamefont {S.-F.}\ \bibnamefont {Chang}},\ }\bibfield
  {title} {\bibinfo {title} {Semi-supervised learning using greedy
  {M}ax-{C}ut},\ }\href {https://dl.acm.org/doi/10.5555/2567709.2502605}
  {\bibfield  {journal} {\bibinfo  {journal} {The Journal of Machine Learning
  Research}\ }\textbf {\bibinfo {volume} {14}},\ \bibinfo {pages} {771}
  (\bibinfo {year} {2013})}\BibitemShut {NoStop}%
\bibitem [{\citenamefont {Deza}\ and\ \citenamefont
  {Laurent}(1994{\natexlab{a}})}]{DEZA1994191}%
  \BibitemOpen
  \bibfield  {author} {\bibinfo {author} {\bibfnamefont {M.}~\bibnamefont
  {Deza}}\ and\ \bibinfo {author} {\bibfnamefont {M.}~\bibnamefont {Laurent}},\
  }\bibfield  {title} {\bibinfo {title} {Applications of cut polyhedra — i},\
  }\href {https://doi.org/https://doi.org/10.1016/0377-0427(94)90020-5}
  {\bibfield  {journal} {\bibinfo  {journal} {Journal of Computational and
  Applied Mathematics}\ }\textbf {\bibinfo {volume} {55}},\ \bibinfo {pages}
  {191} (\bibinfo {year} {1994}{\natexlab{a}})}\BibitemShut {NoStop}%
\bibitem [{\citenamefont {Deza}\ and\ \citenamefont
  {Laurent}(1994{\natexlab{b}})}]{DEZA1994217}%
  \BibitemOpen
  \bibfield  {author} {\bibinfo {author} {\bibfnamefont {M.}~\bibnamefont
  {Deza}}\ and\ \bibinfo {author} {\bibfnamefont {M.}~\bibnamefont {Laurent}},\
  }\bibfield  {title} {\bibinfo {title} {Applications of cut polyhedra —
  ii},\ }\href {https://doi.org/https://doi.org/10.1016/0377-0427(94)90021-3}
  {\bibfield  {journal} {\bibinfo  {journal} {Journal of Computational and
  Applied Mathematics}\ }\textbf {\bibinfo {volume} {55}},\ \bibinfo {pages}
  {217} (\bibinfo {year} {1994}{\natexlab{b}})}\BibitemShut {NoStop}%
\bibitem [{\citenamefont {Karp}(1972)}]{karp1972reducibility}%
  \BibitemOpen
  \bibfield  {author} {\bibinfo {author} {\bibfnamefont {R.~M.}\ \bibnamefont
  {Karp}},\ }\bibfield  {title} {\bibinfo {title} {Reducibility among
  combinatorial problems},\ }in\ \href
  {https://doi.org/10.1007/978-1-4684-2001-2_9} {\emph {\bibinfo {booktitle}
  {Complexity of computer computations}}}\ (\bibinfo  {publisher} {Springer},\
  \bibinfo {year} {1972})\ pp.\ \bibinfo {pages} {85--103}\BibitemShut
  {NoStop}%
\bibitem [{\citenamefont {Halperin}\ \emph {et~al.}(2004)\citenamefont
  {Halperin}, \citenamefont {Livnat},\ and\ \citenamefont
  {Zwick}}]{halperin2004max}%
  \BibitemOpen
  \bibfield  {author} {\bibinfo {author} {\bibfnamefont {E.}~\bibnamefont
  {Halperin}}, \bibinfo {author} {\bibfnamefont {D.}~\bibnamefont {Livnat}},\
  and\ \bibinfo {author} {\bibfnamefont {U.}~\bibnamefont {Zwick}},\ }\bibfield
   {title} {\bibinfo {title} {{MAX CUT} in cubic graphs},\ }\href
  {https://doi.org/https://doi.org/10.1016/j.jalgor.2004.06.001} {\bibfield
  {journal} {\bibinfo  {journal} {Journal of Algorithms}\ }\textbf {\bibinfo
  {volume} {53}},\ \bibinfo {pages} {169} (\bibinfo {year} {2004})}\BibitemShut
  {NoStop}%
\bibitem [{\citenamefont {Feige}\ \emph {et~al.}(2002)\citenamefont {Feige},
  \citenamefont {Karpinski},\ and\ \citenamefont
  {Langberg}}]{feige2002improved}%
  \BibitemOpen
  \bibfield  {author} {\bibinfo {author} {\bibfnamefont {U.}~\bibnamefont
  {Feige}}, \bibinfo {author} {\bibfnamefont {M.}~\bibnamefont {Karpinski}},\
  and\ \bibinfo {author} {\bibfnamefont {M.}~\bibnamefont {Langberg}},\
  }\bibfield  {title} {\bibinfo {title} {Improved approximation of {M}ax-{C}ut
  on graphs of bounded degree},\ }\href
  {https://doi.org/https://doi.org/10.1016/S0196-6774(02)00005-6} {\bibfield
  {journal} {\bibinfo  {journal} {Journal of Algorithms}\ }\textbf {\bibinfo
  {volume} {43}},\ \bibinfo {pages} {201} (\bibinfo {year} {2002})}\BibitemShut
  {NoStop}%
\bibitem [{\citenamefont {Khot}(2002)}]{khot2002power}%
  \BibitemOpen
  \bibfield  {author} {\bibinfo {author} {\bibfnamefont {S.}~\bibnamefont
  {Khot}},\ }\bibfield  {title} {\bibinfo {title} {On the power of unique
  2-prover 1-round games},\ }in\ \href {https://doi.org/10.1145/509907.510017}
  {\emph {\bibinfo {booktitle} {Proceedings of the Thiry-Fourth Annual ACM
  Symposium on Theory of Computing}}},\ \bibinfo {series and number} {STOC
  '02}\ (\bibinfo  {publisher} {Association for Computing Machinery},\ \bibinfo
  {address} {New York, NY, USA},\ \bibinfo {year} {2002})\ p.\ \bibinfo {pages}
  {767–775}\BibitemShut {NoStop}%
\bibitem [{\citenamefont {Khot}\ \emph {et~al.}(2007)\citenamefont {Khot},
  \citenamefont {Kindler}, \citenamefont {Mossel},\ and\ \citenamefont
  {O’Donnell}}]{khot2007optimal}%
  \BibitemOpen
  \bibfield  {author} {\bibinfo {author} {\bibfnamefont {S.}~\bibnamefont
  {Khot}}, \bibinfo {author} {\bibfnamefont {G.}~\bibnamefont {Kindler}},
  \bibinfo {author} {\bibfnamefont {E.}~\bibnamefont {Mossel}},\ and\ \bibinfo
  {author} {\bibfnamefont {R.}~\bibnamefont {O’Donnell}},\ }\bibfield
  {title} {\bibinfo {title} {Optimal inapproximability results for {MAX}-{CUT}
  and other 2-variable {CSP}s?},\ }\href
  {https://doi.org/10.1137/S0097539705447372} {\bibfield  {journal} {\bibinfo
  {journal} {SIAM Journal on Computing}\ }\textbf {\bibinfo {volume} {37}},\
  \bibinfo {pages} {319} (\bibinfo {year} {2007})}\BibitemShut {NoStop}%
\bibitem [{\citenamefont {Khot}(2010)}]{5497893}%
  \BibitemOpen
  \bibfield  {author} {\bibinfo {author} {\bibfnamefont {S.}~\bibnamefont
  {Khot}},\ }\bibfield  {title} {\bibinfo {title} {On the unique games
  conjecture (invited survey)},\ }in\ \href
  {https://doi.org/10.1109/CCC.2010.19} {\emph {\bibinfo {booktitle} {2010 IEEE
  25th Annual Conference on Computational Complexity}}}\ (\bibinfo {year}
  {2010})\ pp.\ \bibinfo {pages} {99--121}\BibitemShut {NoStop}%
\bibitem [{\citenamefont {Khot}\ and\ \citenamefont
  {Vishnoi}(2015)}]{khot2015unique}%
  \BibitemOpen
  \bibfield  {author} {\bibinfo {author} {\bibfnamefont {S.~A.}\ \bibnamefont
  {Khot}}\ and\ \bibinfo {author} {\bibfnamefont {N.~K.}\ \bibnamefont
  {Vishnoi}},\ }\bibfield  {title} {\bibinfo {title} {The unique games
  conjecture, integrality gap for cut problems and embeddability of
  negative-type metrics into $l_1$},\ }\href {https://doi.org/10.1145/2629614}
  {\bibfield  {journal} {\bibinfo  {journal} {Journal of the ACM (JACM)}\
  }\textbf {\bibinfo {volume} {62}},\ \bibinfo {pages} {1} (\bibinfo {year}
  {2015})}\BibitemShut {NoStop}%
\bibitem [{\citenamefont {H{\aa}stad}(2001)}]{haastad2001some}%
  \BibitemOpen
  \bibfield  {author} {\bibinfo {author} {\bibfnamefont {J.}~\bibnamefont
  {H{\aa}stad}},\ }\bibfield  {title} {\bibinfo {title} {Some optimal
  inapproximability results},\ }\href {https://doi.org/10.1145/502090.502098}
  {\bibfield  {journal} {\bibinfo  {journal} {Journal of the ACM (JACM)}\
  }\textbf {\bibinfo {volume} {48}},\ \bibinfo {pages} {798} (\bibinfo {year}
  {2001})}\BibitemShut {NoStop}%
\bibitem [{\citenamefont {Trevisan}\ \emph {et~al.}(2000)\citenamefont
  {Trevisan}, \citenamefont {Sorkin}, \citenamefont {Sudan},\ and\
  \citenamefont {Williamson}}]{trevisan2000gadgets}%
  \BibitemOpen
  \bibfield  {author} {\bibinfo {author} {\bibfnamefont {L.}~\bibnamefont
  {Trevisan}}, \bibinfo {author} {\bibfnamefont {G.~B.}\ \bibnamefont
  {Sorkin}}, \bibinfo {author} {\bibfnamefont {M.}~\bibnamefont {Sudan}},\ and\
  \bibinfo {author} {\bibfnamefont {D.~P.}\ \bibnamefont {Williamson}},\
  }\bibfield  {title} {\bibinfo {title} {Gadgets, approximation, and linear
  programming},\ }\href {https://doi.org/10.1137/S0097539797328847} {\bibfield
  {journal} {\bibinfo  {journal} {SIAM Journal on Computing}\ }\textbf
  {\bibinfo {volume} {29}},\ \bibinfo {pages} {2074} (\bibinfo {year}
  {2000})}\BibitemShut {NoStop}%
\bibitem [{\citenamefont {Wang}\ \emph {et~al.}(2018)\citenamefont {Wang},
  \citenamefont {Hadfield}, \citenamefont {Jiang},\ and\ \citenamefont
  {Rieffel}}]{wang2018quantum}%
  \BibitemOpen
  \bibfield  {author} {\bibinfo {author} {\bibfnamefont {Z.}~\bibnamefont
  {Wang}}, \bibinfo {author} {\bibfnamefont {S.}~\bibnamefont {Hadfield}},
  \bibinfo {author} {\bibfnamefont {Z.}~\bibnamefont {Jiang}},\ and\ \bibinfo
  {author} {\bibfnamefont {E.~G.}\ \bibnamefont {Rieffel}},\ }\bibfield
  {title} {\bibinfo {title} {Quantum approximate optimization algorithm for
  {M}ax{C}ut: A fermionic view},\ }\href
  {https://doi.org/10.1103/PhysRevA.97.022304} {\bibfield  {journal} {\bibinfo
  {journal} {Phys. Rev. A}\ }\textbf {\bibinfo {volume} {97}},\ \bibinfo
  {pages} {022304} (\bibinfo {year} {2018})}\BibitemShut {NoStop}%
\bibitem [{\citenamefont {Hadfield}(2018)}]{hadfield2018quantum}%
  \BibitemOpen
  \bibfield  {author} {\bibinfo {author} {\bibfnamefont {S.~A.}\ \bibnamefont
  {Hadfield}},\ }\emph {\bibinfo {title} {Quantum algorithms for scientific
  computing and approximate optimization}},\ \href
  {https://doi.org/10.7916/D8X650C9} {Ph.D. thesis},\ \bibinfo  {school}
  {Columbia University} (\bibinfo {year} {2018})\BibitemShut {NoStop}%
\bibitem [{\citenamefont {Hadfield}\ \emph {et~al.}(2022)\citenamefont
  {Hadfield}, \citenamefont {Hogg},\ and\ \citenamefont
  {Rieffel}}]{hadfield2022analytical}%
  \BibitemOpen
  \bibfield  {author} {\bibinfo {author} {\bibfnamefont {S.}~\bibnamefont
  {Hadfield}}, \bibinfo {author} {\bibfnamefont {T.}~\bibnamefont {Hogg}},\
  and\ \bibinfo {author} {\bibfnamefont {E.~G.}\ \bibnamefont {Rieffel}},\
  }\bibfield  {title} {\bibinfo {title} {Analytical framework for quantum
  alternating operator ansätze},\ }\href
  {https://doi.org/10.1088/2058-9565/aca3ce} {\bibfield  {journal} {\bibinfo
  {journal} {Quantum Science and Technology}\ }\textbf {\bibinfo {volume}
  {8}},\ \bibinfo {pages} {015017} (\bibinfo {year} {2022})}\BibitemShut
  {NoStop}%
\bibitem [{\citenamefont {Ozaeta}\ \emph {et~al.}(2022)\citenamefont {Ozaeta},
  \citenamefont {van Dam},\ and\ \citenamefont
  {McMahon}}]{ozaeta2022expectation}%
  \BibitemOpen
  \bibfield  {author} {\bibinfo {author} {\bibfnamefont {A.}~\bibnamefont
  {Ozaeta}}, \bibinfo {author} {\bibfnamefont {W.}~\bibnamefont {van Dam}},\
  and\ \bibinfo {author} {\bibfnamefont {P.~L.}\ \bibnamefont {McMahon}},\
  }\bibfield  {title} {\bibinfo {title} {Expectation values from the
  single-layer quantum approximate optimization algorithm on {I}sing
  problems},\ }\href {https://doi.org/10.1088/2058-9565/ac9013} {\bibfield
  {journal} {\bibinfo  {journal} {Quantum Science and Technology}\ }\textbf
  {\bibinfo {volume} {7}},\ \bibinfo {pages} {045036} (\bibinfo {year}
  {2022})}\BibitemShut {NoStop}%
\bibitem [{\citenamefont {Farhi}\ \emph {et~al.}(2017)\citenamefont {Farhi},
  \citenamefont {Goldstone}, \citenamefont {Gutmann},\ and\ \citenamefont
  {Neven}}]{farhi2017quantum}%
  \BibitemOpen
  \bibfield  {author} {\bibinfo {author} {\bibfnamefont {E.}~\bibnamefont
  {Farhi}}, \bibinfo {author} {\bibfnamefont {J.}~\bibnamefont {Goldstone}},
  \bibinfo {author} {\bibfnamefont {S.}~\bibnamefont {Gutmann}},\ and\ \bibinfo
  {author} {\bibfnamefont {H.}~\bibnamefont {Neven}},\ }\bibfield  {title}
  {\bibinfo {title} {Quantum algorithms for fixed qubit architectures},\ }\href
  {https://doi.org/10.48550/arXiv.1703.06199} {\bibfield  {journal} {\bibinfo
  {journal} {arXiv preprint arXiv:1703.06199}\ } (\bibinfo {year}
  {2017})}\BibitemShut {NoStop}%
\bibitem [{\citenamefont {Steger}\ and\ \citenamefont
  {Wormald}(1999)}]{steger1999generating}%
  \BibitemOpen
  \bibfield  {author} {\bibinfo {author} {\bibfnamefont {A.}~\bibnamefont
  {Steger}}\ and\ \bibinfo {author} {\bibfnamefont {N.~C.}\ \bibnamefont
  {Wormald}},\ }\bibfield  {title} {\bibinfo {title} {Generating random regular
  graphs quickly},\ }\href {https://doi.org/10.1017/S0963548399003867}
  {\bibfield  {journal} {\bibinfo  {journal} {Combinatorics, Probability and
  Computing}\ }\textbf {\bibinfo {volume} {8}},\ \bibinfo {pages} {377}
  (\bibinfo {year} {1999})}\BibitemShut {NoStop}%
\bibitem [{\citenamefont {{Gurobi Optimization, LLC}}(2022)}]{gurobi}%
  \BibitemOpen
  \bibfield  {author} {\bibinfo {author} {\bibnamefont {{Gurobi Optimization,
  LLC}}},\ }\href {https://www.gurobi.com} {\bibinfo {title} {{Gurobi Optimizer
  Reference Manual}}} (\bibinfo {year} {2022})\BibitemShut {NoStop}%
\bibitem [{\citenamefont {{Qiskit contributors}}(2023)}]{Qiskit}%
  \BibitemOpen
  \bibfield  {author} {\bibinfo {author} {\bibnamefont {{Qiskit
  contributors}}},\ }\href {https://doi.org/10.5281/zenodo.2573505} {\bibinfo
  {title} {Qiskit: An open-source framework for quantum computing}} (\bibinfo
  {year} {2023})\BibitemShut {NoStop}%
\bibitem [{\citenamefont {Larkin}\ \emph {et~al.}(2022)\citenamefont {Larkin},
  \citenamefont {Jonsson}, \citenamefont {Justice},\ and\ \citenamefont
  {Guerreschi}}]{larkin2022evaluation}%
  \BibitemOpen
  \bibfield  {author} {\bibinfo {author} {\bibfnamefont {J.}~\bibnamefont
  {Larkin}}, \bibinfo {author} {\bibfnamefont {M.}~\bibnamefont {Jonsson}},
  \bibinfo {author} {\bibfnamefont {D.}~\bibnamefont {Justice}},\ and\ \bibinfo
  {author} {\bibfnamefont {G.~G.}\ \bibnamefont {Guerreschi}},\ }\bibfield
  {title} {\bibinfo {title} {Evaluation of {QAOA} based on the approximation
  ratio of individual samples},\ }\href
  {https://doi.org/10.1088/2058-9565/ac6973} {\bibfield  {journal} {\bibinfo
  {journal} {Quantum Science and Technology}\ }\textbf {\bibinfo {volume}
  {7}},\ \bibinfo {pages} {045014} (\bibinfo {year} {2022})}\BibitemShut
  {NoStop}%
\bibitem [{\citenamefont {McClean}\ \emph {et~al.}(2018)\citenamefont
  {McClean}, \citenamefont {Boixo}, \citenamefont {Smelyanskiy}, \citenamefont
  {Babbush},\ and\ \citenamefont {Neven}}]{mcclean2018barren}%
  \BibitemOpen
  \bibfield  {author} {\bibinfo {author} {\bibfnamefont {J.~R.}\ \bibnamefont
  {McClean}}, \bibinfo {author} {\bibfnamefont {S.}~\bibnamefont {Boixo}},
  \bibinfo {author} {\bibfnamefont {V.~N.}\ \bibnamefont {Smelyanskiy}},
  \bibinfo {author} {\bibfnamefont {R.}~\bibnamefont {Babbush}},\ and\ \bibinfo
  {author} {\bibfnamefont {H.}~\bibnamefont {Neven}},\ }\bibfield  {title}
  {\bibinfo {title} {Barren plateaus in quantum neural network training
  landscapes},\ }\href {https://doi.org/10.1038/s41467-018-07090-4} {\bibfield
  {journal} {\bibinfo  {journal} {Nature communications}\ }\textbf {\bibinfo
  {volume} {9}},\ \bibinfo {pages} {4812} (\bibinfo {year} {2018})}\BibitemShut
  {NoStop}%
\bibitem [{\citenamefont {Powell}(1964)}]{powell1964efficient}%
  \BibitemOpen
  \bibfield  {author} {\bibinfo {author} {\bibfnamefont {M.~J.}\ \bibnamefont
  {Powell}},\ }\bibfield  {title} {\bibinfo {title} {An efficient method for
  finding the minimum of a function of several variables without calculating
  derivatives},\ }\href {https://doi.org/10.1093/comjnl/7.2.155} {\bibfield
  {journal} {\bibinfo  {journal} {The computer journal}\ }\textbf {\bibinfo
  {volume} {7}},\ \bibinfo {pages} {155} (\bibinfo {year} {1964})}\BibitemShut
  {NoStop}%
\bibitem [{\citenamefont {Gaidai}\ and\ \citenamefont
  {Herrman}(2023)}]{gaidai2023performance}%
  \BibitemOpen
  \bibfield  {author} {\bibinfo {author} {\bibfnamefont {I.}~\bibnamefont
  {Gaidai}}\ and\ \bibinfo {author} {\bibfnamefont {R.}~\bibnamefont
  {Herrman}},\ }\bibfield  {title} {\bibinfo {title} {Performance analysis of
  multi-angle {QAOA} for p > 1},\ }\href
  {https://doi.org/10.48550/arXiv.2312.00200} {\bibfield  {journal} {\bibinfo
  {journal} {arXiv preprint arXiv:2312.00200}\ } (\bibinfo {year}
  {2023})}\BibitemShut {NoStop}%
\bibitem [{\citenamefont {Holmes}\ \emph {et~al.}(2022)\citenamefont {Holmes},
  \citenamefont {Sharma}, \citenamefont {Cerezo},\ and\ \citenamefont
  {Coles}}]{holmes2022connecting}%
  \BibitemOpen
  \bibfield  {author} {\bibinfo {author} {\bibfnamefont {Z.}~\bibnamefont
  {Holmes}}, \bibinfo {author} {\bibfnamefont {K.}~\bibnamefont {Sharma}},
  \bibinfo {author} {\bibfnamefont {M.}~\bibnamefont {Cerezo}},\ and\ \bibinfo
  {author} {\bibfnamefont {P.~J.}\ \bibnamefont {Coles}},\ }\bibfield  {title}
  {\bibinfo {title} {Connecting ansatz expressibility to gradient magnitudes
  and barren plateaus},\ }\href {https://doi.org/10.1103/PRXQuantum.3.010313}
  {\bibfield  {journal} {\bibinfo  {journal} {PRX Quantum}\ }\textbf {\bibinfo
  {volume} {3}},\ \bibinfo {pages} {010313} (\bibinfo {year}
  {2022})}\BibitemShut {NoStop}%
\bibitem [{\citenamefont {Larocca}\ \emph {et~al.}(2022)\citenamefont
  {Larocca}, \citenamefont {Czarnik}, \citenamefont {Sharma}, \citenamefont
  {Muraleedharan}, \citenamefont {Coles},\ and\ \citenamefont
  {Cerezo}}]{larocca2022diagnosing}%
  \BibitemOpen
  \bibfield  {author} {\bibinfo {author} {\bibfnamefont {M.}~\bibnamefont
  {Larocca}}, \bibinfo {author} {\bibfnamefont {P.}~\bibnamefont {Czarnik}},
  \bibinfo {author} {\bibfnamefont {K.}~\bibnamefont {Sharma}}, \bibinfo
  {author} {\bibfnamefont {G.}~\bibnamefont {Muraleedharan}}, \bibinfo {author}
  {\bibfnamefont {P.~J.}\ \bibnamefont {Coles}},\ and\ \bibinfo {author}
  {\bibfnamefont {M.}~\bibnamefont {Cerezo}},\ }\bibfield  {title} {\bibinfo
  {title} {Diagnosing {B}arren {P}lateaus with {T}ools from {Q}uantum {O}ptimal
  {C}ontrol},\ }\href {https://doi.org/10.22331/q-2022-09-29-824} {\bibfield
  {journal} {\bibinfo  {journal} {{Quantum}}\ }\textbf {\bibinfo {volume}
  {6}},\ \bibinfo {pages} {824} (\bibinfo {year} {2022})}\BibitemShut {NoStop}%
\bibitem [{\citenamefont {Allen-Zhu}\ \emph {et~al.}(2019)\citenamefont
  {Allen-Zhu}, \citenamefont {Li},\ and\ \citenamefont
  {Song}}]{allen2019convergence}%
  \BibitemOpen
  \bibfield  {author} {\bibinfo {author} {\bibfnamefont {Z.}~\bibnamefont
  {Allen-Zhu}}, \bibinfo {author} {\bibfnamefont {Y.}~\bibnamefont {Li}},\ and\
  \bibinfo {author} {\bibfnamefont {Z.}~\bibnamefont {Song}},\ }\bibfield
  {title} {\bibinfo {title} {A convergence theory for deep learning via
  over-parameterization},\ }in\ \href
  {https://proceedings.mlr.press/v97/allen-zhu19a.html} {\emph {\bibinfo
  {booktitle} {International conference on machine learning}}}\ (\bibinfo
  {organization} {PMLR},\ \bibinfo {year} {2019})\ pp.\ \bibinfo {pages}
  {242--252}\BibitemShut {NoStop}%
\bibitem [{\citenamefont {Du}\ \emph {et~al.}(2019)\citenamefont {Du},
  \citenamefont {Lee}, \citenamefont {Li}, \citenamefont {Wang},\ and\
  \citenamefont {Zhai}}]{du2019gradient}%
  \BibitemOpen
  \bibfield  {author} {\bibinfo {author} {\bibfnamefont {S.}~\bibnamefont
  {Du}}, \bibinfo {author} {\bibfnamefont {J.}~\bibnamefont {Lee}}, \bibinfo
  {author} {\bibfnamefont {H.}~\bibnamefont {Li}}, \bibinfo {author}
  {\bibfnamefont {L.}~\bibnamefont {Wang}},\ and\ \bibinfo {author}
  {\bibfnamefont {X.}~\bibnamefont {Zhai}},\ }\bibfield  {title} {\bibinfo
  {title} {Gradient descent finds global minima of deep neural networks},\ }in\
  \href {https://proceedings.mlr.press/v97/du19c.html} {\emph {\bibinfo
  {booktitle} {International conference on machine learning}}}\ (\bibinfo
  {organization} {PMLR},\ \bibinfo {year} {2019})\ pp.\ \bibinfo {pages}
  {1675--1685}\BibitemShut {NoStop}%
\bibitem [{\citenamefont {Buhai}\ \emph {et~al.}(2020)\citenamefont {Buhai},
  \citenamefont {Halpern}, \citenamefont {Kim}, \citenamefont {Risteski},\ and\
  \citenamefont {Sontag}}]{buhai2020empirical}%
  \BibitemOpen
  \bibfield  {author} {\bibinfo {author} {\bibfnamefont {R.-D.}\ \bibnamefont
  {Buhai}}, \bibinfo {author} {\bibfnamefont {Y.}~\bibnamefont {Halpern}},
  \bibinfo {author} {\bibfnamefont {Y.}~\bibnamefont {Kim}}, \bibinfo {author}
  {\bibfnamefont {A.}~\bibnamefont {Risteski}},\ and\ \bibinfo {author}
  {\bibfnamefont {D.}~\bibnamefont {Sontag}},\ }\bibfield  {title} {\bibinfo
  {title} {Empirical study of the benefits of overparameterization in learning
  latent variable models},\ }in\ \href
  {https://proceedings.mlr.press/v119/buhai20a.html} {\emph {\bibinfo
  {booktitle} {International Conference on Machine Learning}}}\ (\bibinfo
  {organization} {PMLR},\ \bibinfo {year} {2020})\ pp.\ \bibinfo {pages}
  {1211--1219}\BibitemShut {NoStop}%
\bibitem [{\citenamefont {Du}\ \emph {et~al.}(2018)\citenamefont {Du},
  \citenamefont {Zhai}, \citenamefont {Poczos},\ and\ \citenamefont
  {Singh}}]{du2018gradient}%
  \BibitemOpen
  \bibfield  {author} {\bibinfo {author} {\bibfnamefont {S.~S.}\ \bibnamefont
  {Du}}, \bibinfo {author} {\bibfnamefont {X.}~\bibnamefont {Zhai}}, \bibinfo
  {author} {\bibfnamefont {B.}~\bibnamefont {Poczos}},\ and\ \bibinfo {author}
  {\bibfnamefont {A.}~\bibnamefont {Singh}},\ }\bibfield  {title} {\bibinfo
  {title} {Gradient descent provably optimizes over-parameterized neural
  networks},\ }\href {https://doi.org/10.48550/arXiv.1810.02054} {\bibfield
  {journal} {\bibinfo  {journal} {arXiv preprint arXiv:1810.02054}\ } (\bibinfo
  {year} {2018})}\BibitemShut {NoStop}%
\bibitem [{\citenamefont {Brutzkus}\ \emph {et~al.}(2017)\citenamefont
  {Brutzkus}, \citenamefont {Globerson}, \citenamefont {Malach},\ and\
  \citenamefont {Shalev-Shwartz}}]{brutzkus2017sgd}%
  \BibitemOpen
  \bibfield  {author} {\bibinfo {author} {\bibfnamefont {A.}~\bibnamefont
  {Brutzkus}}, \bibinfo {author} {\bibfnamefont {A.}~\bibnamefont {Globerson}},
  \bibinfo {author} {\bibfnamefont {E.}~\bibnamefont {Malach}},\ and\ \bibinfo
  {author} {\bibfnamefont {S.}~\bibnamefont {Shalev-Shwartz}},\ }\bibfield
  {title} {\bibinfo {title} {{SGD} learns over-parameterized networks that
  provably generalize on linearly separable data},\ }\href
  {https://doi.org/10.48550/arXiv.1710.10174} {\bibfield  {journal} {\bibinfo
  {journal} {arXiv preprint arXiv:1710.10174}\ } (\bibinfo {year}
  {2017})}\BibitemShut {NoStop}%
\bibitem [{\citenamefont {Larocca}\ \emph {et~al.}(2023)\citenamefont
  {Larocca}, \citenamefont {Ju}, \citenamefont {Garc{\'\i}a-Mart{\'\i}n},
  \citenamefont {Coles},\ and\ \citenamefont {Cerezo}}]{larocca2023theory}%
  \BibitemOpen
  \bibfield  {author} {\bibinfo {author} {\bibfnamefont {M.}~\bibnamefont
  {Larocca}}, \bibinfo {author} {\bibfnamefont {N.}~\bibnamefont {Ju}},
  \bibinfo {author} {\bibfnamefont {D.}~\bibnamefont
  {Garc{\'\i}a-Mart{\'\i}n}}, \bibinfo {author} {\bibfnamefont {P.~J.}\
  \bibnamefont {Coles}},\ and\ \bibinfo {author} {\bibfnamefont
  {M.}~\bibnamefont {Cerezo}},\ }\bibfield  {title} {\bibinfo {title} {Theory
  of overparametrization in quantum neural networks},\ }\href
  {https://doi.org/10.1038/s43588-023-00467-6} {\bibfield  {journal} {\bibinfo
  {journal} {Nature Computational Science}\ }\textbf {\bibinfo {volume} {3}},\
  \bibinfo {pages} {542} (\bibinfo {year} {2023})}\BibitemShut {NoStop}%
\bibitem [{\citenamefont {Arrasmith}\ \emph {et~al.}(2022)\citenamefont
  {Arrasmith}, \citenamefont {Holmes}, \citenamefont {Cerezo},\ and\
  \citenamefont {Coles}}]{arrasmith2022equivalence}%
  \BibitemOpen
  \bibfield  {author} {\bibinfo {author} {\bibfnamefont {A.}~\bibnamefont
  {Arrasmith}}, \bibinfo {author} {\bibfnamefont {Z.}~\bibnamefont {Holmes}},
  \bibinfo {author} {\bibfnamefont {M.}~\bibnamefont {Cerezo}},\ and\ \bibinfo
  {author} {\bibfnamefont {P.~J.}\ \bibnamefont {Coles}},\ }\bibfield  {title}
  {\bibinfo {title} {Equivalence of quantum barren plateaus to cost
  concentration and narrow gorges},\ }\href
  {https://doi.org/10.1088/2058-9565/ac7d06} {\bibfield  {journal} {\bibinfo
  {journal} {Quantum Science and Technology}\ }\textbf {\bibinfo {volume}
  {7}},\ \bibinfo {pages} {045015} (\bibinfo {year} {2022})}\BibitemShut
  {NoStop}%
\bibitem [{\citenamefont {You}\ \emph {et~al.}(2022)\citenamefont {You},
  \citenamefont {Chakrabarti},\ and\ \citenamefont {Wu}}]{you2022convergence}%
  \BibitemOpen
  \bibfield  {author} {\bibinfo {author} {\bibfnamefont {X.}~\bibnamefont
  {You}}, \bibinfo {author} {\bibfnamefont {S.}~\bibnamefont {Chakrabarti}},\
  and\ \bibinfo {author} {\bibfnamefont {X.}~\bibnamefont {Wu}},\ }\bibfield
  {title} {\bibinfo {title} {A convergence theory for over-parameterized
  variational quantum eigensolvers},\ }\href
  {https://doi.org/10.48550/arXiv.2205.12481} {\bibfield  {journal} {\bibinfo
  {journal} {arXiv preprint arXiv:2205.12481}\ } (\bibinfo {year}
  {2022})}\BibitemShut {NoStop}%
\bibitem [{\citenamefont {Liu}\ \emph {et~al.}(2023)\citenamefont {Liu},
  \citenamefont {Najafi}, \citenamefont {Sharma}, \citenamefont {Tacchino},
  \citenamefont {Jiang},\ and\ \citenamefont {Mezzacapo}}]{liu2023analytic}%
  \BibitemOpen
  \bibfield  {author} {\bibinfo {author} {\bibfnamefont {J.}~\bibnamefont
  {Liu}}, \bibinfo {author} {\bibfnamefont {K.}~\bibnamefont {Najafi}},
  \bibinfo {author} {\bibfnamefont {K.}~\bibnamefont {Sharma}}, \bibinfo
  {author} {\bibfnamefont {F.}~\bibnamefont {Tacchino}}, \bibinfo {author}
  {\bibfnamefont {L.}~\bibnamefont {Jiang}},\ and\ \bibinfo {author}
  {\bibfnamefont {A.}~\bibnamefont {Mezzacapo}},\ }\bibfield  {title} {\bibinfo
  {title} {Analytic theory for the dynamics of wide quantum neural networks},\
  }\href {https://doi.org/10.1103/PhysRevLett.130.150601} {\bibfield  {journal}
  {\bibinfo  {journal} {Phys. Rev. Lett.}\ }\textbf {\bibinfo {volume} {130}},\
  \bibinfo {pages} {150601} (\bibinfo {year} {2023})}\BibitemShut {NoStop}%
\bibitem [{\citenamefont {St{\k{e}}ch{\l}y}\ \emph {et~al.}(2023)\citenamefont
  {St{\k{e}}ch{\l}y}, \citenamefont {Gao}, \citenamefont {Yogendran},
  \citenamefont {Fontana},\ and\ \citenamefont
  {Rudolph}}]{stkechly2023connecting}%
  \BibitemOpen
  \bibfield  {author} {\bibinfo {author} {\bibfnamefont {M.}~\bibnamefont
  {St{\k{e}}ch{\l}y}}, \bibinfo {author} {\bibfnamefont {L.}~\bibnamefont
  {Gao}}, \bibinfo {author} {\bibfnamefont {B.}~\bibnamefont {Yogendran}},
  \bibinfo {author} {\bibfnamefont {E.}~\bibnamefont {Fontana}},\ and\ \bibinfo
  {author} {\bibfnamefont {M.}~\bibnamefont {Rudolph}},\ }\bibfield  {title}
  {\bibinfo {title} {Connecting the {H}amiltonian structure to the {QAOA}
  energy and {F}ourier landscape structure},\ }\href
  {https://doi.org/10.48550/arXiv.2305.13594} {\bibfield  {journal} {\bibinfo
  {journal} {arXiv preprint arXiv:2305.13594}\ } (\bibinfo {year}
  {2023})}\BibitemShut {NoStop}%
\bibitem [{\citenamefont {Kiani}\ \emph {et~al.}(2020)\citenamefont {Kiani},
  \citenamefont {Lloyd},\ and\ \citenamefont {Maity}}]{kiani2020learning}%
  \BibitemOpen
  \bibfield  {author} {\bibinfo {author} {\bibfnamefont {B.~T.}\ \bibnamefont
  {Kiani}}, \bibinfo {author} {\bibfnamefont {S.}~\bibnamefont {Lloyd}},\ and\
  \bibinfo {author} {\bibfnamefont {R.}~\bibnamefont {Maity}},\ }\bibfield
  {title} {\bibinfo {title} {Learning unitaries by gradient descent},\ }\href
  {https://doi.org/10.48550/arXiv.2001.11897} {\bibfield  {journal} {\bibinfo
  {journal} {arXiv preprint arXiv:2001.11897}\ } (\bibinfo {year}
  {2020})}\BibitemShut {NoStop}%
\bibitem [{\citenamefont {Garc{\'\i}a-Mart{\'\i}n}\ \emph
  {et~al.}(2023)\citenamefont {Garc{\'\i}a-Mart{\'\i}n}, \citenamefont
  {Larocca},\ and\ \citenamefont {Cerezo}}]{garcia2023effects}%
  \BibitemOpen
  \bibfield  {author} {\bibinfo {author} {\bibfnamefont {D.}~\bibnamefont
  {Garc{\'\i}a-Mart{\'\i}n}}, \bibinfo {author} {\bibfnamefont
  {M.}~\bibnamefont {Larocca}},\ and\ \bibinfo {author} {\bibfnamefont
  {M.}~\bibnamefont {Cerezo}},\ }\bibfield  {title} {\bibinfo {title} {Effects
  of noise on the overparametrization of quantum neural networks},\ }\href
  {https://doi.org/10.48550/arXiv.2302.05059} {\bibfield  {journal} {\bibinfo
  {journal} {arXiv preprint arXiv:2302.05059}\ } (\bibinfo {year}
  {2023})}\BibitemShut {NoStop}%
\bibitem [{\citenamefont {Brandao}\ \emph {et~al.}(2018)\citenamefont
  {Brandao}, \citenamefont {Broughton}, \citenamefont {Farhi}, \citenamefont
  {Gutmann},\ and\ \citenamefont {Neven}}]{brandao2018fixed}%
  \BibitemOpen
  \bibfield  {author} {\bibinfo {author} {\bibfnamefont {F.~G.}\ \bibnamefont
  {Brandao}}, \bibinfo {author} {\bibfnamefont {M.}~\bibnamefont {Broughton}},
  \bibinfo {author} {\bibfnamefont {E.}~\bibnamefont {Farhi}}, \bibinfo
  {author} {\bibfnamefont {S.}~\bibnamefont {Gutmann}},\ and\ \bibinfo {author}
  {\bibfnamefont {H.}~\bibnamefont {Neven}},\ }\bibfield  {title} {\bibinfo
  {title} {For fixed control parameters the quantum approximate optimization
  algorithm's objective function value concentrates for typical instances},\
  }\href {https://doi.org/10.48550/arXiv.1812.04170} {\bibfield  {journal}
  {\bibinfo  {journal} {arXiv preprint arXiv:1812.04170}\ } (\bibinfo {year}
  {2018})}\BibitemShut {NoStop}%
\bibitem [{\citenamefont {Sim}\ \emph {et~al.}(2019)\citenamefont {Sim},
  \citenamefont {Johnson},\ and\ \citenamefont
  {Aspuru-Guzik}}]{sim2019expressibility}%
  \BibitemOpen
  \bibfield  {author} {\bibinfo {author} {\bibfnamefont {S.}~\bibnamefont
  {Sim}}, \bibinfo {author} {\bibfnamefont {P.~D.}\ \bibnamefont {Johnson}},\
  and\ \bibinfo {author} {\bibfnamefont {A.}~\bibnamefont {Aspuru-Guzik}},\
  }\bibfield  {title} {\bibinfo {title} {Expressibility and entangling
  capability of parameterized quantum circuits for hybrid quantum-classical
  algorithms},\ }\href {https://doi.org/https://doi.org/10.1002/qute.201900070}
  {\bibfield  {journal} {\bibinfo  {journal} {Advanced Quantum Technologies}\
  }\textbf {\bibinfo {volume} {2}},\ \bibinfo {pages} {1900070} (\bibinfo
  {year} {2019})}\BibitemShut {NoStop}%
\bibitem [{\citenamefont {Haug}\ \emph {et~al.}(2021)\citenamefont {Haug},
  \citenamefont {Bharti},\ and\ \citenamefont {Kim}}]{haug_quantum_geometry}%
  \BibitemOpen
  \bibfield  {author} {\bibinfo {author} {\bibfnamefont {T.}~\bibnamefont
  {Haug}}, \bibinfo {author} {\bibfnamefont {K.}~\bibnamefont {Bharti}},\ and\
  \bibinfo {author} {\bibfnamefont {M.}~\bibnamefont {Kim}},\ }\bibfield
  {title} {\bibinfo {title} {Capacity and quantum geometry of parametrized
  quantum circuits},\ }\href {https://doi.org/10.1103/PRXQuantum.2.040309}
  {\bibfield  {journal} {\bibinfo  {journal} {PRX Quantum}\ }\textbf {\bibinfo
  {volume} {2}},\ \bibinfo {pages} {040309} (\bibinfo {year}
  {2021})}\BibitemShut {NoStop}%
\bibitem [{\citenamefont {Bittel}\ and\ \citenamefont
  {Kliesch}(2021)}]{PhysRevLett.127.120502}%
  \BibitemOpen
  \bibfield  {author} {\bibinfo {author} {\bibfnamefont {L.}~\bibnamefont
  {Bittel}}\ and\ \bibinfo {author} {\bibfnamefont {M.}~\bibnamefont
  {Kliesch}},\ }\bibfield  {title} {\bibinfo {title} {Training variational
  quantum algorithms is {NP}-hard},\ }\href
  {https://doi.org/10.1103/PhysRevLett.127.120502} {\bibfield  {journal}
  {\bibinfo  {journal} {Phys. Rev. Lett.}\ }\textbf {\bibinfo {volume} {127}},\
  \bibinfo {pages} {120502} (\bibinfo {year} {2021})}\BibitemShut {NoStop}%
\bibitem [{\citenamefont {Bittel}\ \emph {et~al.}(2023)\citenamefont {Bittel},
  \citenamefont {Gharibian},\ and\ \citenamefont {Kliesch}}]{qcma_hard}%
  \BibitemOpen
  \bibfield  {author} {\bibinfo {author} {\bibfnamefont {L.}~\bibnamefont
  {Bittel}}, \bibinfo {author} {\bibfnamefont {S.}~\bibnamefont {Gharibian}},\
  and\ \bibinfo {author} {\bibfnamefont {M.}~\bibnamefont {Kliesch}},\
  }\bibfield  {title} {\bibinfo {title} {{The Optimal Depth of Variational
  Quantum Algorithms Is QCMA-Hard to Approximate}},\ }in\ \href
  {https://doi.org/10.4230/LIPIcs.CCC.2023.34} {\emph {\bibinfo {booktitle}
  {38th Computational Complexity Conference (CCC 2023)}}},\ \bibinfo {series}
  {Leibniz International Proceedings in Informatics (LIPIcs)}, Vol.\ \bibinfo
  {volume} {264},\ \bibinfo {editor} {edited by\ \bibinfo {editor}
  {\bibfnamefont {A.}~\bibnamefont {Ta-Shma}}}\ (\bibinfo  {publisher} {Schloss
  Dagstuhl -- Leibniz-Zentrum f{\"u}r Informatik},\ \bibinfo {address}
  {Dagstuhl, Germany},\ \bibinfo {year} {2023})\ pp.\ \bibinfo {pages}
  {34:1--34:24}\BibitemShut {NoStop}%
\bibitem [{\citenamefont {Wang}\ \emph {et~al.}(2023)\citenamefont {Wang},
  \citenamefont {Liu}, \citenamefont {Song}, \citenamefont {Gao}, \citenamefont
  {Qin},\ and\ \citenamefont {Wen}}]{wang2022quantum}%
  \BibitemOpen
  \bibfield  {author} {\bibinfo {author} {\bibfnamefont {S.-S.}\ \bibnamefont
  {Wang}}, \bibinfo {author} {\bibfnamefont {H.-L.}\ \bibnamefont {Liu}},
  \bibinfo {author} {\bibfnamefont {Y.-Q.}\ \bibnamefont {Song}}, \bibinfo
  {author} {\bibfnamefont {F.}~\bibnamefont {Gao}}, \bibinfo {author}
  {\bibfnamefont {S.-J.}\ \bibnamefont {Qin}},\ and\ \bibinfo {author}
  {\bibfnamefont {Q.-Y.}\ \bibnamefont {Wen}},\ }\bibfield  {title} {\bibinfo
  {title} {Quantum alternating operator ansatz for solving the minimum exact
  cover problem},\ }\href
  {https://doi.org/https://doi.org/10.1016/j.physa.2023.129089} {\bibfield
  {journal} {\bibinfo  {journal} {Physica A: Statistical Mechanics and its
  Applications}\ }\textbf {\bibinfo {volume} {626}},\ \bibinfo {pages} {129089}
  (\bibinfo {year} {2023})}\BibitemShut {NoStop}%
\bibitem [{\citenamefont {Bengtsson}\ \emph {et~al.}(2020)\citenamefont
  {Bengtsson}, \citenamefont {Vikst\aa{}l}, \citenamefont {Warren},
  \citenamefont {Svensson}, \citenamefont {Gu}, \citenamefont {Kockum},
  \citenamefont {Krantz}, \citenamefont {Kri\ifmmode~\check{z}\else
  \v{z}\fi{}an}, \citenamefont {Shiri}, \citenamefont {Svensson}, \citenamefont
  {Tancredi}, \citenamefont {Johansson}, \citenamefont {Delsing}, \citenamefont
  {Ferrini},\ and\ \citenamefont {Bylander}}]{bengtsson2020improved}%
  \BibitemOpen
  \bibfield  {author} {\bibinfo {author} {\bibfnamefont {A.}~\bibnamefont
  {Bengtsson}}, \bibinfo {author} {\bibfnamefont {P.}~\bibnamefont
  {Vikst\aa{}l}}, \bibinfo {author} {\bibfnamefont {C.}~\bibnamefont {Warren}},
  \bibinfo {author} {\bibfnamefont {M.}~\bibnamefont {Svensson}}, \bibinfo
  {author} {\bibfnamefont {X.}~\bibnamefont {Gu}}, \bibinfo {author}
  {\bibfnamefont {A.~F.}\ \bibnamefont {Kockum}}, \bibinfo {author}
  {\bibfnamefont {P.}~\bibnamefont {Krantz}}, \bibinfo {author} {\bibfnamefont
  {C.}~\bibnamefont {Kri\ifmmode~\check{z}\else \v{z}\fi{}an}}, \bibinfo
  {author} {\bibfnamefont {D.}~\bibnamefont {Shiri}}, \bibinfo {author}
  {\bibfnamefont {I.-M.}\ \bibnamefont {Svensson}}, \bibinfo {author}
  {\bibfnamefont {G.}~\bibnamefont {Tancredi}}, \bibinfo {author}
  {\bibfnamefont {G.}~\bibnamefont {Johansson}}, \bibinfo {author}
  {\bibfnamefont {P.}~\bibnamefont {Delsing}}, \bibinfo {author} {\bibfnamefont
  {G.}~\bibnamefont {Ferrini}},\ and\ \bibinfo {author} {\bibfnamefont
  {J.}~\bibnamefont {Bylander}},\ }\bibfield  {title} {\bibinfo {title}
  {Improved success probability with greater circuit depth for the quantum
  approximate optimization algorithm},\ }\href
  {https://doi.org/10.1103/PhysRevApplied.14.034010} {\bibfield  {journal}
  {\bibinfo  {journal} {Phys. Rev. Appl.}\ }\textbf {\bibinfo {volume} {14}},\
  \bibinfo {pages} {034010} (\bibinfo {year} {2020})}\BibitemShut {NoStop}%
\bibitem [{\citenamefont {Basso}\ \emph {et~al.}(2022)\citenamefont {Basso},
  \citenamefont {Gamarnik}, \citenamefont {Mei},\ and\ \citenamefont
  {Zhou}}]{basso2022performance}%
  \BibitemOpen
  \bibfield  {author} {\bibinfo {author} {\bibfnamefont {J.}~\bibnamefont
  {Basso}}, \bibinfo {author} {\bibfnamefont {D.}~\bibnamefont {Gamarnik}},
  \bibinfo {author} {\bibfnamefont {S.}~\bibnamefont {Mei}},\ and\ \bibinfo
  {author} {\bibfnamefont {L.}~\bibnamefont {Zhou}},\ }\bibfield  {title}
  {\bibinfo {title} {Performance and limitations of the {QAOA} at constant
  levels on large sparse hypergraphs and spin glass models},\ }in\ \href
  {https://doi.org/10.1109/FOCS54457.2022.00039} {\emph {\bibinfo {booktitle}
  {2022 IEEE 63rd Annual Symposium on Foundations of Computer Science
  (FOCS)}}}\ (\bibinfo  {publisher} {IEEE Computer Society},\ \bibinfo {year}
  {2022})\ pp.\ \bibinfo {pages} {335--343}\BibitemShut {NoStop}%
\bibitem [{\citenamefont {Brandhofer}\ \emph {et~al.}(2023)\citenamefont
  {Brandhofer}, \citenamefont {Braun}, \citenamefont {Dehn}, \citenamefont
  {Hellstern}, \citenamefont {H{\"u}ls}, \citenamefont {Ji}, \citenamefont
  {Polian}, \citenamefont {Bhatia},\ and\ \citenamefont
  {Wellens}}]{brandhofer2023benchmarking}%
  \BibitemOpen
  \bibfield  {author} {\bibinfo {author} {\bibfnamefont {S.}~\bibnamefont
  {Brandhofer}}, \bibinfo {author} {\bibfnamefont {D.}~\bibnamefont {Braun}},
  \bibinfo {author} {\bibfnamefont {V.}~\bibnamefont {Dehn}}, \bibinfo {author}
  {\bibfnamefont {G.}~\bibnamefont {Hellstern}}, \bibinfo {author}
  {\bibfnamefont {M.}~\bibnamefont {H{\"u}ls}}, \bibinfo {author}
  {\bibfnamefont {Y.}~\bibnamefont {Ji}}, \bibinfo {author} {\bibfnamefont
  {I.}~\bibnamefont {Polian}}, \bibinfo {author} {\bibfnamefont {A.~S.}\
  \bibnamefont {Bhatia}},\ and\ \bibinfo {author} {\bibfnamefont
  {T.}~\bibnamefont {Wellens}},\ }\bibfield  {title} {\bibinfo {title}
  {Benchmarking the performance of portfolio optimization with {QAOA}},\ }\href
  {https://doi.org/10.1007/s11128-022-03766-5} {\bibfield  {journal} {\bibinfo
  {journal} {Quantum Information Processing}\ }\textbf {\bibinfo {volume}
  {22}},\ \bibinfo {pages} {1} (\bibinfo {year} {2023})}\BibitemShut {NoStop}%
\bibitem [{\citenamefont {Kremenetski}\ \emph {et~al.}(2021)\citenamefont
  {Kremenetski}, \citenamefont {Hogg}, \citenamefont {Hadfield}, \citenamefont
  {Cotton},\ and\ \citenamefont {Tubman}}]{kremenetski2021quantum}%
  \BibitemOpen
  \bibfield  {author} {\bibinfo {author} {\bibfnamefont {V.}~\bibnamefont
  {Kremenetski}}, \bibinfo {author} {\bibfnamefont {T.}~\bibnamefont {Hogg}},
  \bibinfo {author} {\bibfnamefont {S.}~\bibnamefont {Hadfield}}, \bibinfo
  {author} {\bibfnamefont {S.~J.}\ \bibnamefont {Cotton}},\ and\ \bibinfo
  {author} {\bibfnamefont {N.~M.}\ \bibnamefont {Tubman}},\ }\bibfield  {title}
  {\bibinfo {title} {Quantum alternating operator ansatz ({QAOA}) phase
  diagrams and applications for quantum chemistry},\ }\href
  {https://doi.org/10.48550/arXiv.2108.13056} {\bibfield  {journal} {\bibinfo
  {journal} {arXiv preprint arXiv:2108.13056}\ } (\bibinfo {year}
  {2021})}\BibitemShut {NoStop}%
\bibitem [{\citenamefont {Mustafa}\ \emph {et~al.}(2022)\citenamefont
  {Mustafa}, \citenamefont {Morapakula}, \citenamefont {Jain},\ and\
  \citenamefont {Ganguly}}]{mustafa2022variational}%
  \BibitemOpen
  \bibfield  {author} {\bibinfo {author} {\bibfnamefont {H.}~\bibnamefont
  {Mustafa}}, \bibinfo {author} {\bibfnamefont {S.~N.}\ \bibnamefont
  {Morapakula}}, \bibinfo {author} {\bibfnamefont {P.}~\bibnamefont {Jain}},\
  and\ \bibinfo {author} {\bibfnamefont {S.}~\bibnamefont {Ganguly}},\
  }\bibfield  {title} {\bibinfo {title} {Variational quantum algorithms for
  chemical simulation and drug discovery},\ }in\ \href
  {https://doi.org/10.1109/TQCEBT54229.2022.10041453} {\emph {\bibinfo
  {booktitle} {2022 International Conference on Trends in Quantum Computing and
  Emerging Business Technologies (TQCEBT)}}}\ (\bibinfo {organization} {IEEE},\
  \bibinfo {year} {2022})\ pp.\ \bibinfo {pages} {1--8}\BibitemShut {NoStop}%
\bibitem [{\citenamefont {Mesman}\ \emph {et~al.}(2021)\citenamefont {Mesman},
  \citenamefont {Al-Ars},\ and\ \citenamefont {M{\"o}ller}}]{mesman2021qpack}%
  \BibitemOpen
  \bibfield  {author} {\bibinfo {author} {\bibfnamefont {K.}~\bibnamefont
  {Mesman}}, \bibinfo {author} {\bibfnamefont {Z.}~\bibnamefont {Al-Ars}},\
  and\ \bibinfo {author} {\bibfnamefont {M.}~\bibnamefont {M{\"o}ller}},\
  }\bibfield  {title} {\bibinfo {title} {Qpack: Quantum approximate
  optimization algorithms as universal benchmark for quantum computers},\
  }\href {https://doi.org/10.48550/arXiv.2103.17193} {\bibfield  {journal}
  {\bibinfo  {journal} {arXiv preprint arXiv:2103.17193}\ } (\bibinfo {year}
  {2021})}\BibitemShut {NoStop}%
\bibitem [{\citenamefont {Zhou}\ \emph {et~al.}(2020)\citenamefont {Zhou},
  \citenamefont {Wang}, \citenamefont {Choi}, \citenamefont {Pichler},\ and\
  \citenamefont {Lukin}}]{leo2020quantum}%
  \BibitemOpen
  \bibfield  {author} {\bibinfo {author} {\bibfnamefont {L.}~\bibnamefont
  {Zhou}}, \bibinfo {author} {\bibfnamefont {S.-T.}\ \bibnamefont {Wang}},
  \bibinfo {author} {\bibfnamefont {S.}~\bibnamefont {Choi}}, \bibinfo {author}
  {\bibfnamefont {H.}~\bibnamefont {Pichler}},\ and\ \bibinfo {author}
  {\bibfnamefont {M.~D.}\ \bibnamefont {Lukin}},\ }\bibfield  {title} {\bibinfo
  {title} {Quantum approximate optimization algorithm: Performance, mechanism,
  and implementation on near-term devices},\ }\href
  {https://doi.org/10.1103/PhysRevX.10.021067} {\bibfield  {journal} {\bibinfo
  {journal} {Phys. Rev. X}\ }\textbf {\bibinfo {volume} {10}},\ \bibinfo
  {pages} {021067} (\bibinfo {year} {2020})}\BibitemShut {NoStop}%
\bibitem [{\citenamefont {Lubinski}\ \emph {et~al.}(2023)\citenamefont
  {Lubinski}, \citenamefont {Coffrin}, \citenamefont {McGeoch}, \citenamefont
  {Sathe}, \citenamefont {Apanavicius},\ and\ \citenamefont
  {Neira}}]{lubinski2023optimization}%
  \BibitemOpen
  \bibfield  {author} {\bibinfo {author} {\bibfnamefont {T.}~\bibnamefont
  {Lubinski}}, \bibinfo {author} {\bibfnamefont {C.}~\bibnamefont {Coffrin}},
  \bibinfo {author} {\bibfnamefont {C.}~\bibnamefont {McGeoch}}, \bibinfo
  {author} {\bibfnamefont {P.}~\bibnamefont {Sathe}}, \bibinfo {author}
  {\bibfnamefont {J.}~\bibnamefont {Apanavicius}},\ and\ \bibinfo {author}
  {\bibfnamefont {D.~E.~B.}\ \bibnamefont {Neira}},\ }\bibfield  {title}
  {\bibinfo {title} {Optimization applications as quantum performance
  benchmarks},\ }\href {https://doi.org/10.48550/arXiv.2302.02278} {\bibfield
  {journal} {\bibinfo  {journal} {arXiv preprint arXiv:2302.02278}\ } (\bibinfo
  {year} {2023})}\BibitemShut {NoStop}%
\bibitem [{\citenamefont {Vijendran}(2023)}]{V_Vijendran_XQAOA-Dataset_2023}%
  \BibitemOpen
  \bibfield  {author} {\bibinfo {author} {\bibfnamefont {V.}~\bibnamefont
  {Vijendran}},\ }\href {https://github.com/vijeycreative/XQAOA-Dataset}
  {\bibinfo {title} {{XQAOA-Dataset}}} (\bibinfo {year} {2023})\BibitemShut
  {NoStop}%
\bibitem [{\citenamefont {Liu}\ and\ \citenamefont
  {Nocedal}(1989)}]{liu1989limited}%
  \BibitemOpen
  \bibfield  {author} {\bibinfo {author} {\bibfnamefont {D.~C.}\ \bibnamefont
  {Liu}}\ and\ \bibinfo {author} {\bibfnamefont {J.}~\bibnamefont {Nocedal}},\
  }\bibfield  {title} {\bibinfo {title} {On the limited memory {BFGS} method
  for large scale optimization},\ }\href {https://doi.org/10.1007/BF01589116}
  {\bibfield  {journal} {\bibinfo  {journal} {Mathematical programming}\
  }\textbf {\bibinfo {volume} {45}},\ \bibinfo {pages} {503} (\bibinfo {year}
  {1989})}\BibitemShut {NoStop}%
\bibitem [{\citenamefont {Gerber}\ and\ \citenamefont
  {Furrer}(2018)}]{gerber2019optimparallel}%
  \BibitemOpen
  \bibfield  {author} {\bibinfo {author} {\bibfnamefont {F.}~\bibnamefont
  {Gerber}}\ and\ \bibinfo {author} {\bibfnamefont {R.}~\bibnamefont
  {Furrer}},\ }\bibfield  {title} {\bibinfo {title} {optim{P}arallel: an {R}
  package providing parallel versions of the gradient-based optimization
  methods of optim()},\ }\href {https://doi.org/10.48550/arXiv.1804.11058}
  {\bibfield  {journal} {\bibinfo  {journal} {arXiv preprint arXiv:1804.11058}\
  } (\bibinfo {year} {2018})}\BibitemShut {NoStop}%
\bibitem [{\citenamefont {Keller}\ and\ \citenamefont
  {Trotter}(2017)}]{keller2017applied}%
  \BibitemOpen
  \bibfield  {author} {\bibinfo {author} {\bibfnamefont {M.~T.}\ \bibnamefont
  {Keller}}\ and\ \bibinfo {author} {\bibfnamefont {W.~T.}\ \bibnamefont
  {Trotter}},\ }\href {https://www.appliedcombinatorics.org/appcomb} {\emph
  {\bibinfo {title} {Applied Combinatorics}}}\ (\bibinfo  {publisher} {Mitchel
  T. Keller, William T. Trotter},\ \bibinfo {year} {2017})\BibitemShut
  {NoStop}%
\end{thebibliography}%

\appendix


\section{Goemans-Williamson (GW) Algorithm} \label{gw_algorithm}

The Goemans-Williamson algorithm~\cite{goemans1995improved} is a polynomial-time approximation algorithm for approximately solving the MaxCut problem. The algorithm works by constructing a semidefinite programming relaxation of the MaxCut problem and then rounding the solution to get a near-optimal cut in the original graph.

Recall in \cref{maxcut} that the MaxCut problem can be formulated as a binary quadratic program of the form
\begin{equation}
\begin{array}{ll@{}ll}
\text{Maximise}  & \displaystyle\sum\limits_{1 \leq i < j \leq n} \frac{1}{2} w_{ij} \left(1 - y_i y_j \right) &\\[0.5cm]
\text{s.t}& y_i \in \{-1, 1 \} \quad \forall i \in V.
\end{array}
\label{maxcut_quad_program}
\end{equation}
We can relax this program to a vector program by allowing the binary variables $y_i$ to be $n$-dimensional vector variables $\boldsymbol{v}_i$ that lie on the $n$-dimensional unit sphere $S_n$. Replacing the product of scalar terms in \cref{maxcut_quad_program} with the corresponding inner product, we obtain the following vector program for MaxCut.
\begin{equation}
\begin{array}{ll@{}ll}
\text{Maximise}  & \displaystyle\sum\limits_{1 \leq i < j \leq n} \frac{1}{2} w_{ij} \left(1 - \boldsymbol{v}_i \cdot \boldsymbol{v}_j \right) &\\[0.5cm]
\text{s.t}& \boldsymbol{v}_i \in S_n \quad \forall i \in V.
\end{array}
\label{maxcut_vec_program}
\end{equation}
This relaxed vector program can be efficiently solved by semidefinite programming, which allows us to obtain a set of optimal vectors $\boldsymbol{v}^*_i$ for each node in the original graph. The Goeman-Williamson algorithm then uses a random $n$-dimensional vector $\boldsymbol{r}$ from $S_n$ to partition the vertices into two sets by assigning $\text{sign}(\boldsymbol{r} \cdot \boldsymbol{v}^*_i)$ to each node. The sign function returns 1 for non-negative inputs and -1 elsewhere, meaning that each node's rounding depends on its position relative to the hyperplane defined by $\boldsymbol{r}$ that passes through the origin. The probability of the hyperplane rounding cutting an edge $\{i, j\}$ is proportional to the angle between the vectors and can be expressed as 
\begin{equation}
    \text{Pr}[\text{sign}(\boldsymbol{r} \cdot \boldsymbol{v}_i) \neq \text{sign}(\boldsymbol{r} \cdot \boldsymbol{v}_j)]  = \dfrac{\arccos (\boldsymbol{v}_i \cdot \boldsymbol{v}_j)}{\pi}.
\end{equation}
The expected weight of the cut found by the algorithm is calculated by adding up the expected contributions of each edge, where the contribution of an individual edge is its probability of being cut. We can write the sum as follows
\begin{equation}
\begin{split}
	\mathbb{E}[W] &= \sum_{1 \leq i < j \leq n}^n w_{ij} \text{Pr}[\text{sign}(\boldsymbol{r} \cdot \boldsymbol{v}_i) \neq \text{sign}(\boldsymbol{r} \cdot \boldsymbol{v}_j)] \\
	&= \frac{1}{\pi} \sum_{1 \leq i < j \leq n}^n w_{ij} \arccos (\boldsymbol{v}_i \cdot \boldsymbol{v}_j).
\end{split}
\label{expected_W}
\end{equation}
To find the approximation ratio, we need to compare the expected weight of the cut produced by the algorithm to the optimal cut. This is done by comparing the ratio $\alpha$ of individual edge contributions for each edge $\{i, j\}$ in \cref{expected_W} and \cref{maxcut_vec_program} and finding the minimum value:
\begin{equation}
\begin{split}
    \alpha &= \dfrac{\arccos{(\boldsymbol{v}_i \cdot \boldsymbol{v}_j)}}{\pi} \dfrac{2}{1-\boldsymbol{v}_i \cdot \boldsymbol{v}_j} \\
    & = \dfrac{\theta}{\pi} \dfrac{2}{1-\cos \theta},
\end{split}
\end{equation}
where $\theta=\arccos{(\boldsymbol{v}_i \cdot \boldsymbol{v}_j)}$ is the angle between the vectors $\boldsymbol{v_i}$ and $\boldsymbol{v_j}$. Minimising the above expression, we get
\begin{equation}
    \alpha = \min_{0 \leq \theta \leq \pi} \dfrac{2}{\pi} \dfrac{\theta}{1 - \cos \theta} \approx 0.87856.
\end{equation}
Having determined that each edge's contribution to the cut is expected to be no less than $0.87856$ of the optimal value for $\theta = 2.331122$, we can use the linearity of expectation to conclude that the total expected value is also no less than $0.87856$ of the optimal value. If the Unique Games Conjecture~\cite{khot2002power, khot2007optimal, 5497893, khot2015unique} proves true, this method offers the strongest possible guarantee that any classical algorithm can achieve in polynomial time.

\section{Parallel-LBFGS Algorithm} \label{plbfgs}

\begin{figure}[htpb]
\vspace{-0.5cm}
\centering
\includegraphics[width=0.9\columnwidth]{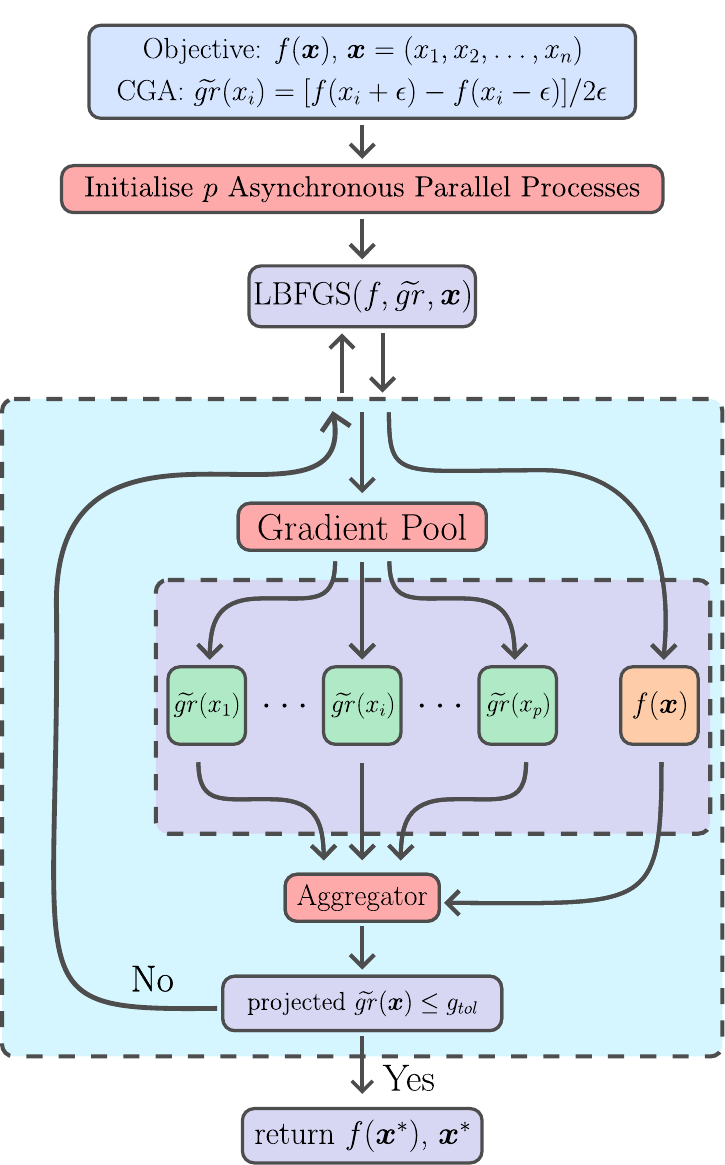}
    \captionsetup{justification=raggedright, singlelinecheck=false}
    \caption{The Parallel-LBFGS algorithm and CGA method are used to optimise an objective function. The objective function $f(\boldsymbol{x})$, gradient function $\widetilde{gr}(x_i)$, and initial points $\boldsymbol{x} = (x_1, x_2, \dots, x_n)$ are provided to the LBFGS algorithm, which is initialised with $p$ asynchronous processes. A gradient pool distributes the gradient calculation across $p$ available processors. When the LBFGS algorithm calls either the objective or gradient function, the wrapper interface (blue box) begins evaluating the objective function and all the gradients in parallel. The results are then combined and returned to the LBFGS algorithm. This process is repeated until the optimisation converges to a solution (i.e., when all the gradients are less than or equal to the specified gradient tolerance value).}
    \label{Parallel_LBFGS}
\end{figure}

\begin{figure*}[htbp]
    \begin{subfigure}[b]{0.325\textwidth}
        \includegraphics[width=\textwidth]{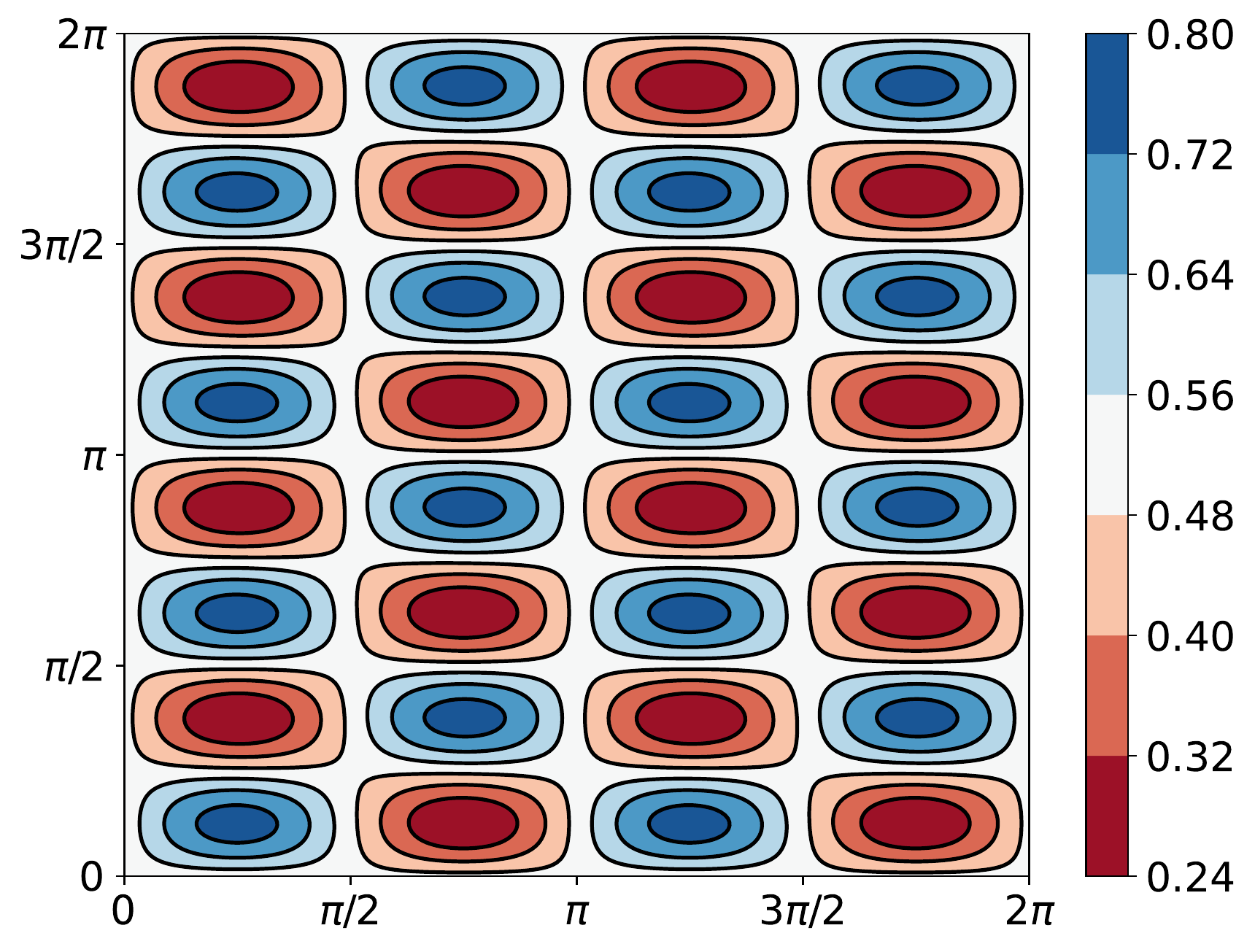}
        \caption{2-Regular Graph}
    \end{subfigure}
    \hfill
    \begin{subfigure}[b]{0.325\textwidth}
        \includegraphics[width=\textwidth]{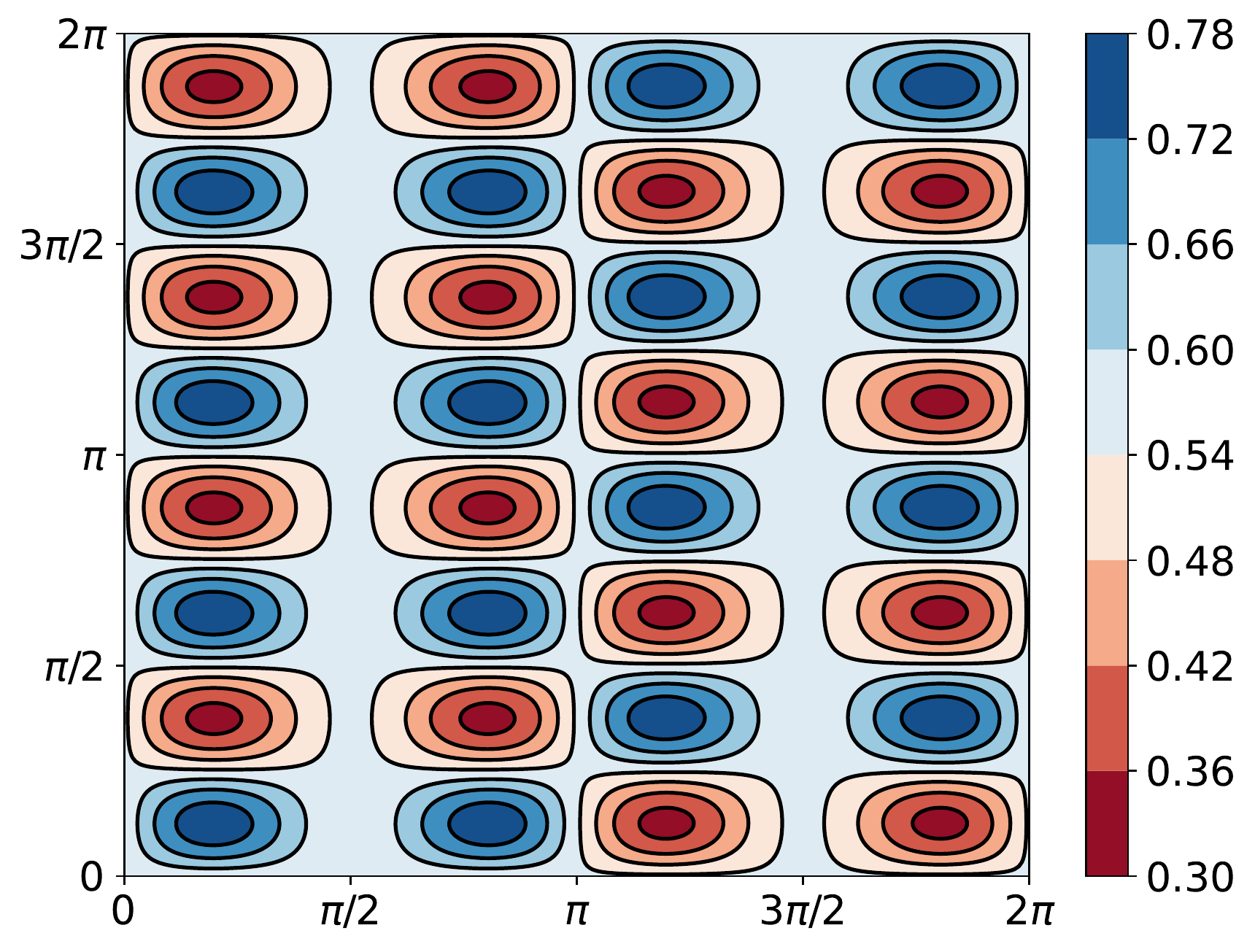}
        \caption{3-Regular Graph}
    \end{subfigure}
    \hfill
    \begin{subfigure}[b]{0.325\textwidth}
        \includegraphics[width=\textwidth]{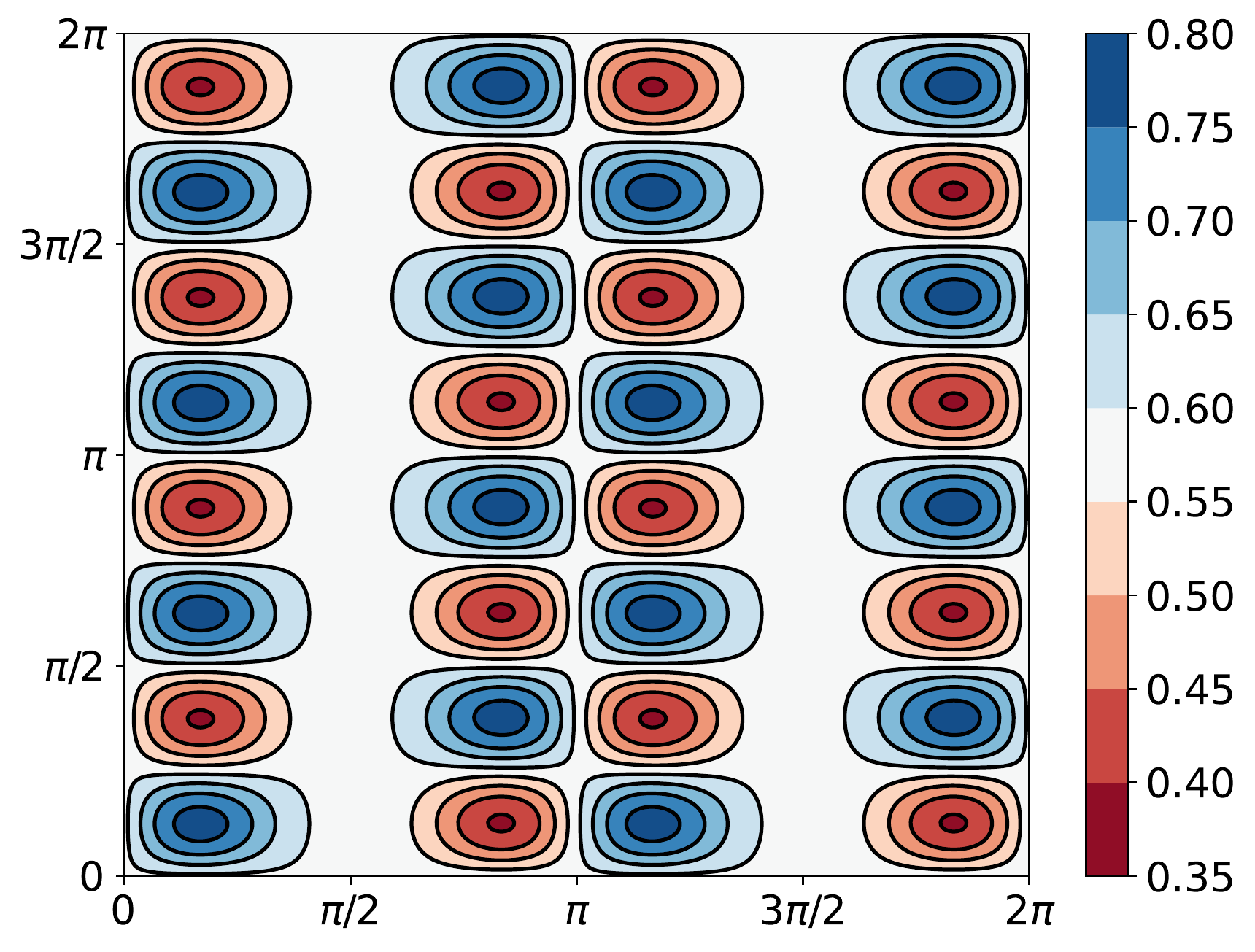}
        \caption{4-Regular Graph}
    \end{subfigure}
    \hfill
    \begin{subfigure}[b]{0.325\textwidth}
        \includegraphics[width=\textwidth]{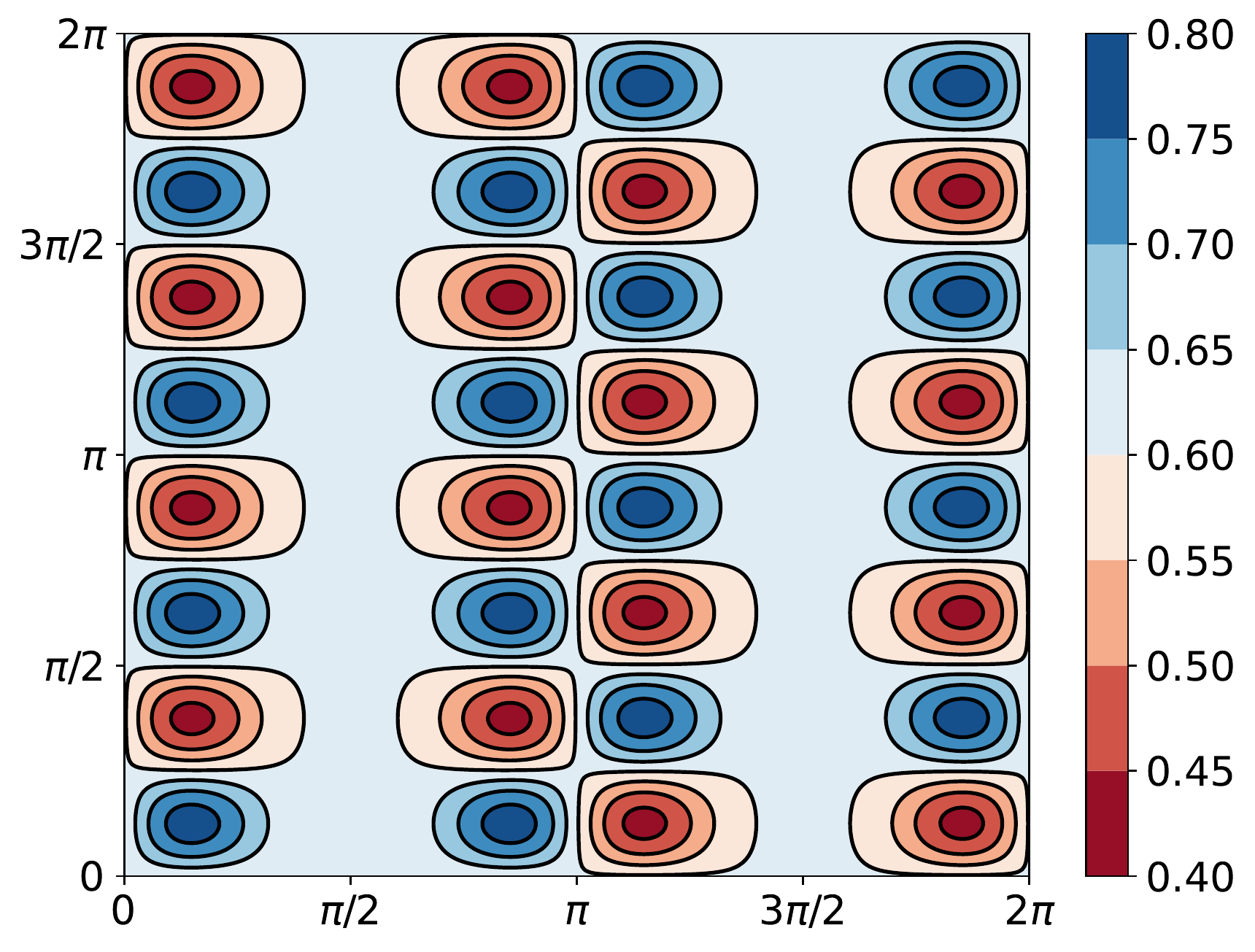}
        \caption{5-Regular Graph}
    \end{subfigure}
    \hfill
    \begin{subfigure}[b]{0.325\textwidth}
        \includegraphics[width=\textwidth]{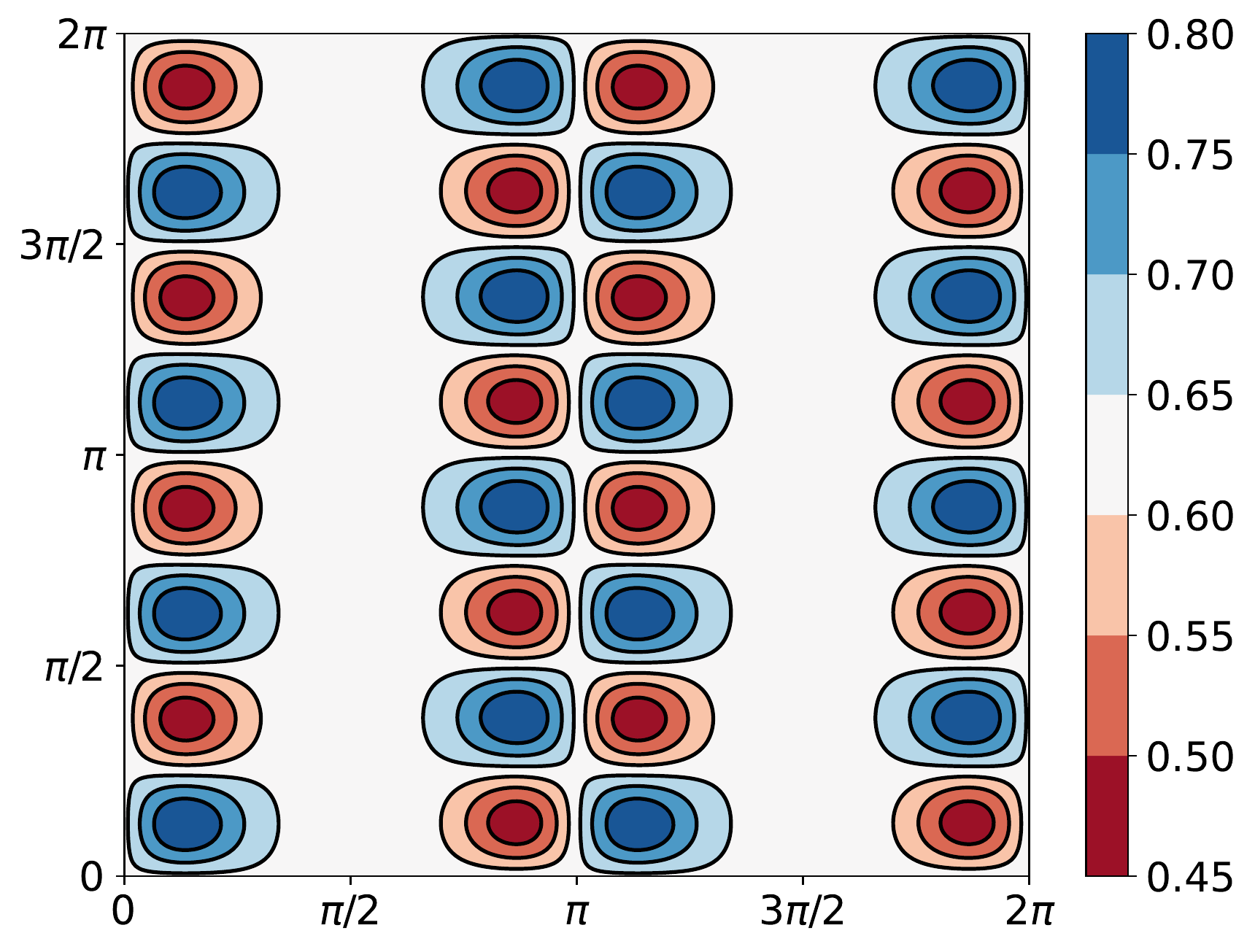}
        \caption{6-Regular Graph}
    \end{subfigure}
    \hfill
    \begin{subfigure}[b]{0.325\textwidth}
        \includegraphics[width=\textwidth]{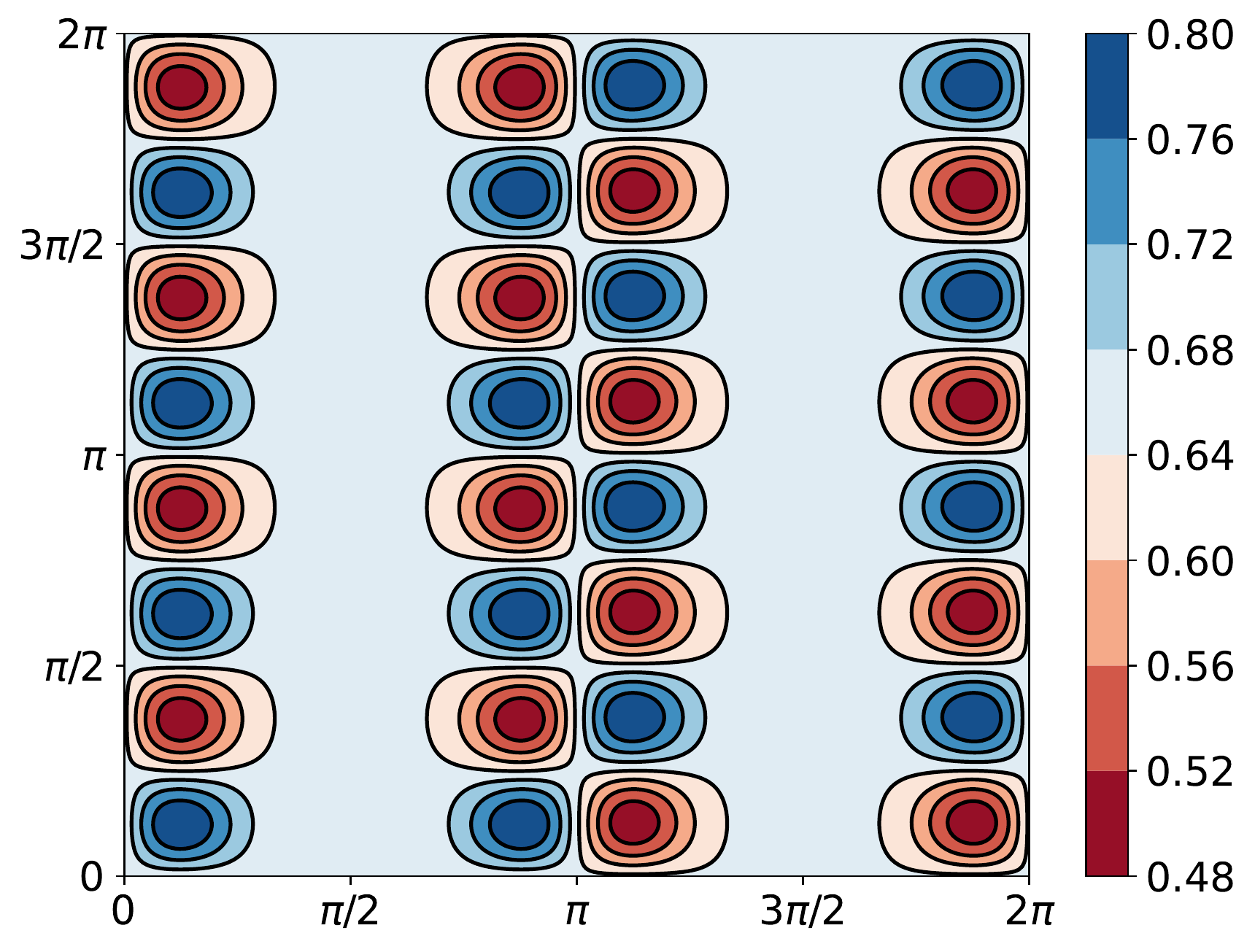}
        \caption{7-Regular Graph}
    \end{subfigure}
    \hfill
    \begin{subfigure}[b]{0.325\textwidth}
        \includegraphics[width=\textwidth]{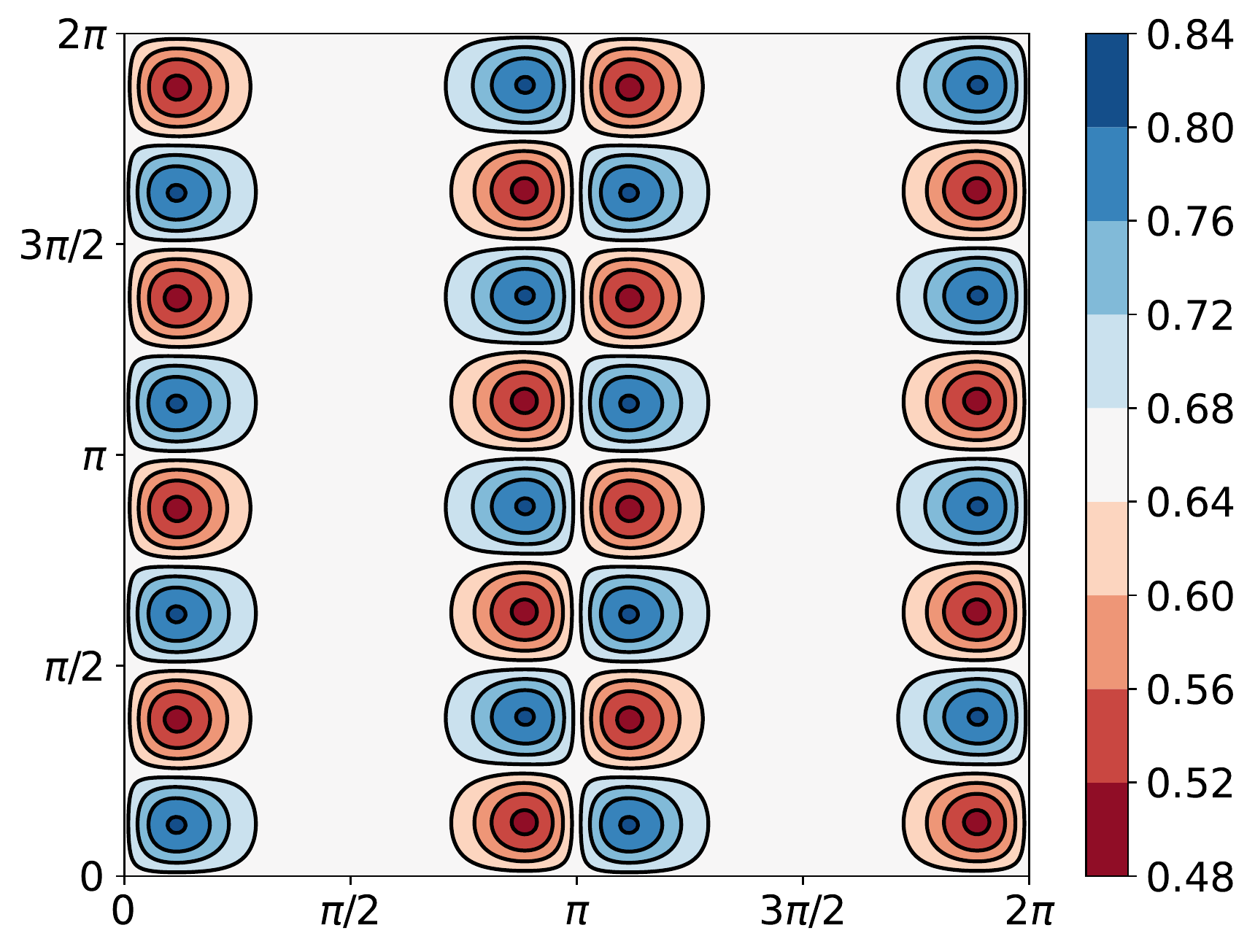}
        \caption{8-Regular Graph}
    \end{subfigure}
    \hfill
    \begin{subfigure}[b]{0.325\textwidth}
        \includegraphics[width=\textwidth]{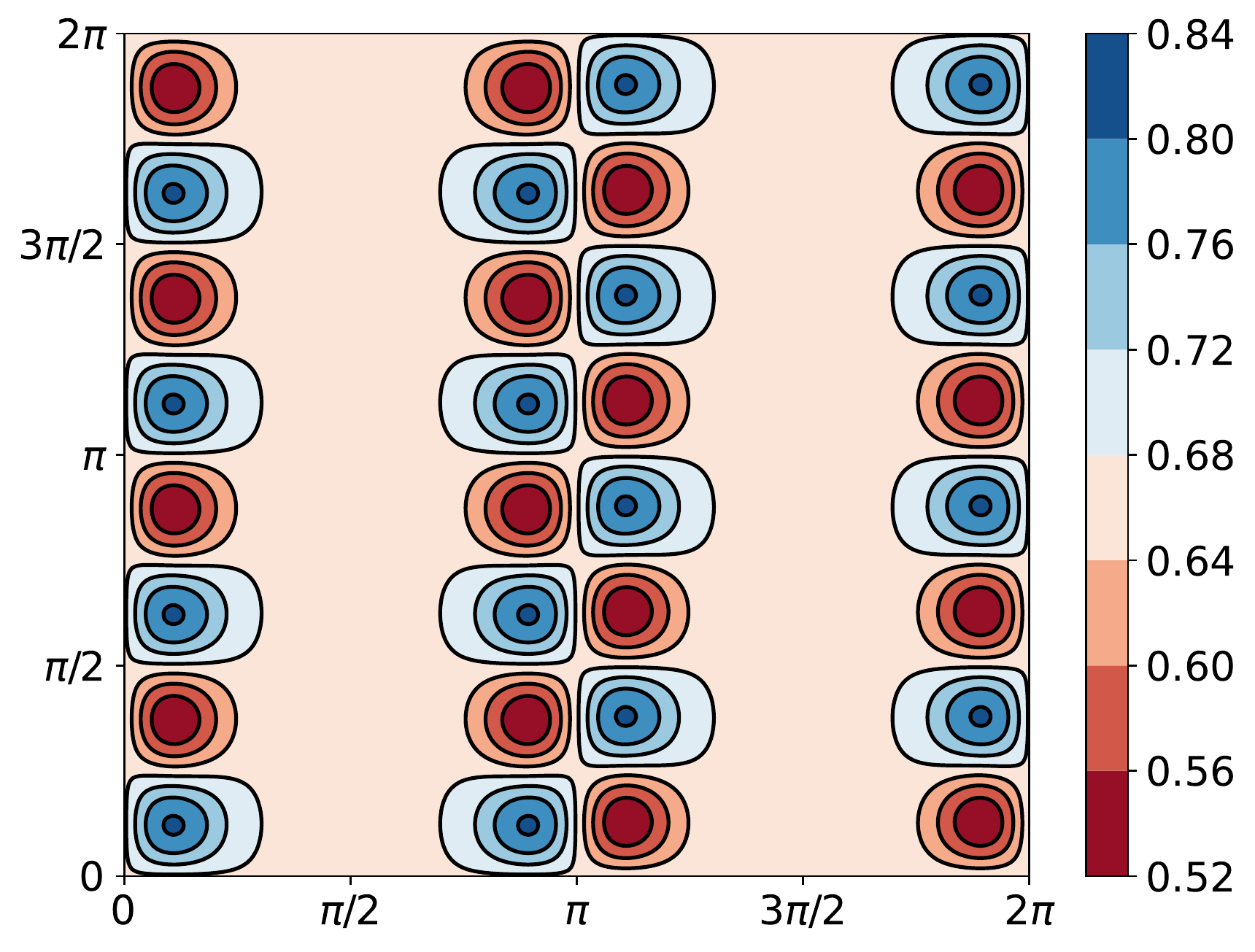}
        \caption{9-Regular Graph}
    \end{subfigure}
    \hfill
    \begin{subfigure}[b]{0.325\textwidth}
        \includegraphics[width=\textwidth]{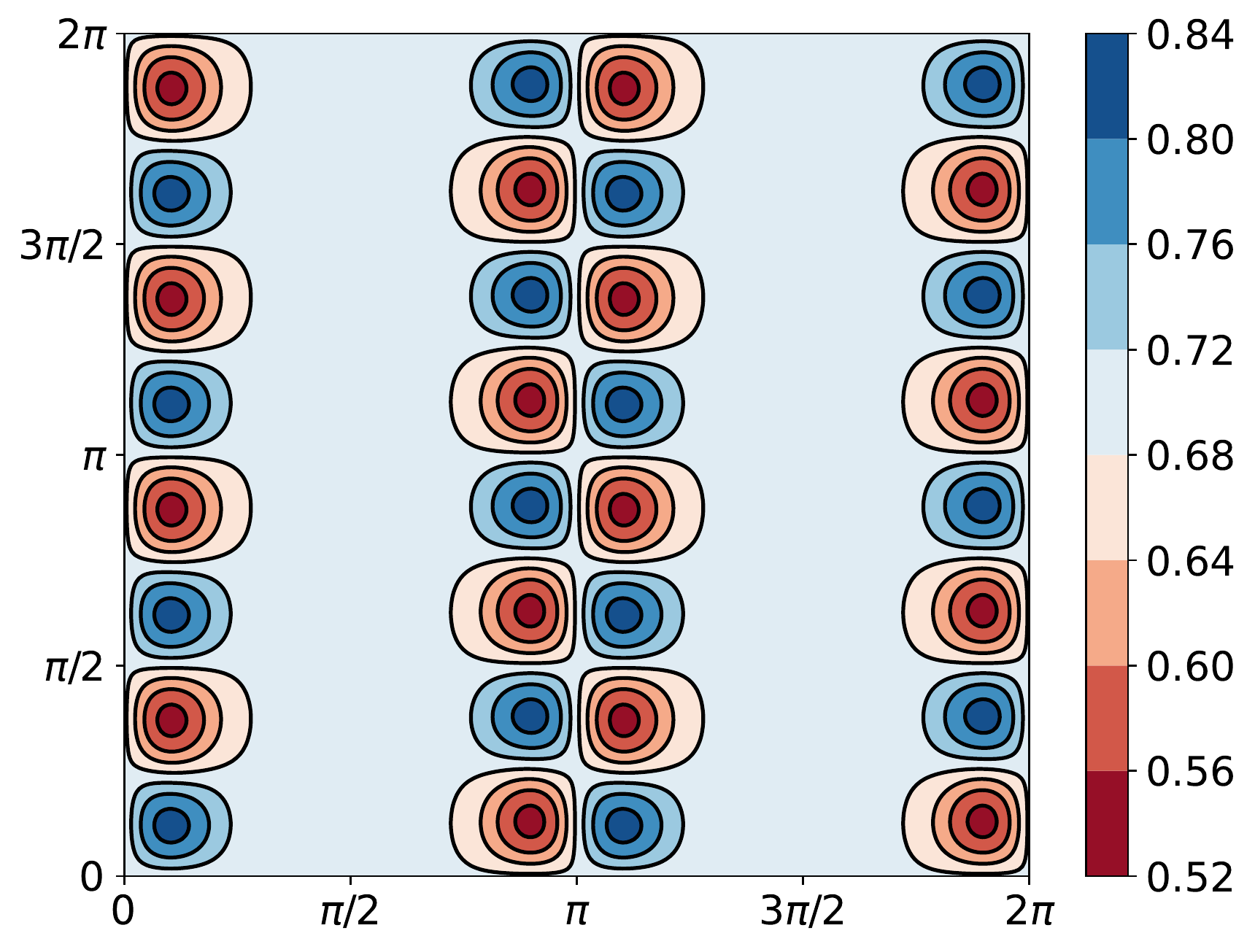}
        \caption{10-Regular Graph}
    \end{subfigure}
    \captionsetup{justification=raggedright, singlelinecheck=false}
    \caption{The cost function for QAOA$_1$ ansatz for $D$-regular graphs with 128 vertices is plotted on contour plots for $2 \leq D \leq 10$, with the $\beta$ and $\gamma$ angles represented on the $x$ and $y$-axis, respectively. The dark blue (red) regions represent points that maximise (minimise) the objective function. As the degree of the graph increases, the size of the barren plateaus in the cost function also increases. The barren plateaus are those lightly coloured regions without any contour lines. Despite this increase, the locations of the maxima and minima remain almost at similar locations for $0 \leq \beta \leq \frac{\pi}{4}$ and $\frac{7 \pi}{4} \leq \beta \leq 2 \pi$, alternating between degrees at $\frac{3 \pi}{4} \leq \beta \leq \frac{5 \pi}{4}$.}
    \label{qaoa_contour_plots}
\end{figure*}

The XQAOA$_1$ and MA-QAOA$_1$ methods require more classical effort to identify angles that maximise the approximation ratio due to the presence of more variables to optimise. To compute the angles for QAOA$_1$, MA-QAOA$_1$, and XQAOA$_1$ on the collection of graphs, we used a parallel implementation of the Limited-Memory Broyden-Fletcher-Goldfarb-Shanno (LBFGS)~\cite{liu1989limited} algorithm. This implementation, similar to that described in \cite{gerber2019optimparallel}, utilises parallelism to compute each variable's approximate gradients and function in parallel through the use of a wrapper class that interfaces with the standard LBFGS code. Approximate gradients are calculated using a numeric central difference gradient approximation (CGA), which requires $2n$ evaluations of the objective function if the function has $n$ parameters. The LBFGS algorithm sequentially evaluates the objective function $1 + 2n$ times per iteration. When run on $p$ available processor cores, the Parallel-LBFGS algorithm can evaluate all objective function calls in parallel, reducing the running time by a factor of $p$. \Cref{Parallel_LBFGS} illustrates the operation of the Parallel-LBFGS algorithm.

\section{Barren-Plateau Free Classical Optimisation of QAOA Ansatz} \label{qaoa_optimization}

It has been previously demonstrated in~\cite{brandao2018fixed} that when the problem instance comes from a reasonable distribution, the cost landscape and the optimal parameters of the QAOA$_1$ ansatz are independent of the specific instance. This means that, for typical instances, the value of the objective function and optimal parameters are nearly the same. As a result, a strategy for finding good parameters is to take one instance of the problem and invest time and resources into finding good parameters. Although this may be computationally expensive, once this has been done, these same parameter values will result in good cost function values on other randomly chosen instances. In other words, the overall cost of solving multiple instances becomes smaller as the number of instances increases. In the case of the MaxCut problem on $D$-regular graphs, all graphs have this property: all vertices are connected to $D$ other vertices. This same reasoning will apply to other combinatorial search problems that have a restriction where the number of clauses in which any variable can appear does not grow with $n$ or at least grows only slowly with high probability.

While we are not focused on reducing the computation cost of optimising the QAOA$_1$ ansatz for a group of graphs, we are interested in identifying the location of the barren plateaus in the cost landscape in order to avoid choosing initial points in these areas. \Cref{qaoa_contour_plots} shows the contour plots of the optimisation landscape of the QAOA$_1$ ansatz for regular graphs with 128 vertices and degrees ranging from 2 to 10. The flat lightly coloured regions without contour lines are the barren plateaus, and it is important to avoid these areas as initial points for the classical optimiser, as failing to do so may result in the optimiser converging to a suboptimal solution and not providing the maximal solution possible using the QAOA$_1$ ansatz. Choosing good initial points is especially important for regular graphs with larger degrees, as the barren plateau's size increases with the graph's degree.

In the contour plots shown in \cref{qaoa_contour_plots}, the locations of the maxima and minima remain nearly the same for $0 \leq \beta \leq \frac{\pi}{4}$ and $\frac{7 \pi}{4} \leq \beta \leq 2 \pi$, alternating between degrees at $\frac{3 \pi}{4} \leq \beta \leq \frac{5 \pi}{4}$. Therefore, to help the classical optimiser converge on the best possible solution, we can set the initial points $\gamma \in [0, \frac{\pi}{4}]$ and $\beta \in [0, \frac{\pi}{4}]$, which will allow the classical optimiser to always converge on the optimal solution in the lower left corner of the optimisation landscape.

\onecolumngrid
\section{Quantifying the Performance of XQAOA, MA-QAOA, and QAOA}

The goal of this appendix is to derive the analytical expressions given by \cref{xqaoa_full_exp}, \cref{ma_qaoa_full_exp}, and \cref{qaoa_formula} for the components of the $\XQAOA{XY}$, MA-QAOA and QAOA cost functions, respectively.

\subsection{Some Useful Identities}
\label{sec:useful_identities}

Before proceeding with derivations of the analytical formulas, let us first state and prove some identities that will be used in this paper. We denote the set of positive integers by $\mathbb Z^+ = \{1,2,\ldots\}$ and the set of $n$-bit strings by $\bbF^n = \{0,1\}^n$. When $n=1$, we write $ \bbF =\bbF^1 = \{0,1\}$.

\begin{lemma}\label{lemma1}
Let $f \in \mathbb{Z}^+$ and $x_1, x_2, \dots, x_f, y_1, y_2, \dots, y_f \in \mathbb{R}$ be real numbers. Then,
\begin{align}
	\prod_{i = 1}^{f} \cos(x_i - y_i) &= \sum_{\mu \in \bbF^f} \prod_{i = 1}^{f} \cos^{1-\mu_i}x_i \cos^{1-\mu_i}y_i \sin^{\mu_i}x_i \sin^{\mu_i}y_i ,
 \label{eq:cosine_identity_negative}\\
  \prod_{i = 1}^{f} \cos(x_i + y_i) &= \sum_{\mu \in \bbF^f} (-1)^{|\mu|} \prod_{i = 1}^{f} \cos^{1-\mu_i}x_i \cos^{1-\mu_i}y_i \sin^{\mu_i}x_i \sin^{\mu_i}y_i,
  \label{eq:cosine_identity_positive}
\end{align}
where $|\mu| = \sum_{i=1}^f \mu_i$.
\end{lemma}
\begin{proof} The first identity \cref{eq:cosine_identity_negative} follows from recognising that for any $a,b\in \mathbb R$, the sum $a+b$ can be written as the following sum of products:
\begin{align} \label{eq:product_to_sumss_trick}
a+b=\sum_{\mu \in 
\bbF}a^{1-\mu}b^\mu.
\end{align}
Applying \cref{eq:product_to_sumss_trick} to each term 
of the sum $\cos x_i \cos y_i + \sin x_i \sin y_i$ gives \begin{align*}
		\prod_{i = 1}^{f} \cos(x_i - y_i) &= \prod_{i = 1}^{f} \left(\cos x_i \cos y_i + \sin x_i \sin y_i \right) \\
		&= \prod_{i = 1}^{f} \sum_{\mu_i \in
  \bbF} \cos^{1-\mu_i}x_i \cos^{1-\mu_i}y_i \sin^{\mu_i}x_i \sin^{\mu_i}y_i \\
		&=  \sum_{\mu \in \bbF^f} \prod_{i = 1}^{f} \cos^{1-\mu_i}x_i \cos^{1-\mu_i}y_i \sin^{\mu_i}x_i \sin^{\mu_i}y_i.
\end{align*}
The identity in \cref{eq:cosine_identity_positive} follows immediately from replacing each $y_i$ in \cref{eq:cosine_identity_negative} with $-y_i$.
\end{proof}

By taking the sum (difference, respectively) of \cref{eq:cosine_identity_negative} and \cref{eq:cosine_identity_positive}, only the even (odd, respectively) terms remain. Hence, it follows that
\begin{lemma}
Let $f \in \mathbb{Z}^+$ and $x_1, x_2, \dots, x_f, y_1, y_2, \dots, y_f \in \mathbb{R}$. Then, we have that
	\begin{align}
\prod_{i = 1}^{f} \cos(x_i - y_i) + \prod_{i = 1}^{f} \cos(x_i + y_i) &= 2 \sum_{\substack{\mu \in \bbF^f \\ |\mu| \mathrm{\, even}}} \prod_{i = 1}^{f} \cos^{1-\mu_i}x_i \cos^{1-\mu_i}y_i \sin^{\mu_i}x_i \sin^{\mu_i}y_i,
\label{eq:cosine_identity_even}
\\
\prod_{i = 1}^{f} \cos(x_i - y_i) - \prod_{i = 1}^{f} \cos(x_i + y_i) &= 2 \sum_{\substack{\mu \in \bbF^f \\ |\mu| \mathrm{\, odd}}} \prod_{i = 1}^{f} \cos^{1-\mu_i}x_i \cos^{1-\mu_i}y_i \sin^{\mu_i}x_i \sin^{\mu_i}y_i.\label{eq:cosine_identity_odd}
\end{align}
 \label{lemma3}
\end{lemma}
We now make a remark about the above identities: while the right-hand sides of \cref{eq:cosine_identity_negative}, \cref{eq:cosine_identity_positive}, \cref{eq:cosine_identity_even}, and \cref{eq:cosine_identity_odd} each involves a sum over exponentially (in $f$) many terms (the cardinality of $\mathbb F_2^f$ is $2^f$), their left-hand sides involve just products of polynomially (in $f$) many terms. Hence, going from the right-hand sides of these identities to their left-hand sides results in exponential savings in computational cost. This will be useful for the expressions that we derive in \cref{xqaoa_proof1}. 

\subsection{Proof of \texorpdfstring{\cref{XQAOA_Full_Thm}}{3}} \label{xqaoa_proof1}

\begin{proof}[Proof of \cref{XQAOA_Full_Thm}]
Consider the XQAOA ansatz applied to MaxCut with $\ket{s} = \ket{+}^{\otimes n}$ and $Q =  e^{-i \bolddot\alpha Y}e^{-i \bolddot\beta X}e^{-i \bolddot\gamma C}$,
where $\bolddot\alpha Y = \sum_{i=1}^n \alpha_i Y_i$, $\bolddot\beta X = \sum_{i=1}^n \beta_i X_i$, and $\bolddot\gamma C = \sum_{\{u,v\}\in E} \gamma_{uv} C_{uv}
=\frac 12\sum_{\{u,v\}\in E} \gamma_{uv} w_{uv}(I-Z_u Z_v)
$.

Observe that 
\begin{equation}
    \langle C\rangle_{\mathrm{XY}} = \left\langle s\left|Q^{\dagger} C Q\right| s\right\rangle= \sum_{\{u,v\} \in E} \frac{w_{uv}}{2}-\frac{w_{uv}}{2} \langle s|Q^{\dagger} Z_{u} Z_{v} Q| s\rangle .
\end{equation}
Hence, to compute the expectation value $\langle C \rangle_{\mathrm{XY}}$, it suffices to compute each term
\begin{equation}
    \left\langle C_{u v}\right\rangle_{\mathrm{XY}}:=\frac{w_{uv}}{2}-\frac{w_{uv}}{2} \bra s Q^{\dagger} Z_{u} Z_{v} Q\ket  s
\label{XQAOA_Intro_Exp}
\end{equation}
in the sum separately. 

For the rest of this proof, we fix an edge $\{u,v\} \in E$. We shall evaluate the product $Q^{\dagger} Z_{u} Z_{v}Q$ in \cref{XQAOA_Intro_Exp} by conjugating the Pauli operator $Z_{u} Z_{v}$ by the mixing unitary and then by the problem unitary. By the commutation properties of the Pauli matrices $\Lambda$ and the fact that they satisfy $e^{-i \theta \Lambda}=\cos \theta I-i \sin \theta \Lambda$, it follows that most of the terms in the mixing unitary $e^{-i \bolddot\alpha Y}e^{-i \bolddot{\beta}X}=\prod_{j=1}^{n} e^{-i \alpha_j Y_{j}}e^{-i \beta_j X_{j}}$ commute through $Z_u Z_v$ and annihilate their inverses. Hence,
\begin{equation}
    \begin{split}
        e^{i \bolddot{\beta} X} e^{i \bolddot{\alpha} Y} Z_u Z_v  e^{-i \bolddot{\alpha} Y} e^{-i \bolddot{\beta} X} = \left( e^{i \beta_u X_u}e^{2i \alpha_u Y_u}e^{i \beta_u X_u}  Z_u  \right)  \otimes   \left(e^{i \beta_v X_v}e^{2i \alpha_v Y_v} e^{i \beta_v X_v} Z_v  \right) .
    \end{split}
    \label{eq:XQAOA_Mixing_Exp_0}
\end{equation}
The factors in the above Kronecker product can be expanded as follows. For $a\in\{u,v\}$,
\begin{equation}
    \begin{split}
         e^{i \beta_a X_a}e^{2i \alpha_a Y_a}e^{i \beta_a X_a}  Z_a &= \left( \cos \beta_a + i\sin \beta_a X_a \right)\left( \cos 2\alpha_a  + i\sin 2\alpha_a Y_u \right)\left( \cos \beta_a  + i\sin \beta_a X_a \right)Z_a \\
        &= \cos2\alpha_a  \cos2\beta_a Z_a + \cos2\alpha_a  \sin2\beta_a Y_a - \sin2\alpha_a X_a .
    \end{split}
\label{XQAOA_Mixing_Exp}
\end{equation}
Substituting this into equation \cref{eq:XQAOA_Mixing_Exp_0} gives
\begin{equation}
\begin{split}
e^{i \bolddot{\beta} X} e^{i \bolddot{\alpha} Y} Z_u Z_v  e^{-i \bolddot{\alpha} Y} e^{-i \bolddot{\beta} X}
&= \cos2\alpha_u  \cos2\beta_u \cos2\alpha_v  \cos2\beta_v Z_u Z_v + \cos2\alpha_u  \cos2\beta_u \cos2\alpha_v \sin2\beta_v Z_u Y_v\\
& \quad - \cos2\alpha_u  \cos2\beta_u \sin 2 \alpha_v Z_uX_v + \cos2\alpha_u \sin 2 \beta_u \cos 2\alpha_v \cos 2\beta_v Y_u Z_v\\
& \quad + \cos 2 \alpha_u \sin 2\beta_u \cos 2 \alpha_v \sin 2\beta_v Y_uY_v - \cos2 \alpha_u \sin 2 \beta_u \sin 2 \alpha_v Y_u X_v\\
& \quad - \sin 2 \alpha_u \cos 2 \alpha_v \cos 2 \beta_v X_uZ_v - \sin 2 \alpha_u \cos 2 \alpha_v \sin 2 \beta_v X_u Y_v\\
& \quad + \sin 2\alpha_u \sin 2 \alpha_v X_uX_v.
\end{split}
\label{A3_Expansion}
\end{equation}
Hence, by substituting this expression into \cref{XQAOA_Intro_Exp}, the expected cost function corresponding to the edge $\{u,v\}\in E$ can be written as
\begin{align}
    \langle C_{uv}\rangle_{\mathrm{XY}} = \frac{w_{uv}}2 - \frac{w_{uv}}2 &\Big\{ 
    \cos2\alpha_u  \cos2\beta_u \cos2\alpha_v  \cos2\beta_v \xi(Z,Z) + \cos2\alpha_u  \cos2\beta_u \cos2\alpha_v \sin2\beta_v \xi(Z,Y)
    \nonumber\\
& \quad - \cos2\alpha_u  \cos2\beta_u \sin 2 \alpha_v \xi(Z,X) + \cos2\alpha_u \sin 2 \beta_u \cos 2\alpha_v \cos 2\beta_v \xi(Y,Z)\nonumber\\
& \quad + \cos 2 \alpha_u \sin 2\beta_u \cos 2 \alpha_v \sin 2\beta_v \xi(Y,Y) - \cos2 \alpha_u \sin 2 \beta_u \sin 2 \alpha_v \xi(Y,X)\nonumber\\
& \quad - \sin 2 \alpha_u \cos 2 \alpha_v \cos 2 \beta_v \xi(X,Z) - \sin 2 \alpha_u \cos 2 \alpha_v \sin 2 \beta_v \xi(X,Y)\nonumber\\
& \quad + \sin 2\alpha_u \sin 2 \alpha_v \xi(X,X)\Big\}, 
\label{eq:expected_cost_function_interm}
\end{align}
where for single-qubit Pauli matrices $P,Q \in \{X,Y,Z\}$, we have defined
\begin{align}
    \xi(P,Q) &= \bra s \eta(P,Q) \ket s,
    \label{eq:xi}
    \\
    \eta(P,Q) &= e^{i \bolddot{\gamma} C} P_u Q_v e^{-i \bolddot{\gamma} C}.
    \label{eq:eta}
\end{align}

Moving forward, the approach we take is as follows. Firstly, we shall evaluate the expression for $\eta(P,Q)$ in \cref{eq:eta} for all $P,Q \in \{X,Y,Z\}$. Secondly, we shall 
substitute our expressions for $\eta(P,Q)$ into \cref{eq:xi} to derive analytical expressions for $\xi(P,Q)$. Thirdly and also finally, we shall substitute these analytical expressions into \cref{eq:expected_cost_function_interm} to obtain our desired expression \cref{xqaoa_full_exp}.

Before we execute these three steps, we first introduce some notation to help us keep track of the neighbours of the edge $\{u,v\}$: let
\begin{align}
    \mathcal N_{u\bbackslash v} = \mathcal N(u) \backslash \left(
    \mathcal N(v) \cup \{v\}
    \right)
    =\{\omega_1, \omega_2,\ldots, \omega_b\}
\end{align}
be the set of neighbors of $u$ that are of distance 2 or greater from $v$. Similarly, let 
\begin{align}
    \mathcal N_{v\bbackslash u} = \mathcal N(v) \backslash \left(
    \mathcal N(u) \cup \{u\}
    \right)
    =\{q_1, q_2,\ldots, q_c\}
\end{align}
be the set of neighbors of $v$ that are of distance 2 or greater from $u$. Next, let
\begin{align}
    \mathcal N_{uv} = \mathcal N(u) \cap \mathcal N(v) = \{a_1, a_2,\ldots ,a_f\} = \{\omega_{b+1}, \omega_{b+2},\ldots,\omega_d\} =
    \{q_{c+1}, q_{c+2},\ldots,q_e\},
    \label{eq:Nuv}
\end{align}
where $a_i = \omega_{b+i} = q_{c+i}$ for $i=1,\ldots ,f$,
be the set of vertices that are neighbours of both $u$ and $v$, i.e., $\mathcal N_{uv}$ comprises those nodes in $V$ that form a triangle with both $u$ and $v$. Finally, let
\begin{align}
    \mathcal N_{u\backslash v} =  \mathcal N(u) \backslash \{v\} =  \mathcal N_{u\bbackslash v} \cup \mathcal N_{uv} = \{\omega_1,\ldots, \omega_b,a_1,\ldots, a_f\} = \{\omega_1,\ldots, \omega_b,\ldots, \omega_d\}
    \label{eq:N_u_slash_v}
\end{align}
be the set of neighbours of $u$ that are not $v$ and let 
\begin{align}
    \mathcal N_{v\backslash u} = \mathcal N_{v\bbackslash u} \cup \mathcal N_{uv} = \{q_1,\ldots, q_c,a_1,\ldots, a_f\} = \{q_1, \ldots q_c,\ldots, q_e\}
    \label{eq:N_v_slash_u}
\end{align}
be the set of neighours of $v$ that are not $u$. We are now ready to execute our aforementioned three steps.
\\~\\
\underline{Step 1: Evaluation of $\eta(P,Q)$}
\\~\\
First, we rewrite \cref{eq:eta} as
\begin{align}
    \eta(P,Q) &= e^{i \bolddot{\gamma} {C'}} P_u Q_v e^{-i \bolddot{\gamma} {C'}},
    \label{eq:eta_with_prime}
\end{align}
where
\begin{align}
    \bolddot{\gamma} {C'} = \sum_{\{x,y\}\in E} \gamma_{xy} C'_{xy}, \quad C'_{xy}=
    -\frac 12  w_{xy} Z_x Z_y,
\end{align}
since the term $\frac 12\sum_{\{u,v\}\in E} \gamma_{uv} w_{uv} I$ in $C_{uv}$ is cancelled by its inverse and does not contribute to \cref{eq:eta}. Next, we expand $
    \bolddot{\gamma} {C'} $ as 
\begin{align}
    \bolddot{\gamma} {C'} = \gamma_{uv} C_{uv}' + \bolddot{\gamma} {C_u'} + \bolddot{\gamma} {C_v'} + \sum_{\{x,y\}\in E_{(uv)}} \gamma_{xy} C'_{xy},
    \label{eq:gammaCprime}
\end{align}
where 
\begin{align}
    \bolddot{\gamma} {C_u'} &=
    \sum_{i=1}^d \gamma_{u \omega_i} C_{u \omega_i}' =
    \sum_{i=1}^b \gamma_{u \omega_i} C_{u \omega_i}' +
    \sum_{i=1}^f \gamma_{u a_i} C_{u a_i}'
\end{align}
contains terms involving $u$ but not $v$, and 
\begin{align}
    \bolddot{\gamma} {C_v'} &=
    \sum_{i=1}^e \gamma_{v q_i} C_{v q_i}' =
    \sum_{i=1}^c \gamma_{v q_i} C_{v q_i}' +
    \sum_{i=1}^f \gamma_{v a_i} C_{v a_i}'
\end{align}
contains terms involving $v$ but not $u$, and $E_{(uv)} = \{ \{a,b\}\in E : a,b \notin \{u,v\}\}$ denotes the set of edges that do not contain either $u$ or $v$ as an endpoint. By substituting \cref{eq:gammaCprime} into \cref{eq:eta_with_prime}, we obtain
\begin{align}
    \eta(P,Q) = 
    e^{i \bolddot{\gamma}{C_u'}} e^{i \bolddot{\gamma}{C_v'}}e^{i \gamma_{uv}C'_{uv}} P_u Q_v e^{-i \gamma_{uv}C'_{uv}} e^{-i \bolddot{\gamma}{C_v'}} e^{-i \bolddot{\gamma}{C_u'}}
    = \Gamma_{PQ} P_u Q_v,
    \label{eq:etaPQGamma}
\end{align}
where $\Gamma_{PQ} = \eta(P,Q) P_u Q_v$. By using the fact that Pauli operators either commute or anti-commute, evaluating \cref{eq:etaPQGamma} gives
\begin{align}
    \Gamma_{XX} &= \Gamma_{XY} = \Gamma_{YX} = \Gamma_{YY} = e^{2i \bolddot{\gamma}{C_u'}}
    e^{2i \bolddot{\gamma}{C_v'}},
    \label{eq:gammaXX}
    \\
    \Gamma_{XZ} &= \Gamma_{YZ} = e^{2i \gamma_{uv}C_{uv}'} 
    e^{2i \bolddot{\gamma}{C_u'}} ,
    \label{eq:gammaXZ}
    \\
    \Gamma_{ZX} &= \Gamma_{ZY} = e^{2i \gamma_{uv}C_{uv}'} 
    e^{2i \bolddot{\gamma}{C_v'}},
    \label{eq:gammaZX}
    \\
    \Gamma_{ZZ} &= I.
    \label{eq:gammaZZ}
\end{align}

To evaluate \cref{eq:gammaXX}, we first compute
\begin{align}
    e^{2 i \bolddot{\gamma}{C_u'}} &=
    \prod_{i=1}^d \left[
    \cos(\gamma_{u\omega_i}w_{u \omega_i})I - i \sin (\gamma_{u\omega_i}w_{u \omega_i}) Z_{u} Z_{\omega_i}
    \right]
    \nonumber\\
    &= \sum_{x\in \mathbb F_2^d} \cos^{1-x_1}(\gamma_{u\omega_1}w_{u \omega_1})\cdots
    \cos^{1-x_d}(\gamma_{u\omega_d}w_{u \omega_d}) [-i \sin(\gamma_{u\omega_1}w_{u \omega_1})]^{x_1}\cdots
    [-i \sin(\gamma_{u\omega_d}w_{u \omega_d})]^{x_d} Z_u^{|x|} Z_{\omega_1}^{x_1}\cdots
    Z_{\omega_d}^{x_d},
    \label{eq:gammaXX_u}
\end{align}
where we applied the trick in \cref{eq:product_to_sumss_trick} to each term in the above product. In the above expression, $|x| = \sum_{i=1}^d x_i$ denotes the Hamming weight of the string $x\in \{0,1\}^d$. Similarly,
\begin{align}
    e^{2 i \bolddot{\gamma}{C_v'}} 
    = \sum_{y\in \mathbb F_2^e} \cos^{1-y_1}(\gamma_{vq_1}w_{v q_1})\cdots
    \cos^{1-y_e}(\gamma_{vq_e}w_{v q_e}) [-i \sin(\gamma_{vq_1}w_{v q_1})]^{y_1}\cdots
    [-i \sin(\gamma_{vq_e}w_{v q_e})]^{y_e} Z_v^{|y|} Z_{q_1}^{y_1}\cdots
    Z_{q_e}^{y_e}.
    \label{eq:gammaXX_v}
\end{align}

By taking the product of \cref{eq:gammaXX_u}
and \cref{eq:gammaXX_v}, and using \cref{eq:Nuv} to relabel the vertices in $\mathcal N_{uv}$ by $a_i$'s, we obtain the following expression for \cref{eq:gammaXX}:
\begin{align}
\label{eq:gammaXX_new}
    \Gamma_{XX} &= 
    \sum_{\alpha \in \mathbb F_2^b}
    \sum_{\beta \in \mathbb F_2^c}
    \sum_{\mu \in \mathbb F_2^f}
    \sum_{\nu \in \mathbb F_2^f}
    \cos^{1-\alpha_1}(\gamma_{u\omega_1}w_{u \omega_1})\cdots
    \cos^{1-\alpha_b}(\gamma_{u\omega_b}w_{u \omega_b})
    \cos^{1-\mu_1}(\gamma_{u a_1}w_{u a_1})\cdots
    \cos^{1-\mu_f}(\gamma_{u a_f}w_{u a_f})\nonumber\\
    &\qquad\times
        \cos^{1-\beta_1}(\gamma_{vq_1}w_{v q_1})\cdots
    \cos^{1-\beta_c}(\gamma_{vq_c}w_{v q_c})
    \cos^{1-\nu_1}(\gamma_{v a_1}w_{v a_1})\cdots
    \cos^{1-\nu_f}(\gamma_{v a_f}w_{v a_f})\nonumber\\
    &\qquad\times
    [-i \sin(\gamma_{u\omega_1}w_{u \omega_1})]^{\alpha_1}\cdots
    [-i \sin(\gamma_{u\omega_b}w_{u \omega_b})]^{\alpha_b}
    [-i \sin(\gamma_{u a_1}w_{u a_1})]^{\mu_1}\cdots
    [-i \sin(\gamma_{ua_f}w_{u a_f})]^{\mu_f}\nonumber\\
    &\qquad\times
    [-i \sin(\gamma_{vq_1}w_{v q_1})]^{\beta_1}\cdots
    [-i \sin(\gamma_{vq_c}w_{v q_c})]^{\beta_c}
    [-i \sin(\gamma_{v a_1}w_{v a_1})]^{\nu_1}\cdots
    [-i \sin(\gamma_{va_f}w_{v a_f})]^{\nu_f}\nonumber\\
    &\qquad\times
    Z_u^{|\alpha|+|\mu|} \cdot
    Z_v^{|\beta|+|\nu|} \cdot
    \prod_{i=1}^b Z_{\omega_i}^{\alpha_i}\cdot
    \prod_{i=1}^c Z_{q_i}^{\beta_i}
    \cdot
    \prod_{i=1}^f Z_{a_i}^{\mu_i+\nu_i}.
\end{align}

Next, by substituting \cref{eq:gammaXX_u} into \cref{eq:gammaXZ}, we obtain
\begin{align}
\label{eq:gammaXZ_new}
    \Gamma_{XZ} &=
    \sum_{x_0 x_1 \ldots x_d \in \mathbb F_2^{d+1}}
    \cos^{1-x_0}(\gamma_{uv}w_{uv})
    \cos^{1-x_1}(\gamma_{u\omega_1}w_{u \omega_1})\cdots
    \cos^{1-x_d}(\gamma_{u\omega_d}w_{u \omega_d}) \nonumber\\
    &\qquad\times [-i \sin(\gamma_{uv}w_{u v})]^{x_0}[-i \sin(\gamma_{u\omega_1}w_{u \omega_1})]^{x_1}\cdots
    [-i \sin(\gamma_{u\omega_d}w_{u \omega_d})]^{x_d} Z_u^{|x|} Z_v^{x_0} Z_{\omega_1}^{x_1}\cdots
    Z_{\omega_d}^{x_d}.
\end{align}
Similarly, by substituting \cref{eq:gammaXX_v} into \cref{eq:gammaZX}, we obtain
\begin{align}
\label{eq:gammaZX_new}
    \Gamma_{ZX} &=
    \sum_{y_0 y_1 \ldots y_{{e}} \in \mathbb F_2^{e+1}}
    \cos^{1-y_0}(\gamma_{uv}w_{uv})
    \cos^{1-y_1}(\gamma_{vq_1}w_{v q_1})\cdots
    \cos^{1-y_e}(\gamma_{vq_e}w_{v q_e}) \nonumber\\
    &\qquad\times [-i \sin(\gamma_{uv}w_{u v})]^{y_0}[-i \sin(\gamma_{vq_1}w_{v q_1})]^{y_1}\cdots
    [-i \sin(\gamma_{vq_e}w_{v q_e})]^{y_e} Z_u^{y_0} Z_v^{|y|} Z_{q_1}^{y_1}\cdots
    Z_{q_{{e}}}^{y_{{e}}}.
\end{align}

This completes Step 1, where $\eta(P,Q)$ is specified by \cref{eq:etaPQGamma}, \cref{eq:gammaXX_new}, \cref{eq:gammaXZ_new}, \cref{eq:gammaZX_new}, and
\cref{eq:gammaZZ}.
\\~\\
\underline{Step 2: Evaluation of $\xi(P,Q)$}
\\~\\
There are $3^2=9$ different choices of $PQ \in \{X,Y,Z\}^2$. We will split these choices into four cases, as follows:
\begin{itemize}
    \item[\small $\bullet$]
\underline{Case 1}: $PQ = XX, XY, YX, YY$. By  substituting \cref{eq:etaPQGamma} and \cref{eq:gammaXX_new} into \cref{eq:xi}, we obtain
\begin{align}
    \xi(P,Q) &= \sum_{\alpha \in \mathbb F_2^b}
    \sum_{\beta \in \mathbb F_2^c}
    \sum_{\mu \in \mathbb F_2^f}
    \sum_{\nu \in \mathbb F_2^f}
    \cos^{1-\alpha_1}(\gamma_{u\omega_1}w_{u \omega_1})\cdots
    \cos^{1-\alpha_b}(\gamma_{u\omega_b}w_{u \omega_b})
    \cos^{1-\mu_1}(\gamma_{u a_1}w_{u a_1})\cdots
    \cos^{1-\mu_f}(\gamma_{u a_f}w_{u a_f})\nonumber\\
    &\qquad\times
        \cos^{1-\beta_1}(\gamma_{vq_1}w_{v q_1})\cdots
    \cos^{1-\beta_c}(\gamma_{vq_c}w_{v q_c})
    \cos^{1-\nu_1}(\gamma_{v a_1}w_{v a_1})\cdots
    \cos^{1-v_f}(\gamma_{v a_f}w_{v a_f})\nonumber\\
    &\qquad\times
    [-i \sin(\gamma_{u\omega_1}w_{u \omega_1})]^{\alpha_1}\cdots
    [-i \sin(\gamma_{u\omega_b}w_{u \omega_b})]^{\alpha_b}
    [-i \sin(\gamma_{u a_1}w_{u a_1})]^{\mu_1}\cdots
    [-i \sin(\gamma_{ua_f}w_{u a_f})]^{\mu_f}\nonumber\\
    &\qquad\times
    [-i \sin(\gamma_{vq_1}w_{v q_1})]^{\beta_1}\cdots
    [-i \sin(\gamma_{vq_c}w_{v q_c})]^{\beta_c}
    [-i \sin(\gamma_{v a_1}w_{v a_1})]^{\nu_1}\cdots
    [-i \sin(\gamma_{va_f}w_{v a_f})]^{\nu_f}\nonumber\\
    &\qquad\times
    \underbrace{
    \tr \left\{ \ketbra{s}{s}  Z_u^{|\alpha|+|\mu|} P_u \cdot
    Z_v^{|\beta|+|\nu|} Q_v \cdot
    \prod_{i=1}^b Z_{\omega_i}^{\alpha_i}\cdot
    \prod_{i=1}^c Z_{q_i}^{\beta_i}
    \cdot
    \prod_{i=1}^f Z_{a_i}^{\mu_i+\nu_i}
    \right\}
    }_{\circled{1}}.
    \label{eq:xi_PQ_1}
\end{align}
By writing the initial state as $\ketbra ss = \bigotimes_{j=1}^n \frac{1}{2}(I+X_j)$, the last line of \cref{eq:xi_PQ_1} can be expanded as
\begin{align}
\label{eq:circled_1}
    \circled{1} &= \tr\left\{\frac{I+X_u}{2} Z_u^{|\alpha|+|\mu|} P_u\right\}
    \tr\left\{\frac{I+X_v}{2} Z_v^{|\beta|+|\nu|} Q_v\right\} \nonumber\\
    &\qquad\times
    \prod_{i=1}^b
    \underbrace{\tr\left\{\frac{I+X_{\omega_i}}{2} Z_{\omega_i}^{\alpha_i} \right\}}_{= \, \delta_{\alpha_i,0}}
    \prod_{i=1}^c
    \underbrace{\tr\left\{\frac{I+X_{q_i}}{2} Z_{q_i}^{\beta_i} \right\}}_{= \, \delta_{\beta_i,0}}
    \prod_{i=1}^f
    \underbrace{\tr\left\{\frac{I+X_{a_i}}{2} Z_{a_i}^{\mu_i+\nu_i} \right\}}_{= \, \delta_{\mu_i, \nu_i}}
    \nonumber\\
    &= \underbrace{\tr\left\{\frac{I+X_{{u}}}{2} Z_{{u}}^{|\mu|} P_{{u}} \right\}
    }_{\circled{2}}
    \underbrace{\tr\left\{\frac{I+X_{{v}}}{2} Z_{{v}}^{|\nu|} Q_{{v}}\right\}
    }_{\circled{3}}
    \prod_{i=1}^b \delta_{\alpha_i,0}
    \prod_{i=1}^c \delta_{\beta_i,0}
    \prod_{i=1}^f \delta_{\mu_i,\nu_i}.
\end{align}
Now, for $P \in \{X,Y\}$ and for $|\mu| \in \mathbb N$, 
\begin{align}
    \circled{2} = [P=X][|\mu| \mbox{ even}] - i
    [P=Y][|\mu| \mbox{ odd}],
    \label{eq:circled_2}
\end{align}
where the $[A]$ denotes the Iverson bracket of statement $A$; i.e.~$[A]=1$ if $A$ is true and 0 otherwise. Similarly, 
\begin{align}
    \circled{3} = [Q=X][|\mu| \mbox{ even}] - i
    [Q=Y][|\mu| \mbox{ odd}].
    \label{eq:circled_3}
\end{align}
Hence, the product of \cref{eq:circled_2} and \cref{eq:circled_3} is
\begin{align}
    \circled{2}\times
    \circled{3}
    =
    [P=Q=X][|\mu| \mbox{ even}] - [P=Q=Y][|\mu| \mbox{ odd}].
    \label{eq:circled_2_3}
\end{align}
Substituting \cref{eq:circled_2_3} into \cref{eq:circled_1} gives
\begin{align}
    \circled{1} = \delta_{\alpha,0} \delta_{\beta,0}
    \delta_{\mu,\nu} ([P=Q=X][|\mu| \mbox{ even}] - [P=Q=Y][|\mu| \mbox{ odd}]).
    \label{eq:circled_1_a}
\end{align}
Substituting \cref{eq:circled_1_a} into \cref{eq:xi_PQ_1} and using the Kronecker deltas to eliminate terms in the sum, we get
\begin{align}
    \xi(P,Q) &= \sum_{\mu \in \mathbb F_2^f}
    \cos(\gamma_{u\omega_1}w_{u \omega_1})\cdots
    \cos(\gamma_{u\omega_b}w_{u \omega_b})
    \cos^{1-\mu_1}(\gamma_{u a_1}w_{u a_1})\cdots
    \cos^{1-\mu_f}(\gamma_{u a_f}w_{u a_f})\nonumber\\
    &\qquad\times
        \cos(\gamma_{vq_1}w_{v q_1})\cdots
    \cos(\gamma_{vq_c}w_{v q_c})
    \cos^{1-\mu_1}(\gamma_{v a_1}w_{v a_1})\cdots
    \cos^{1-\mu_f}(\gamma_{v a_f}w_{v a_f})
    \nonumber\\
    &\qquad\times
    (-i)^{2|\mu|} \sin^{\mu_1}(\gamma_{ua_1}w_{u a_1})\cdots
    \sin^{\mu_f}(\gamma_{ua_f}w_{u a_f})
    \sin^{\mu_1}(\gamma_{va_1}w_{v a_1})\cdots
    \sin^{\mu_f}(\gamma_{va_f}w_{v a_f})
    \nonumber\\
    &\qquad\times 
    \left([P=Q=X][|\mu| \mbox{ even}] - [P=Q=Y][|\mu| \mbox{ odd}]\right)
    \nonumber\\
    &=
    \prod_{\omega\in \mathcal N_{u\bbackslash v}} \cos (\gamma_{u\omega} w_{u\omega})    \prod_{q\in \mathcal N_{v\bbackslash u}} \cos (\gamma_{vq} w_{vq}) 
    \Bigg(
    [P=Q=X] \sum_{
    \substack{\mu \in \mathbb F_2^f \\ |\mu| \text{ even}}}
    +
    [P=Q=Y] \sum_{
    \substack{\mu \in \mathbb F_2^f \\ |\mu| \text{ odd}}}
    \Bigg)
    \nonumber\\
    &\qquad
    \prod_{i=1}^f \cos^{1-\mu_i}(\gamma_{ua_i}w_{ua_i}) \cos^{1-\mu_i}(\gamma_{va_i}w_{va_i})
    \sin^{\mu_i}(\gamma_{ua_i}w_{ua_i}) \sin^{\mu_i}(\gamma_{va_i}w_{va_i}).    
    \label{eq:xi_PQ_2}
\end{align}
Since the Iverson brackets in \cref{eq:xi_PQ_2} vanish when $PQ = XY$ or $YX$, it follows that 
\begin{align}
    \xi(X,Y) = \xi(Y,X) = 0.
    \label{eq:xi_XY_final}
\end{align}
When $PQ=XX$, \cref{eq:xi_PQ_2} reduces to
\begin{align}
    \xi(X,X) &=
    \prod_{\omega\in \mathcal N_{u\bbackslash v}} \cos (\gamma_{u\omega} w_{u\omega})    \prod_{q\in \mathcal N_{v\bbackslash u}} \cos (\gamma_{vq} w_{vq}) 
    \nonumber\\
    &\qquad
    \times \sum_{
    \substack{\mu \in \mathbb F_2^f \\ |\mu| \text{ even}}}\prod_{i=1}^f \cos^{1-\mu_i}(\gamma_{ua_i}w_{ua_i}) \cos^{1-\mu_i}(\gamma_{va_i}w_{va_i})
    \sin^{\mu_i}(\gamma_{ua_i}w_{ua_i}) \sin^{\mu_i}(\gamma_{va_i}w_{va_i})    
    \nonumber\\
    &=
    \frac 12 \prod_{\omega\in \mathcal N_{u\bbackslash v}} \cos (\gamma_{u\omega} w_{u\omega})    \prod_{q\in \mathcal N_{v\bbackslash u}} \cos (\gamma_{vq} w_{vq}) \nonumber\\
    &\qquad \times 
    \left[ \prod_{i=1}^f \cos(\gamma_{u a_i} w_{ua_i}+\gamma_{va_i} w_{va_i})+
    \prod_{i=1}^f \cos(\gamma_{u a_i} w_{ua_i}-\gamma_{va_i} w_{va_i})
    \right]
    \nonumber\\
    &=
    \frac 12 \prod_{\omega\in \mathcal N_{u\bbackslash v}} \cos (\gamma_{u\omega} w_{u\omega})    \prod_{q\in \mathcal N_{v\bbackslash u}} \cos (\gamma_{vq} w_{vq}) \nonumber\\
    &\qquad \times 
    \left[ \prod_{q\in \mathcal N_{uv}} \cos(\gamma_{u q} w_{uq}-\gamma_{vq} w_{vq})+
    \prod_{q\in \mathcal N_{uv}} \cos(\gamma_{uq} w_{uq}+\gamma_{vq} w_{vq})
    \right],
    \label{eq:xi_XX_final}
\end{align}
where the second equality follows from \cref{eq:cosine_identity_even}. 

Similarly, for the case $PQ=YY$, utilising \cref{eq:cosine_identity_odd} gives
\begin{align}
    \xi(Y,Y)&=
    \frac 12 \prod_{\omega\in \mathcal N_{u\bbackslash v}} \cos (\gamma_{u\omega} w_{u\omega})    \prod_{q\in \mathcal N_{v\bbackslash u}} \cos (\gamma_{vq} w_{vq}) \nonumber\\
    &\qquad \times 
    \left[ \prod_{q\in \mathcal N_{uv}} \cos(\gamma_{u {q}} w_{u{q}}-\gamma_{v{q}} w_{v{q}})-
    \prod_{q\in \mathcal N_{uv}} \cos(\gamma_{u {q}} w_{u {q}}+\gamma_{v {q}} w_{v {q}})
    \right],
    \label{eq:xi_YY_final}
\end{align}
which differs from the expression given by \cref{eq:xi_XX_final} for $\xi(X,X)$ by only a single minus sign. As we mentioned in \cref{sec:useful_identities}, the identities \cref{eq:cosine_identity_even} and \cref{eq:cosine_identity_odd} allowed us to simplify the exponential sums in \cref{eq:xi_PQ_2} to get exponential savings in the computational cost in the worst case. 
\item[\small $\bullet$]
\underline{Case 2}: $PQ=XZ, YZ$. By substituting \cref{eq:etaPQGamma} and \cref{eq:gammaXZ_new} into \cref{eq:xi}, we obtain, for $P=X,Y$,
\begin{align}
    \xi(P,Z) &= \sum_{x_0 x_1 \ldots x_d \in \mathbb F_2^{d+1}}
    \cos^{1-x_0}(\gamma_{uv}w_{uv})
    \cos^{1-x_1}(\gamma_{u\omega_1}w_{u \omega_1})\cdots
    \cos^{1-x_d}(\gamma_{u\omega_d}w_{u \omega_d}) \nonumber\\
    &\qquad\times [-i \sin(\gamma_{uv}w_{u v})]^{x_0}[-i \sin(\gamma_{u\omega_1}w_{u \omega_1})]^{x_1}\cdots
    [-i \sin(\gamma_{u\omega_d}w_{u \omega_d})]^{x_d}
    \nonumber\\
    &\qquad\times
    \underbrace{\tr\left\{\ketbra{s}{s} Z_u^{|x|}P_u Z_v^{x_0+1} Z_{\omega_1}^{x_1}\cdots
    Z_{\omega_d}^{x_d}
    \right\}
    }_{\circled{4}}.
    \label{eq:xi_PQ_3}
\end{align}
By writing $\ketbra ss = \bigotimes_{j=1}^n \frac{1}{2}(I+X_j)$, the last line of {\cref{eq:xi_PQ_3}} can be expanded as
\begin{align}
\label{eq:circled_4}
    \circled{4} &= \tr\left\{\frac{I+X_u}{2} Z_u^{x_0+x_1+\cdots+x_d} P_u\right\}
    \tr\left\{\frac{I+X_v}{2} Z_v^{x_0+1} \right\}
    \prod_{i=1}^d
    \underbrace{\tr\left\{\frac{I+X_{\omega_i}}{2} Z_{\omega_i}^{x_i} \right\}}_{= \, \delta_{x_i,0}}
    \nonumber\\
    &= \tr\left\{ \frac{I+X{_u}}{2} Z{_u} P{_u}\right\} \delta_{x_0,1} \delta_{x_1,0}\ldots\delta_{x_d,0}
    \nonumber\\
    &= -i [P=Y] \delta_{x_0,1} \delta_{x_1,0}\ldots\delta_{x_d,0}
    .
\end{align}
Hence, by substituting \cref{eq:circled_4} into \cref{eq:xi_PQ_3}, we obtain
\begin{align}
    \xi(X,Z) &= 0,
    \label{eq:xi_XZ_final}
    \\
    \xi(Y,Z) &= -\sin(\gamma_{uv}w_{uv}) \prod_{\omega \in \mathcal N_{u\backslash v}}\cos(\gamma_{u\omega}w_{u\omega}). 
    \label{eq:xi_YZ_final}
\end{align}
\item \underline{Case 3}: $PQ=ZX, ZY$. This case is identical to Case 2, but with the vertices $u$ and $v$ swapped. Hence, from \cref{eq:xi_XZ_final} and \cref{eq:xi_YZ_final}, we deduce that
\begin{align}
    \xi(Z,X) &= 0, \label{eq:xi_ZX_final}
    \\
    \xi(Z,Y) &= -\sin(\gamma_{uv}w_{uv}) \prod_{q \in \mathcal N_{v\backslash u}}\cos(\gamma_{vq}w_{vq}). \label{eq:xi_ZY_final}
\end{align}
\item \underline{Case 4}: $PQ=ZZ$. This case is straightforward. One readily computes that $\xi(Z,Z)$ vanishes:
\begin{align}
    \xi(Z,Z) &= \tr\left(\bigotimes_{j=1}^n \frac 12(I+X_j) Z_u Z_v\right) \nonumber\\
    &=
    \prod_{a\neq u,v} \tr\left(\frac {I+X_a}2\right)\cdot \tr\left(\frac {I+X_u}2 Z_u\right)
    \cdot \underbrace{\tr\left(\frac {I+X_v}2 Z_v\right)}_{=0} 
    \nonumber\\
    &=0.
\end{align}
\end{itemize}
From the above calculations, we see that for five of the nine choices of $P$ and $Q$, $\xi(P,Q)$ vanishes: $\xi(Z,Z)=\xi(Z,X)=\xi(Y,X)=\xi(X,Z)=\xi(X,Y)=0$. The remaining four non-vanishing ones are given by
\cref{eq:xi_XX_final}, \cref{eq:xi_YY_final}, \cref{eq:xi_YZ_final}, and \cref{eq:xi_ZY_final}.
\\~\\
\underline{Step 3: Derivation of \cref{xqaoa_full_exp}}
\\~\\
We are now essentially done. By substituting the expressions for $\xi(P,Q)$ that we obtained in Step 2 into \cref{eq:expected_cost_function_interm}, and noting that $\mathcal N_{v\backslash u}=e$, $\mathcal N_{u \backslash v}{=d}$, $\mathcal N_{u\bbackslash v} = d\backslash F$, $\mathcal N_{v\bbackslash u} = e\backslash F$, and $N_{uv} = F$, we obtain \cref{xqaoa_full_exp}.
\end{proof}

\subsection{Proofs of \texorpdfstring{Theorems~\ref{ma_qaoa_theorem} and~\ref{qaoa_hadfield1}}{Theorems 1 and 2}} \label{MA_QAOA_Proof}

The analytical expressions for the expected cost function of both MA-QAOA and QAOA can be derived from XQAOA's analytical formula. Setting all $\alpha_i$ to 0 in \cref{xqaoa_full_exp} gives \cref{ma_qaoa_full_exp} in \cref{ma_qaoa_theorem}. Moreover, setting all $\beta_i$ to $\beta$ and $\gamma_i$ to $\gamma$ in \cref{ma_qaoa_full_exp} and simplifying gives \cref{qaoa_formula} in \cref{qaoa_hadfield1}.

\subsection{Proof of corollary~\ref{XQAOA_Exact_Thm}} \label{xqaoa_proof2}

\begin{proof}[Proof of corollary~\ref{XQAOA_Exact_Thm}]

By setting $\beta_k = 0$ for all $k$ in \cref{xqaoa_full_exp},
we obtain the following expected cost function (corresponding to the edge $\{u,v\}$) for $\XQAOA{Y}$ :
\begin{align}
\langle C_{uv}\rangle_{\mathrm{Y}} =
\frac {w_{uv}}2 - \frac {w_{uv}}4 \sin{2\alpha_u}\sin{2\alpha_v} \prod_{\substack{w \in e \\ w \notin F}}\cos\gamma_{wv}' \prod_{\substack{w \in d \\ w \notin F}}\cos \gamma_{uw}' \left( \prod_{f \in F} \cos(\gamma_{uf}' + \gamma_{vf}') + \prod_{f \in F} \cos(\gamma_{uf}' - \gamma_{vf}') \right).
\label{eq:XQAOAY}
\end{align}
To prove corollary~\ref{XQAOA_Exact_Thm}, we shall specialise \cref{eq:XQAOAY} to unweighted graphs with odd edge degrees. First, we show that if every edge degree of a graph $G=(V,E)$ is odd, then the graph is necessarily two-colourable, i.e., there exists a map $\tau:V\rightarrow \{0,1\}$ such that for all edges $\{u,v\}\in E$, $\tau(u)\neq \tau(v)$. Indeed, an instantiation of such a map $\tau$ is given by:
\begin{align}
    \tau(v) = \begin{cases}
        0, & \deg(v) \text{ even,} \\
        1, & \deg(v) \text{ odd.}
    \end{cases}
\end{align}

It now remains for us to prove that $\tau$ is indeed a two-colouring: let $\{u,v\}\in E$. Note that the edge degree of an edge $\{i,j\}\in E$ is $\deg(\{i,j\}) = |\mathcal{N}(i)\cup \mathcal{N}(j)|-2 = \deg(i)+\deg(j)-2$. By our assumption, the edge degree $\deg(\{u,v\})$ is odd, and hence $\deg(u) + \deg(v)$ is also odd. This implies that exactly one of $\deg(u)$ and $\deg(v)$ is odd, from which it follows that exactly one of $\tau(u)$ and $\tau(v)$ is equal to 1. Hence, $\tau(u)$ and $\tau(v)$ cannot be equal to each other, i.e., $\tau(u) \neq \tau(v)$. This completes our proof that $G$ is two-colourable.

Now, a graph is 2-colourable if and only if it does not contain an odd cycle 
(see \cite[Theorem 5.21]{keller2017applied}, for example). Hence, the graph $G$ considered in corollary~\ref{XQAOA_Exact_Thm} with all edge degrees being odd must be triangle-free, i.e. $G$ does not contain a 3-cycle. This implies that the set $F$ in \cref{eq:XQAOAY} must be empty.

Therefore, setting $F= \emptyset$, $w_{uv}=1$, and $\gamma_{ij}' = \gamma_{ij}$ for all edges $\{i,j\}\in E$ gives the expected $\XQAOA{Y}$ cost function (for the edge $\{u,v\}$) for unweighted graphs with all edge degrees being odd:
\begin{align}
\langle C_{uv}\rangle_{\mathrm{Y}} =
\frac 12 - \frac 12 \sin{2\alpha_u}\sin{2\alpha_v} \prod_{w \in e}\cos\gamma_{wv} \prod_{w \in d}\cos \gamma_{uw}.
\label{eq:XQAOAY_edge_odd}
\end{align}
Note that if all the mixer unitary angles $\alpha_i = \alpha$ are equal and all the problem unitary angles $\gamma_{jk} = \gamma$ are equal, then \cref{eq:XQAOAY_edge_odd} simplifies to
\begin{align}
    \left\langle C_{u v}\right\rangle_{\mathrm{Y}} =\frac{1}{2}- \frac 12 \sin^2 2\alpha  \cos^{|e|+|d|} \gamma.    \label{eq:XQAOAY_edge_odd_equal_angles} 
\end{align}
In \cref{eq:XQAOAY_edge_odd_equal_angles}, if one takes $\alpha=\tfrac \pi 4$ and $\gamma = \pi$, one obtains
\begin{align}
    \left\langle C_{u v}\right\rangle_{\mathrm{Y}} =\frac 12-\frac 12 (-1)^{|e|+|d|}.
\end{align}
By assumption, $G$ has only odd edge degrees. Therefore, $|e|+|d|$ is odd, which implies that
\begin{align}
    \left\langle C_{u v}\right\rangle_{\mathrm{Y}} =\frac 12+\frac 12 = 1.
\end{align}
So, the optimal expected cost function is given by
\begin{align}
    \sum_{\{u,v\}\in E} \left\langle C_{u v}\right\rangle_{\mathrm{Y}} = |E|,
\end{align}
which coincides with the maximum cut size of 2-colourable graphs (the MaxCut of 2-colourable graphs is $|E|$ because one could just choose the maximum cut to be the 2-colouring).

Therefore, the approximation ratio achieved is 1, i.e., the $\XQAOA{Y}$ state $\ket{\gamma, \alpha}$ provides the exact MaxCut solution for $G$.
\end{proof}

\subsection{Proof of \texorpdfstring{corollary~\ref{cor:proof_separation}}{corollary 5}} 
\label{app:proof_separation}

In this appendix, we give a proof of corollary~\ref{cor:proof_separation}, which quantifies the advantage that XQAOA has over QAOA for an unweighted $(k+1)$-vertex star graph $G = S_k$, illustrated in \cref{fig:4-star}. To keep our analysis general, we will take $k$ to be an arbitrary positive integer for now and only later specialise to $k=4$. 

Note that when $k$ is even, all the edge degrees in $S_k$ are odd; hence,
corollary~\ref{XQAOA_Exact_Thm} implies that $\XQAOA{Y}$ with optimal angles
computes the maximum cut of $G$ with an approximation ratio of 1 whenever $k$ is even. Hence, to complete the proof of corollary~\ref{cor:proof_separation}, it remains to show that QAOA$_1$ can do no better than achieve an approximation ratio of $3/4$.

To this end, consider the cost function for
QAOA$_1$ given by \cref{eq:qaoa_unweighted}. For the star graph $S_k$, $|e|=|F|=0$ and $|d|=k-1$. Hence, \cref{eq:qaoa_unweighted} reduces to 
$\langle C_{uv}\rangle = \frac 12 + \frac 14 \sin 4\beta \sin\gamma(1+\cos^{k-1} \gamma)$, which is independent of the edge $\{u,v\}$. Therefore, the expected value of $C$ can be written as
\begin{align}
    \langle C \rangle = \frac k2 + \frac k4 \sin 4\beta \sin\gamma(1+\cos^{k-1} \gamma).
    \label{eq:expC_QAOA1}
\end{align}
Maximising \cref{eq:expC_QAOA1} over all $\beta$'s and $\gamma$'s gives 
\begin{align}
    \langle C\rangle_{\max} &= 
    \frac k2 + \frac k4 \max_{\gamma,\beta\in \mathbb R} \big[\sin 4\beta \sin\gamma(1+\cos^{k-1} \gamma)\big] \nonumber\\
    &= 
    \frac k2 + \frac k4 \max_{\beta\in \mathbb R} \big[\sin 4\beta\big] \max_{\gamma\in \mathbb R} \big[\sin\gamma(1+\cos^{k-1} \gamma)\big]
    \nonumber\\
    &=
    \frac k2 + \frac k4 \max_{\gamma \in \mathbb R} g_k(\gamma),
    \label{eq:C_max}
    \\
    \mbox{where}\quad 
    g_k(\gamma) &= \sin \gamma(1+\cos^{k-1}\gamma).
    \label{eq:g_k}
\end{align}
To find the maximum point(s) of $g_k$, we start by taking its derivative. For $k\geq 2$, differentiating \cref{eq:g_k} and simplifying gives
\begin{align}
    g_k'(\gamma) = k \cos^k \gamma - (k-1)\cos^{k-2}\gamma + \cos\gamma,
    \label{eq:gk_prime}
\end{align}
which is a degree-$k$ polynomial in $\cos \gamma$. Hence, finding the maximum point(s) of $g_k(\gamma)$ involves finding the roots of this polynomial. At this point, we specialise to $k=4$. This allows to factorise the quartic polynomial \cref{eq:gk_prime} as
\begin{align}
    g_4'(\gamma) &= 4 \cos^4 \gamma - 3\cos^2\gamma + \cos\gamma
    =
    \cos \gamma (\cos\gamma+1)(2\cos\gamma-1)^2.
\end{align}
Hence, at a maximum point $\gamma^* \in \arg\max_{\gamma \in \mathbb R} g_k(\gamma)$, solving $g_4'(\gamma^*)=0$ gives $\cos\gamma^* \in \{0,-1,\tfrac 12\}$. When $\cos\gamma^*=0$, $g_4(\gamma^*)=\pm 1$; when $\cos\gamma^*=-1$, $g_4(\gamma^*)=0$; and when $\cos\gamma^*=\tfrac 12$, $g_4(\gamma^*)=\pm \tfrac 9{16}\sqrt{3} \approx \pm 0.974 < 1$. Taking the maximum of all these values gives $\max_{\gamma \in \mathbb R} g_k(\gamma) = 1$. Substituting this into \cref{eq:C_max} gives an approximation ratio of
\begin{align}
    \tfrac 1k \langle C\rangle_{\max}= \frac 12 + \frac 14 \cdot 1 = \frac 34.
\end{align}
In conclusion, for the unweighted 5-vertex star graph $S_4$, $\XQAOA{Y}$ achieves an approximation ratio of 1, whereas QAOA$_1$ achieves an approximation ratio of at most $\frac 34$.

\end{document}